\documentclass[10pt,reqno]{article}
\usepackage{fullpage} 
\usepackage{amssymb} 
\usepackage{graphicx} 
\usepackage{amsmath} 
\usepackage{amsthm,amssymb,latexsym} 
\usepackage{amstext, amsfonts,amsbsy} 
\usepackage{enumerate} 
\usepackage{upgreek} 
\usepackage{float} 
\usepackage{color} 
\usepackage{accents} 
\usepackage{mathtools} 
\usepackage{tikz}
\usetikzlibrary{matrix,graphs,arrows,positioning,calc,decorations.markings,shapes.symbols,shapes.geometric,shapes.misc}
\usepackage{calc, ifthen, ae, xstring, etoolbox, xkeyval,pgfplots}
\usepackage[pdftex,bookmarks,colorlinks,breaklinks]{hyperref}  
\definecolor{dullmagenta}{rgb}{0.4,0,0.4}   
\definecolor{darkblue}{rgb}{0,0,0.4}
\hypersetup{linkcolor=red,citecolor=blue,filecolor=dullmagenta,urlcolor=darkblue} 
\usepackage{eucal}

\newtheorem{theorem}{Theorem}

\newtheorem{lemma}[theorem]{Lemma}

\newcommand{\Pain}[1]{\text{P}_{\mathrm{#1}}}

\newcommand{\Ham}[2]{\mathrm{H}^{\mathrm{#1}}_{\mathrm{#2}}}
\newcommand{\Symp}[1]{\Omega^{\mathrm{#1}}}
\newcommand{\symp}[1]{\omega^{\mathrm{#1}}}
\newcommand{\sympt}[1]{\tilde{\omega}^{\mathrm{#1}}}

\theoremstyle{definition}

\theoremstyle{remark}
\newtheorem{remark}{Remark}

\newtheorem*{notation*}{Notation} 

\numberwithin{equation}{section}

\begin{document}

{\noindent\Large\bf Different Hamiltonians for differential Painlev\'e equations and their\\ identification 
using a geometric approach 
}
\medskip

\begin{flushleft}

\textbf{Anton Dzhamay}\\
School of Mathematical Sciences, The University of Northern Colorado, Greeley, CO 80639, US and \\
Yanqi Lake Beijing Institute of Mathematical Sciences and Applications (BIMSA), Beijing, China\\
E-mail: \href{mailto:anton.dzhamay@unco.edu}{\texttt{anton.dzhamay@unco.edu}}\qquad 
ORCID ID: \href{https://orcid.org/0000-0001-8400-2406}{\texttt{0000-0001-8400-2406}}\\[5pt]

\textbf{Galina Filipuk}\\
Institute of Mathematics, University of Warsaw, Banacha 2, Warsaw, 02-097, Poland\\
E-mail: \href{mailto:g.filipuk@uw.edu.pl}{\texttt{g.filipuk@uw.edu.pl}}\qquad
ORCID ID: \href{https://orcid.org/0000-0003-2623-5361}{\texttt{0000-0003-2623-5361}}\\[5pt] 

\textbf{Adam Lig\c{e}za}\\
Institute of Mathematics, University of Warsaw, Banacha 2, Warsaw, 02-097, Poland\\
E-mail: \href{mailto:al320166@students.mimuw.edu.pl}{\texttt{al320166@students.mimuw.edu.pl}}\qquad
ORCID ID: \href{https://orcid.org/0000-0002-1340-3840}{\texttt{0000-0002-1340-3840}}\\[5pt] 

\textbf{Alexander Stokes}\\
Graduate School of Mathematical Sciences, The University of Tokyo, 3--8--1 Komaba, Meguro-ku, \\Tokyo 153-8914, Japan\\
E-mail: \href{mailto:stokes@ms.u-tokyo.ac.jp}{\texttt{stokes@ms.u-tokyo.ac.jp}}\qquad
ORCID ID: \href{https://orcid.org/0000-0001-6874-7141}{\texttt{0000-0001-6874-7141}}\\[8pt]



\end{flushleft}

%

%
\date{}

\begin{abstract}
	It is well-known that differential Painlev\'e equations can be written in a Hamiltonian form. However, a coordinate
	form of such representation	is far from unique -- there are many very different Hamiltonians that result in the same
	differential Painlev\'e equation. Recognizing a Painlev\'e equation, for example when it appears in some applied problem, 
	is known as the \emph{Painlev\'e equivalence problem}, and the question that we consider here is the Hamiltonian form of this
	problem. Making such identification explicit, on the level of coordinate transformations, can be very helpful
	for an applied problem, since it gives access to the wealth of known results about Painlev\'e equations, such 
	as the structure of the symmetry group, special solutions for special values of the parameters, and so on. It can also
	provide an explicit link between different problems that have the same underlying structure.  
	In this paper we describe a systematic procedure for finding changes of coordinates
	that transform different Hamiltonian representations of a Painlev\'e equation into some chosen canonical form. 
	Our approach is based on Sakai's geometric theory of Painlev\'e equations. We explain this procedure in detail for 
	the fourth differential $\Pain{IV}$ equation, and also give a brief summary of some known examples for 
	$\Pain{V}$ and $\Pain{VI}$ cases. It is clear that this approach can easily be adapted to other examples as well, 
	so we expect our paper to be a useful reference for some of the realizations of Okamoto spaces of initial conditions 
	for Painlev\'e equations.
\end{abstract}


\section{Introduction} 
\label{sec:intro}

Painlev\'e equations play an increasingly important role in a wide range of nonlinear problems in Mathematics and Mathematical
Physics \cite{IwaKimShiYos:1991:FGP,Con:1999:PP}. It is well-known that Painlev\'e equations can be written in a Hamiltonian
form. However, there are many different Hamiltonian systems that reduce to the same differential Painlev\'e equation, and the 
relationship between such different systems are far from obvious. The need for such an identification comes from the fact that 
both discrete and differential Painlev\'e equations, as well as related Hamiltonian systems, appear in many interesting
applications, but often, when expressed in the coordinates of the problem, they are difficult to recognize. 
This is sometimes called the \emph{Painlev\'e equivalence problem} \cite{Cla:2019:OPPE}.
Knowing that a 
certain system is just a well-studied Painlev\'e equation in disguise, and explicitly identifying parameters of the problem 
with the standard ones, gives access to a large collection of known facts, such as B\"acklund transformations, special solutions
for certain parameter values, and so on. Often such connections are made on an ad-hoc basis that requires some luck and 
guesswork. The main message of our paper is that the geometric approach makes any guesswork unnecessary and instead gives us
a step-by-step procedure to perform such an identification, when possible. Thus, our objective is both to explain an algorithmic 
scheme for identifying different Hamiltonian systems related to the same Painlev\'e equation via an explicit 
birational change of variables, and to provide the necessary amount of the background material to make this approach accessible. 
This work is motivated by the recently proposed identification scheme for 
discrete Painlev\'e equations \cite{DzhFilSto:2020:RCDOPWHWDPE} and is based on the geometric theory of Painlev\'e equations
initiated in the works of K.~Okamoto \cite{Oka:1979:FAESOPCFPP} and further developed in the seminal paper of H.~Sakai 
\cite{Sak:2001:RSAWARSGPE}. In this paper we consider examples of different Hamiltonians for Painlev\'e equations that 
are given in the foundational papers by K.~Okamoto \cite{Oka:1980:PHAWPE} (and much earlier in essentially the same form
by Malmquist
\cite{Mal:1922:EDSODLGPCF}) and M.~Jimbo and T.~Miwa
\cite{JimMiw:1981:MPDLODEWRC-II}, as well as more recent examples by G.~Filipuk and H.~\.{Z}o\l\c{a}dek \cite{ZolFil:2015:PEEIEF}
and A.~Its and A.~Prokhorov \cite{ItsPro:2018:SHPIF}. In addition, for the $\Pain{IV}$ equation there is a very interesting example
of a Hamiltonian constructed by T.~Kecker \cite{Kec:2016:CHSWMS}, so for that reason in this paper we consider in detail the 
differential $\Pain{IV}$ equation, and then just summarize the data for $\Pain{V}$ and $\Pain{VI}$. The
$\Pain{II}$ and $\Pain{III}$ examples were considered previously in \cite{DzhFilLigSto:2021:HRSPE} and 
\cite{FilLigSto:2023:RBDHFTPE}.

The paper is organized as follows. In Section~\ref{sec:P4-Hams} we give examples of different Hamiltonian forms of the 
$\Pain{IV}$ equation and list explicit change of  coordinates between them. We also show how to explicitly match the 
Hamiltonian functions. In Section~\ref{sec:prelim} we carefully review the construction of the space of initial 
conditions for a Painlev\'e Hamiltonian system and, following \cite{KajNouYam:2017:GAPE}, 
recall some standard facts about the geometry of the standard realization of the 
resulting  $E_{6}^{(1)}$ surface, including the choice of root bases for the surface and symmetry sub-lattices 
in its Picard lattice. We also discuss the Hamiltonian formalism for time-dependent Hamiltonian functions. 
In Section~\ref{sec:P4-matching} we describe an algorithmic procedure on constructing a birational change of coordinates 
matching two different Hamiltonian systems based on identifying the geometry of the space of initial conditions, and then 
apply it to different forms of $\Pain{IV}$ equation, providing detailed descriptions of necessary computations. Then in 
Sections~\ref{sec:P5} and \ref{sec:P6} we give a summary of such matchings for $\Pain{V}$ and $\Pain{VI}$ respectively. 
Conclusions and some possible follow-up questions are discussed in Section~\ref{sec:conclusion}.


\section{The Fourth Painev\'e Equation $\Pain{IV}$ and its Hamiltonian Forms} 
\label{sec:P4-Hams}

Our main example is the differential $\Pain{IV}$ equation 
\begin{equation}\label{eq:P4-std}
\Pain{IV}={\Pain{IV}}_{\alpha,\beta}:\quad \frac{d^2 w}{dt^2} = \frac{1}{2w}\left(\frac{d w}{dt}\right)^2 + 
\frac{3}{2}w^3 + 4tw^2 + 2(t^2 - \alpha)w + \frac{\beta}{w},
\end{equation}
where $t$ is an independent variable, $w=w(t)$ is a dependent variable, and $\alpha$, $\beta$ are complex parameters. 
We begin by showing how this equation appears, in a Hamiltonian form, in a variety of different scenarios, and give some
examples of identification based on just the context of the problem.


\paragraph{Its--Prokhorov and Jimbo-Miwa Hamiltonians} 
\label{par:IP-JM-Ham-4}
We begin by considering the Hamiltonian system obtained by A.~Its and A.~Prokhorov \cite{ItsPro:2018:SHPIF} as isomonodromic deformations of
a  $2\times 2$ linear system with one irregular singular point at $z=\infty$ with the Poincare rank $2$ and one Fuchsian singularity at $z=0$,
\begin{equation}\label{eq:IP-linear-4}
	\begin{aligned}
	\frac{d\Phi}{dz} &= A(z;t) \Phi,\quad A(z;t) = \frac{A_{-1}(t)}{z} + A_{0}(t) + A_{1}(t) z\\
	A_{1} = \begin{bmatrix} 1 & 0 \\ 0 & -1 
	\end{bmatrix},\quad 
	A_{0} &= \begin{bmatrix} t & k \\ -\frac{q(4p-q-2t)+4\Theta_{\infty}}{2k} &-t 
	\end{bmatrix},\quad		 
	A_{-1} = \frac{1}{2}\begin{bmatrix} \frac{q(4p-q-2t)}{2} & -k q\\
	\frac{q^2(4p-q-2t)^2-16\Theta_0^2}{4k q} & -\frac{q(4p-q-2t)}{2} \end{bmatrix},
\end{aligned}
\end{equation}
where $q=q(t)$, $p=p(t)$, $k=k(t)$ and $\Theta_{0}$, $\Theta_{\infty}$ are some (time independent) 
parameters (that are formal monodromy exponents at the 
corresponding singular points). The isomonodromy deformation is then given by the equation
\begin{equation}\label{eq:IP-isom-4}
	\frac{d\Phi}{dt} = B(z;t) \Phi,\qquad B(z) = B_{1}(t) z + B_{0}(t),\qquad B_{1}=A_{1},\quad 
	B_{0} = \begin{bmatrix} 0 & k \\ 
	-\frac{q(4p-q-2t)+4\Theta_{\infty}}{2k} &0 \end{bmatrix}.
\end{equation}
The compatibility condition between \eqref{eq:IP-linear-4} and \eqref{eq:IP-isom-4} is the zero-curvature equation
$A_{t} - B_{z} + [A,B] = 0$ that results in the equations
\begin{align}
	\frac{dq}{dt} &= 4 p q,\label{eq:IP-system-4-1} \\
	\frac{dp}{dt} &= {-2p^2}+3q^2/8+q t+t^2/2-\Theta_{\infty}+1/2-2\Theta_0^2/q^2,\label{eq:IP-system-4-2} \\
	\frac{dk}{dt} &=-(q+2t)k.	\notag
\end{align}
Equations \eqref{eq:IP-system-4-1}  and \eqref{eq:IP-system-4-2} do not depend on the function $k$ and eliminating $p$ 
we obtain the fourth Painlev\'e equation \eqref{eq:P4-std} for the function $q(t)$ with parameters 
\begin{equation}\label{eq:IP-pars-4}
\alpha = 2\Theta_{\infty}-1, \quad \beta=-8\Theta_0^2.
\end{equation}
The system of equations \eqref{eq:IP-system-4-1}  and \eqref{eq:IP-system-4-2} is a non-autonomous Hamiltonian system  with the 
Its-Prokhorov Hamiltonian $\Ham{IP}{IV}=\Ham{IP}{IV}(q,p;t)$ (corresponding to a symplectic form $\symp{IP}$) given by 
 \begin{equation}\label{eq:IP-Ham-4}
 	\left\{
	\begin{aligned}
		\frac{dq}{dt} &= \frac{\partial \Ham{IP}{IV}}{\partial p},\\
		\frac{dp}{dt} &= -\frac{\partial \Ham{IP}{IV}}{\partial q},
	\end{aligned}
	\right.\qquad\text{where}\quad
	\begin{aligned}
			\Ham{IP}{IV}(q,p;t)&=2p^2 q-\frac{q^3}{8}-\frac{t q^2}{2}+\frac{(2\Theta_{\infty}-1-t^2)q}{2}+2\Theta_{\infty}t-\frac{2\Theta_0^2}{q},\\
		\symp{IP} &= dp\wedge dq.
	\end{aligned}
 \end{equation}

 In a much earlier paper \cite{JimMiw:1981:MPDLODEWRC-II} M.~Jimbo and T.~Miwa gave a different parameterization of the coefficient matrices of 
 equations \eqref{eq:IP-system-4-1} -- \eqref{eq:IP-system-4-2},
 \begin{equation}\label{eq:JM-linear-4}
 	A_0=\begin{bmatrix} t & u \\ \frac{2(z-\Theta_0-\Theta_{\infty})}{u} & -t \end{bmatrix},\quad
	A_{-1}=\begin{bmatrix} -z+\Theta_0 & -\frac{u y}{2} \\ \frac{2z(z-2\Theta_0)}{u y} & z-\Theta_0 \end{bmatrix},\quad	
	B_0=\begin{bmatrix} 0 & u \\ 	\frac{2(z-\Theta_0-\Theta_{\infty})}{u} & 0 \end{bmatrix},	 	
 \end{equation}
 where $z=z(t)$, $y=y(t)$, $u=u(t)$, and $\Theta_0$ and $\Theta_{\infty}$ are the same parameters as in \cite{ItsPro:2018:SHPIF}.
 The compatibility condition in this case leads to the following system of nonlinear differential equations for the functions $y$, $z$ and $u$:
 \begin{align}
 	\frac{dy}{dt} &= -4z+y^2+2t y+4\Theta_0,\label{eq:JM-system-4-1} \\
 	\frac{dz}{dt} &= -\frac{2}{y}z^2+\left(-y+\frac{4\Theta_0}{y}\right)z+(\Theta_0+\Theta_{\infty})y,\label{eq:JM-system-4-2} \\
 	\frac{du}{dt} &= -(y-2t)u.	\notag
 \end{align}
 Eliminating the function $z$ from the first two equations \eqref{eq:JM-system-4-1} and  \eqref{eq:JM-system-4-2} 
 one can obtain the fourth Painlev\'e equation for the variable $y$ with the same values of parameters $\alpha$ and $\beta$ as in \eqref{eq:IP-pars-4}
 above. The system \eqref{eq:JM-system-4-1} and  \eqref{eq:JM-system-4-2} is also Hamiltonian with the 
 Jimbo-Miwa Hamiltonian $\Ham{JM}{IV}=\Ham{JM}{IV}(y,z,t)$: 
  \begin{equation}\label{eq:JM-Ham-4}
  	\left\{
 	\begin{aligned}
 		-\frac{1}{y} \frac{dy}{dt} &= \frac{\partial \Ham{JM}{IV}}{\partial z},\\
 		-\frac{1}{y} \frac{dz}{dt} &= -\frac{\partial \Ham{JM}{IV}}{\partial y},			
 	\end{aligned}
 	\right.\qquad\text{where}\quad
	\begin{aligned}
		\Ham{JM}{IV}(y,z;t) &= \frac{2}{y}z^2-\left(y+2t+4\frac{\Theta_0}{y}\right)z+(\Theta_0+\Theta_{\infty})(y+2t),\\
		\symp{JM} &=(1/y)\,dy\wedge dz.
	\end{aligned} 	
  \end{equation}
 Note that this time we need to use the 
 logarithmic symplectic structure given by $\symp{JM}$; we explain the geometry behind this in 
Section~\ref{sub:JM-soic-4}.
 
 
What is the relationship between the Hamiltonian systems \eqref{eq:IP-Ham-4} and \eqref{eq:JM-Ham-4}? In this case we do not need any elaborate tools
to match these two systems exactly. Indeed, comparing the matrix coefficients in \eqref{eq:IP-linear-4} and \eqref{eq:JM-linear-4} we immediately
see that $u=k$, parameters $\Theta_{0}$ and $\Theta_{\infty}$ match, and the variables $(q,p)$ and $(y,z)$ are related by the following 
\emph{birational} change of variables:
 \begin{equation}\label{eq:JMtoIP-4}
	 \varphi: \left\{\begin{aligned}
	 	y(q,p)&=q,\\
		z(q,p)&=\frac{1}{4}(q^2-4p q+2q t+4\Theta_0),
	 \end{aligned}\right.
 \qquad\text{and conversely,} \qquad \varphi^{-1}:
 	\left\{\begin{aligned}
 		 q(x,y) &= y,\\ 
		 p(x,y) &= \frac{2t y+y^2-4z+4\Theta_0}{4y}.
 	\end{aligned}\right.
 \end{equation}
One can then check directly that this change of variables transforms system \eqref{eq:IP-Ham-4} to system \eqref{eq:JM-Ham-4}. 
However, since the change of variables is \emph{time-dependent}, we get some additional terms in the expression 
for the Hamiltonian.  Specifically, the relationship between the Hamiltonians is obtained via the pull-back of a certain 
 $2$-form on the extended phase space,
\begin{equation*}
	\Symp{IP} = dp\wedge dq - d\Ham{IP}{IV}\wedge dt = \varphi^{*}(\Symp{JM}),\qquad \Symp{JM} = \frac{1}{y} dy \wedge dz - d\Ham{JM}{IV}\wedge dt.
\end{equation*}
Thus, using \eqref{eq:JMtoIP-4}, we see that 
\begin{equation*}
	\Ham{JM}{IV}(y,z;t) = \Ham{IP}{IV}(q(y,z,t),p(y,z,t);t) - \frac{y}{2}\qquad\text{and}\qquad 
	\Ham{IP}{IV}(q,p;t) = \Ham{JM}{IV}(y(q,p,t),z(q,p,t);t) + \frac{q}{2}.
\end{equation*}
This remark is essential and so we explain it in more detail  in Section~\ref{sub:non-auto-Hams}.

\paragraph{Okamoto Hamiltonian} 
\label{par:Ok-Ham-4}

In \cite{ItsPro:2018:SHPIF,JimMiw:1981:MPDLODEWRC-II} the fourth Painlev\'e equation appeared from the monodromy preserving deformation of the linear system 
\eqref{eq:IP-linear-4}. On the other hand, in the work of K.~Okamoto the same equation appeared from the monodromy preserving deformation of a scalar 
second order linear differential equation, with coefficients depending on $t$, that has certain singularities in the complex plane. This approach leads 
to a very different Hamiltonian system with a \emph{polynomial} Hamiltonian \cite{Oka:1980:PHAWPE,Oka:1980:PHAWPEIDESPH,Oka:1986:SPEISFPEPP} given by
  \begin{equation}\label{eq:Ok-Ham-4}
  	\left\{
 	\begin{aligned}
 		\frac{df}{dt} &= \frac{\partial \Ham{Ok}{IV}}{\partial g} =4fg-f^2-2tf-2\kappa_0,\\
		\frac{dg}{dt} &= -\frac{\partial \Ham{Ok}{IV}}{\partial f} =-2g^2+(2f+2t)g -\theta_{\infty},			
 	\end{aligned}
 	\right.\qquad\text{where}\quad
	\begin{aligned}
	\Ham{Ok}{IV}(f,g;t) &= 2fg^2 - (f^2 + 2tf + 2 \kappa_0)g + \theta_\infty f.	\\
	\symp{Ok} &= dg\wedge df.
	\end{aligned}
  \end{equation}
Eliminating the function $g=g(t)$ from these equations we get $\Pain{IV}$ equation \eqref{eq:P4-std} for the function $f=f(t)$ with parameters 
\begin{equation}\label{eq:Ok-pars-4}
\alpha=1+2\theta_{\infty}-\kappa_0,\quad \beta=-2\kappa_0^2.
\end{equation}
It turns out that the \emph{Okamoto space of initial conditions} for this system has the simplest geometry, and so this system 
will be the reference example for the present paper.

\begin{remark}We use the recent comprehensive survey paper \cite{KajNouYam:2017:GAPE} (see also \cite{Nou:2004:PETS}) as the main reference for 
	our choice of the geometric data.	
\end{remark}


\paragraph{Filipuk-\.{Z}o\l\c{a}dek Hamiltonian} 
\label{par:FZ-Ham-4}
	A different approach to Hamiltonian structures of Painlev\'e equations that starts directly with the equation itself and not the isomonodromy problem,
	and that leads to rational Hamiltonians, was suggested by G.~Filipuk and H.~\.{Z}o\l\c{a}dek in \cite{ZolFil:2015:PEEIEF}. 
	The key observation there is that Painlev\'e equations can be written in the Li\'enard form
	\begin{equation*}
		\frac{d^{2} x}{dx^{2}} = A(x,t)\left(\frac{dx}{dt}\right)^{2} + B(x,t) \frac{dx}{dt} + C(x,t).
	\end{equation*}
	If we then introduce the new variable $y$ via $\frac{dx}{dt} = D(x,t)y$ (and use the subscript notation for partial derivatives, 
	$A_{x} = \frac{\partial A(x,t)}{\partial x}$, etc.), we observe that this equation transforms into a system that is Hamiltonian 
	with the Hamiltonian function $\Ham{FZ}{}(x,y;t)$, 
	\begin{equation}
		\left\{
		\begin{aligned}
			\frac{dx}{dt} &= D y = \Ham{FZ}{\mathit y},\\
			\frac{dy}{dt} &= (AD - D_{x})y^{2} + \left(B - \frac{D_{t}}{D}\right)y + \frac{C}{D} = - \Ham{FZ}{\mathit x},
		\end{aligned}
		\right.
	\end{equation}
	if the following compatibility condition 
	\begin{equation*}
		\Ham{FZ}{\mathit xy} = D_{x}y = 2(D_{x} - AD)y + \left(\frac{D_{t}}{D} - B\right)\quad\text{or, equivalently,}\quad 
		\left\{\begin{aligned}
			(\log D)_{x} &= 2A\\
			(\log D)_{t} &= B
		\end{aligned},\right.
	\end{equation*}
	holds, which is true when $2A_{t} = B$. Then we get the function $D(x,t)$ and the Hamiltonian $\Ham{FZ}{}(x,y,t)$ is
	\begin{equation}\label{eq:FZ-Ham}
		\Ham{FZ}{}(x,y;t) = D(x,t)\frac{y^{2}}{2} - \int^{x}\frac{C(x,t)}{D(x,t)}\,dx.
	\end{equation}
	
	For $\Pain{IV}$, we get $A(x,t) = \frac{1}{2x}$, $B(x,t)=0$, and $C(x,t) = \frac{3}{2} x^{3} + 4 t x^{2} + 2(t^{2} - \alpha)x + \frac{\beta}{x}$,
	and so the compatibility condition is trivially satisfied and we can take $D(x,t) = \lambda x$ for some $\lambda\neq 0$. Then 
	\begin{equation*}
		\Ham{FZ}{IV}(x,y;s)=\lambda x \frac{y^{2}}{2} - \frac{1}{\lambda}\left(\frac{x^{3}}{2} + 2t x^{2} + 2(t^{2} - \alpha)x - \frac{\beta}{x}\right).
	\end{equation*}
	In \cite{ZolFil:2015:PEEIEF} the authors take $\lambda = 1$, which gives the Hamiltonian 
	\begin{equation*}
		\widetilde{\Ham{FZ}{IV}}(x,\tilde{y};s)=\frac{x\tilde{y}^{2} -x^{3}}{2}  - 2 sx^{2} -2(s^{2}-\alpha) x + \frac{\beta}{x}.
	\end{equation*}
	However, from the geometric point of view, to avoid a normalization coefficient at the symplectic form that corresponds to the 
	standard normalization condition on the root variables, see Section~\ref{ssub:period-map-a21}, it is better to 
	take $\lambda = 4$, which results in a minor rescaling. The
	new Hamiltonian then essentially coincides with the Its-Prokhorov Hamiltonian \cite{ItsPro:2018:SHPIF}  
	(and the choice $\lambda=4$ can be also observed at this point), 
	\begin{equation}\label{eq:FZ-IP-Ham-4}
	\Ham{FZ}{IV}(x,y;s) = (1/4)\widetilde{\Ham{FZ}{IV}}(x,4y;s) = 
	\Ham{IP}{IV}\left(x,y;t; \Theta_{0}^{2} = - \frac{\beta}{8}, \Theta_{\infty} = \frac{1 + \alpha}{2}\right) - t(1+\alpha),
	\end{equation}
	%
	and so we do not consider it further. Essentially the same happens for $\Pain{VI}$, however, for $\Pain{V}$ we get something quite
	different from the other examples, as described in Section~\ref{sub:FZ-5}.

\paragraph{Kecker Hamiltonian} 
\label{par:Kek-Ham-4}

Finally, in this paper we shall also deal with the Hamiltonian that appeared in \cite{Kec:2016:CHSWMS} and was further studied in 
\cite{Kec:2019:SICCHS, Ste:2018:CMF}. This cubic Hamiltonian is given by 
  \begin{equation}\label{eq:Kek-Ham-4}
  	\left\{
 	\begin{aligned}
 		\frac{dx}{dz} &= \frac{\partial \Ham{Kek}{IV}}{\partial y} =y^{2} + z x + \tilde{\alpha},\\
		\frac{dy}{dz} &= -\frac{\partial \Ham{Kek}{IV}}{\partial x} =-x^{2} - z y - \tilde{\beta},			
 	\end{aligned}
 	\right.\qquad\text{where}\quad
	\begin{aligned}
	\Ham{Kek}{IV}(x,y;z) &= \frac{1}{3}(x^{3} + y^{3}) + z xy + \tilde{\alpha} y + \tilde{\beta} x,\\
	\symp{Kek} &= dy\wedge dx.
	\end{aligned}
  \end{equation}
Now the relationship with the standard $\Pain{IV}$ equation is less straightforward.  First, note that now we have $z$ as an independent variable,
$x=x(z)$ and $y=y(z)$ are dependent variables, and $\tilde{\alpha}$, $\tilde{\beta}$ are (complex) parameters. Let us introduce a new
variable $\tilde{w} = \tilde{w}(z) := x(z) + y(z) - z$. Then, as shown in \cite{Kec:2016:CHSWMS,Kec:2019:SICCHS, Ste:2018:CMF}, 
\eqref{eq:Kek-Ham-4} gives the following equation on $\tilde{w}(z)$:
\begin{equation*}
	\frac{d^{2}\tilde{w}}{dz^{2}} = \frac{1}{2 \tilde{w}}\left(\frac{d \tilde{w}}{dz}\right)^{2} - \frac{1}{2}\tilde{w}^{3} - 2 z \tilde{w}^{2} 
	- \frac{\tilde{w}}{w} (3 z^{2} + 2(\tilde{\alpha} + \tilde{\beta})) - \frac{(1 - \tilde{\alpha} + \tilde{\beta})^{2}}{\tilde{w}},
\end{equation*}
which then reduces to the standard equation \eqref{eq:P4-std} for parameters
\begin{equation}\label{eq:Kek-pars}
\alpha = \frac{\mathfrak{i} }{\sqrt{3}}(\tilde{\alpha} + \tilde{\beta}),\qquad \beta = -\frac{2}{9} (1 - \tilde{\alpha} + \tilde{\beta})^{2}
\end{equation}
if we put $z = \left(-\frac{4}{3}\right)^{\frac{1}{4}} t$ and 
$\tilde{w}(z) = \frac{3}{2} \left(-\frac{4}{3}\right)^{\frac{1}{4}} w(t)$.

\begin{remark} We need to mention that not every known example of Hamiltonian systems that reduce to Painlev\'e equations fits into our approach.
	As an example, consider the following system studied in \cite{Tak:2001:PCR}
    \begin{equation}\label{eq:Tak-Ham-4}
    	\left\{
   	\begin{aligned}
   		\frac{dq}{dt} &= \frac{\partial \Ham{Tak}{IV}}{\partial p} =p,\\
  		\frac{dp}{dt} &= -\frac{\partial \Ham{Tak}{IV}}{\partial q} = - \frac{\partial V(q,t)}{\partial q},			
   	\end{aligned}
   	\right.\qquad\text{where}\quad
  	\begin{aligned}
  	\Ham{Tak}{IV}(q,p;t) &= \frac{p^2}{2}+V(q,t),	\\
  	\symp{Tak} &= dp\wedge dq,
  	\end{aligned}
    \end{equation}
	and the potential $V(q,t)$ is given by 
	\begin{equation*}
	V(q,t)=-\frac{1}{2}\left(\frac{q}{2}\right)^6-2t\left(\frac{q}{2}\right)^4-2(t^2-\alpha)\left(\frac{q}{2}\right)^2+\beta\left(\frac{q}{2}\right)^{-2}.
	\end{equation*}
	
	By taking $w=(q/2)^2$ one can easily show that $w$ satisfies \eqref{eq:P4-std}. Moreover, in \cite{Tak:2001:PCR} it was shown that if one takes 
	\begin{equation}\label{eq:OktoTak-4}
	    \left\{\begin{aligned}
    	f(q,p)&=\left(\frac{q}{2}\right)^2,\\
	   	g(q,p)&=\frac{t}{2}+\frac{2\kappa_0}{q^2}+\frac{p}{2q}+\frac{q^2}{16},
	    \end{aligned}\right.
	\end{equation}
	then \eqref{eq:Ok-Ham-4} transforms to \eqref{eq:Tak-Ham-4} with 
	$\alpha=1+2\theta_{\infty}-\kappa_0$ and $\beta=-2\kappa_0^2$, which is exactly  \eqref{eq:Ok-pars-4}.
	Further, the Liouville form giving the canonical transformation is 
	\begin{equation*}
	\theta =  g df - \Ham{Ok}{IV} dt= \frac{1}{4} (p dq-\Ham{Tak}{IV}dt)+d\phi(q, p,t),
	\end{equation*}
	where $d\phi(q,p,t)$ is an exact form which can be calculated explicitly.

	However, the change of variables \eqref{eq:OktoTak-4} is only algebraic, but not birational. In this case it is impossible to 
	construct the space of initial conditions for the system \eqref{eq:Tak-Ham-4} in the sense of  \cite{Oka:1979:FAESOPCFPP}
	by resolving singularities, and so our approach cannot be used. This example is considered in detail from the geometric point of view 
	in \cite{FilSto:2022:TRFPSGRAE}.
\end{remark}

We summarize the relationship between different $\Pain{IV}$ systems considered in this paper in Figure~\ref{fig:allHams-4}.

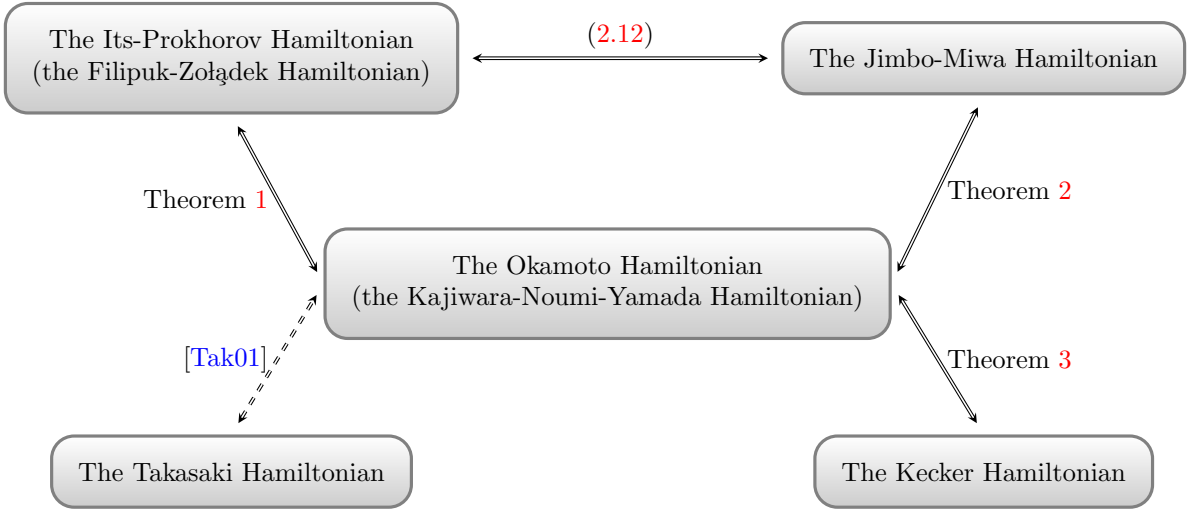
\begin{figure}[ht]
	\begin{center}
	\begin{tikzpicture}[>=stealth,node distance=5mm,
	                    rblock/.style={
	                      rectangle,minimum size=6mm,rounded corners=3mm,
	                      very thick,draw=black!50,
	                      top color=white,bottom color=black!20}]
	  \node (Ok) at (0,0) [rblock,inner sep=10pt,  align=center]   {The Okamoto Hamiltonian \\ 
	  											(the Kajiwara-Noumi-Yamada Hamiltonian)};
	  \node (IP) at (-5,3) [rblock,inner sep=10pt,  align=center] 
	    {The Its-Prokhorov Hamiltonian \\ (the Filipuk-Zo\l\c{a}dek Hamiltonian)};

	  \node  (JM) at (5,3) [rblock,inner sep=10pt,  align=center]    {The Jimbo-Miwa Hamiltonian};

	  \node (Tak) at (-5,-2.5) [rblock,inner sep=10pt,  align=center]    {The Takasaki Hamiltonian};

	  \node (Kek) at (5,-2.5) [rblock,inner sep=10pt,  align=center]    {The Kecker Hamiltonian};
	  
	  \draw[-latex, <->,double, shorten >=5pt, shorten <=5pt] (IP.east)--(JM.west) node[pos=0.5, above] {\eqref{eq:JMtoIP-4}};
	  \draw[-latex, <->,double, shorten >=5pt, shorten <=5pt] (IP.south)--(Ok.west) node[pos=0.5, left] {Theorem~\ref{thm:IPtoOK-4}};
	  \draw[-latex, <->,double, dashed, shorten >=5pt, shorten <=5pt] (Tak.north)--(Ok.west) node[pos=0.5, left] {\cite{Tak:2001:PCR}};
	  \draw[-latex, <->,double, shorten >=5pt, shorten <=5pt] (JM.south)--(Ok.east) node[pos=0.5, right] {Theorem~\ref{thm:JMtoOK-4}};
	  \draw[-latex, <->,double, shorten >=5pt, shorten <=5pt] (Kek.north)--(Ok.east) node[pos=0.5, right] {Theorem~\ref{thm:KektoOk-4}};
	  
	\end{tikzpicture}		
	\end{center}
	\caption{Relationship between different Hamiltonian Systems for $\Pain{IV}$}
	\label{fig:allHams-4}
\end{figure}



\subsection{Coordinate transformations between different Hamiltonian forms for $\Pain{IV}$} 
\label{sub:summary-P4}

We now give a summary of explicit coordinate transformations and parameter matching 
for different versions of Hamiltonian representations for the differential Painlev\'e
$\Pain{IV}$ equation described above. 

%
%
%

\paragraph{Its-Prokhorov and Jimbo-Miwa systems} 
\label{par:IP-JM-4}

\begin{theorem}\label{thm:IPtoOK-4} The Its-Prokhorov Hamiltonian system \eqref{eq:IP-Ham-4} and the Okamoto Hamiltonian system \eqref{eq:Ok-Ham-4} are 
	related by the following change of variables and parameter correspondence:
    \begin{equation}\label{eq:IPtoOk-coords-4}
   	 \left\{\begin{aligned}
   	 	q(f,g,t)&=f,\\
   		p(f,g,t)&= g - \frac{f}{4} - \frac{t}{2} - \frac{\kappa_{0}}{2f},\\
		\Theta_{0}&=\frac{\kappa_{0}}{2},\\  
		\Theta_{\infty} &= 1 + \theta_{\infty} - \frac{\kappa_{0}}{2},
   	 \end{aligned}\right.
    \qquad\text{and conversely} \qquad 
    	\left\{\begin{aligned}
   	 	f(q,p,t)&=q,\\
   		g(q,p,t)&=p + \frac{q}{4} + \frac{t}{2} + \frac{\Theta_{0}}{q},\\
		\kappa_{0} &= 2\Theta_{0},\\  
		\theta_{\infty} &= -1 + \Theta_{0} + \Theta_{\infty}.
    	\end{aligned}\right.
    \end{equation}	
	The Hamiltonians are then related by 
	\begin{equation}\label{eq:IPtoOk-Hams-4}
		\Ham{IP}{IV}(q,p;t;\Theta_{0},\Theta_{\infty}) = \Ham{Ok}{IV}(f(q,p),g(q,p);t;\kappa_{0},\theta_{\infty}) + \frac{q}{2} + 
		2t(\Theta_{0} + \Theta_{\infty}),
	\end{equation}
	where the purely $t$-dependent terms can in fact be ignored.
\end{theorem}
\begin{proof} The change of variables \eqref{eq:IPtoOk-coords-4} is established in Lemma~\ref{lem:coords-IP-Ok-4}. To relate the Hamiltonians, 
	as explained in Section~\ref{sub:non-auto-Hams}
	we need to match the $2$-forms $\Symp{Ok} = \symp{Ok} - d\Ham{Ok}{IV}\wedge dt = \symp{IP} - d\Ham{IP}{IV}\wedge dt=\Symp{IP}$. We get 
	\begin{align*}
		\Symp{Ok} &= \symp{Ok} - d\Ham{Ok}{IV}\wedge dt = dg\wedge df - 
		d\left(2fg^2 - (f^2 + 2tf + 2 \kappa_0)g + \theta_\infty f\right) \wedge dt\\
		& = dp\wedge dq + \frac{1}{2} dt\wedge dq - d\left(2p^2 q-\frac{q^3}{8}-
		\frac{t q^2}{2} - \frac{t^{2} q}{2} +  \Theta_{\infty}q-\frac{2\Theta_0^2}{q} - q - 2 t \Theta_{0}\right) \wedge dt\\
		& = \symp{IP} - d\left(\Ham{IP}{IV} - 2t (\Theta_{0} + \Theta_{\infty})\right)\wedge dt.
	\end{align*}
\end{proof}

A similar result holds for the Jimbo-Miwa system \eqref{eq:JM-Ham-4}.
\begin{theorem}\label{thm:JMtoOK-4}
	The change of coordinates and parameter matching between the Jimbo-Miwa \eqref{eq:JM-Ham-4} and Okamoto \eqref{eq:Ok-Ham-4} Hamiltonian systems is given by
	    \begin{equation}\label{eq:JMtoOk-coords-4}
   	 \left\{\begin{aligned}
   	 	y(f,g,t)&=f,\\
   		z(f,g,t)&= \frac{f^{2}}{2} - fg + ft + \kappa_{0},\\
		\Theta_{0}&=\frac{\kappa_{0}}{2},\\  
		\Theta_{\infty} &= 1 + \theta_{\infty} - \frac{\kappa_{0}}{2},
   	 \end{aligned}\right.
    \qquad\text{and conversely} \qquad 
    	\left\{\begin{aligned}
   	 	f(y,z,t)&=y,\\
   		g(y,z,t)&=\frac{y}{2}- \frac{z}{y}+ t + \frac{2\Theta_{0}}{y},\\
		\kappa_{0} &= 2\Theta_{0},\\  
		\theta_{\infty} &= -1 + \Theta_{0} + \Theta_{\infty}.
    	\end{aligned}\right.
	    \end{equation}
		The Hamiltonians are then related by 
		\begin{equation}\label{eq:JMtoOk-Hams}
			\Ham{JM}{IV}(y,z;t;\Theta_{0},\Theta_{\infty}) = \Ham{Ok}{IV}(f(y,z,t),g(y,z,t);t;\kappa_{0}= 2\Theta_{0},\theta_{\infty} = 
			-1 + \Theta_{0} + \Theta_{\infty}) + y + 2t(\Theta_{0} + \Theta_{\infty}).
		\end{equation}
\end{theorem}
\begin{proof}The change of variables \eqref{eq:JMtoOk-coords-4} is established in Lemma~\ref{lem:JMtoOk-coords-4}. 
	To relate the Hamiltonians, note that 
	\begin{align*}
		\Symp{Ok} &= \symp{Ok} - d\Ham{Ok}{IV}\wedge dt = dg\wedge df - 
		d\left(2fg^2 - (f^2 + 2tf + 2 \kappa_0)g + \theta_\infty f\right) \wedge dt\\
		&= - \frac{dz}{y}\wedge dy + dt \wedge dy - d\left(\frac{2 z^{2}}{y} - zy - 2 t z -
		 \frac{4 \Theta_{0}z}{y} - y + (\Theta_{0} + \Theta_{\infty})y\right)\wedge dt\\
		 &= \frac{1}{y} dy\wedge  dz + d\left(\Ham{JM}{IV} - 2t (\Theta_{0} + \Theta_{\infty})\right)\wedge dt
		  = \symp{JM} - d\Ham{JM}{IV}\wedge dt,
	\end{align*}
	and note that the $t$-dependent terms can be ignored, as usual.
\end{proof}

Note that combining these change of variables immediately gives us the change of variables between the Its-Prokhorov 
coordinates $(q,p)$ and the Jimbo-Miwa coordinates $(y,z)$ that we already obtained in equation \eqref{eq:JMtoIP-4}.

\paragraph{The Kecker system} 
\label{par:Kek-4}
The identification between the Kecker and the Okamoto systems is more complicated. In particular, the geometric normalization of 
the symplectic form differs from the original normalization of the symplectic form  $\symp{Kek}$ in \eqref{eq:Kek-Ham-4};
$\sympt{Kek} = (1/3) \symp{Kek}$. 

\begin{theorem}\label{thm:KektoOk-4} The Kecker Hamiltonian system \eqref{eq:Kek-Ham-4} and the Okamoto Hamiltonian system \eqref{eq:Ok-Ham-4} are 
	related by the following change of variables and parameter correspondence:
    \begin{align}\label{eq:KektoOk-coords-4}
   	 &\left\{\begin{aligned}
   	 	x(f,g,t)&=\frac{(1 + \mathfrak{i})}{4(3)^{3/4}}\Big( 3(\sqrt{3}-\mathfrak{i})f + 12 \mathfrak{i} g + 2 (\sqrt{3} - 3 \mathfrak{i})t \Big),\\
   	 	y(f,g),t&=\frac{(1 + \mathfrak{i})}{4(3)^{3/4}}\Big( 3(\sqrt{3}+\mathfrak{i})f - 12 \mathfrak{i} g + 2 (\sqrt{3} + 3 \mathfrak{i})t \Big),\\
		z(t) &= \left(-\frac{4}{3}\right)^{\frac{1}{4}}t,
   	 \end{aligned}\right.\quad 
	 \begin{aligned}
		\tilde{\alpha}&=\frac{(\sqrt{3} - 3 \mathfrak{i}) - 6 \mathfrak{i} \theta_{\infty} - 3(\sqrt{3} - \mathfrak{i})\kappa_{0}}{2 \sqrt{3}},\\  
		\tilde{\beta}&=\frac{(-\sqrt{3} - 3 \mathfrak{i}) - 6 \mathfrak{i} \theta_{\infty} + 3(\sqrt{3} + \mathfrak{i})\kappa_{0}}{2 \sqrt{3}},	 	
	 \end{aligned}\\
	 \intertext{and converseley,}
    	&\left\{\begin{aligned}
   	 	f(x,y,z)&=\frac{1-\mathfrak{i}}{3^{3/4}}(x + y - z),\\
   		g(x,y,z)&= - \frac{(1 + \mathfrak{i})}{4(3)^{3/4}}\Big( (\sqrt{3}+\mathfrak{i})x - (\sqrt{3} - \mathfrak{i})y + 2 \mathfrak{i} z  \Big),\\
		t(z) &= \left(-\frac{3}{4}\right)^{\frac{1}{4}}z,
    	\end{aligned}\right.\quad
		\begin{aligned}
		\kappa_{0} &= \frac{1 - \tilde{\alpha} + \tilde{\beta}}{3},\\
		\theta_{\infty} &= \frac{-2 + (\mathfrak{i}\sqrt{3} - 1)\tilde{\alpha} + (\mathfrak{i} \sqrt{3}+ 1)\tilde{\beta} }{6}.			
		\end{aligned}\label{eq:OktoKek-coords-4}
    \end{align}
	The Hamiltonians are then related by 
	\begin{align}\label{eq:KektoOk-Hams-4}
		\Ham{Kek}{IV}(x,y;z;\tilde{\alpha},\tilde{\beta}) &= 3\left( \left(\frac{-3}{4}\right)^{\frac{1}{4}}
		\Ham{Ok}{IV}(f(x,y,z),g(x,y,z);t(z);\kappa_{0}(\tilde{\alpha},
		\tilde{\beta}),\theta_{\infty}(\tilde{\alpha},\tilde{\beta})) \right.\\
				&\qquad\qquad \left.   - \frac{1 + \mathfrak{i} \sqrt{3} }{6} x + 
		\frac{1 - \mathfrak{i} \sqrt{3} }{6} y  
		+ \frac{z^{3}}{9} + \frac{1 + \mathfrak{i} \sqrt{3}}{6} z \tilde{\alpha} + \frac{1 - \mathfrak{i} \sqrt{3}}{6} z \tilde{\beta} \right).\notag
	\end{align}
\end{theorem}
\begin{proof} The change of variables (\ref{eq:KektoOk-coords-4}--\ref{eq:OktoKek-coords-4})
	is established in Lemma~\ref{lem:coords-Kek-Ok-4}. For the Hamiltonians, given the difference in normalizations, we 
	should have 
	\begin{equation*}
		\Symp{Ok} = \symp{Ok} - d\Ham{Ok}{IV}\wedge dt = \frac{1}{3}\cdot (\symp{Kek} - d\Ham{Kek}{IV}\wedge dz) = \frac{1}{3} \Symp{Kek}.
	\end{equation*}
	That is,
	\begin{align*}
		\Symp{Ok} &= dg \wedge df - d\left(2fg^2 - f^2g - 2tfg - 2 \kappa_0 g + \theta_\infty f\right)\wedge dt \\
		& = -\frac{2 \sqrt{3} dx \wedge dy - \sqrt{3}(1 + \mathfrak{i}\sqrt{3}) dx\wedge dz  + \sqrt{3}(1 - \mathfrak{i}\sqrt{3})dy \wedge dz}{6 \sqrt{3}}\\
		&\qquad - d\frac{(1 + \mathfrak{i})\Big(2 \sqrt{3} (x^{3} + y^{3} - z^{3} + 3 xyz) + 3(3 \mathfrak{i} + \sqrt{3} + 2 \sqrt{3} \tilde{\beta})x 
		+ 3(3 \mathfrak{i} - \sqrt{3} + 2 \sqrt{3} \tilde{\alpha}) y
		\Big)}{18 (3^{3/4})} \wedge \left(-\frac{3}{4}\right)^{\frac{1}{4}} dz\\
		&= \frac{1}{3}\left(
		dy \wedge dx +  \frac{1 + \mathfrak{i} \sqrt{3}}{2}dx\wedge dz -  \frac{1 - \mathfrak{i} \sqrt{3}}{2}dy\wedge dz    \right. \\
		&\qquad \qquad  \left.- d\Big(x^{3} + y^{3} - z^{3} + \tilde{\beta} x + \frac{1 + \mathfrak{i} \sqrt{3}}{2}x + \tilde{\alpha} y 
		+ \frac{- 1 + \mathfrak{i} \sqrt{3}}{2}y\Big)\wedge dz\right)\\
		&= 	\frac{1}{3}\left(dy \wedge dx - d \Ham{Kek}{IV}(x,y,z;\tilde{\alpha},\tilde{\beta})\wedge dz\right),
	\end{align*}
	which gives us the needed rescaling coefficients, as well as the $x$ and $y$-dependent corrections to the Hamiltonian 
	$\Ham{Ok}{IV}(f(x,y),g(x,y);t(z);\kappa_{0}(\tilde{\alpha},\tilde{\beta}),\theta_{\infty}(\tilde{\alpha},\tilde{\beta}))$. The remaining terms 
	are purely time $z$-dependent and can be ignored. 
\end{proof}

\section{Preliminaries} 
\label{sec:prelim}

\subsection{The Okamoto space of initial conditions} 
\label{sub:Ok-soic-4}

\begin{notation*}
For the Okamoto system we use the following notation: coordinates $(f,g)$, parameters $\kappa_{0}$ and $\theta_{\infty}$; time variable $t$;
base points $q_{i}$, exceptional divisors $F_{i}$.	
\end{notation*}

The foundations of the geometric analysis of Painlev\'e equations were developed by K.~Okamoto \cite{Oka:1979:FAESOPCFPP}.
To make this paper self-contained, in this section we briefly explain how to construct the space of initial conditions 
for the Hamiltonian system \eqref{eq:Ok-Ham-4}. This will also allow us to introduce various notational conventions. This section is closely related
to \cite[Section 2.6]{KajNouYam:2017:GAPE} that we recommend for details.

Recall that the Painlev\'e property essentially requires that the general solution of an ODE has no movable (i.e., dependent of initial conditions)
singularities other than poles. If we think about parameterizing solutions via initial conditions at some time $t_{0}$, we then need to allow 
infinities as initial conditions, i.e., we need to change from $\mathbb{C}$ to $\mathbb{P}^{1}$. Thus, we consider the pair of dependent variables
$(f,g)$ as affine coordinates on the complex projective plane $\mathbb{P}^{1}\times \mathbb{P}^{1}$. We then introduce three more charts 
$(F,g)$, $(f,G)$, and $(F,G)$, where $F = 1/f$ and $G = 1/g$ are coordinates in the neighborhood of infinity, and via direct substitution 
we can easily rewrite our system in those charts:
\begin{alignat*}{2}
  	(f,G):&\quad\left\{
 	\begin{aligned}
 		\frac{df}{dt} &= - f^{2}  + \frac{4f}{G} - 2 t f - 2\kappa_{0},\\
		\frac{dG}{dt} &= \theta_{\infty} G^{2} - 2 f G  - 2 t G + 2,
 	\end{aligned}
 	\right. & \qquad\qquad
  	(F,G):&\quad\left\{
 	\begin{aligned}
 		\frac{dF}{dt} &=  2\kappa_{0} F^{2} - \frac{4F}{G} + 2 t F + 1,\\
		\frac{dG}{dt} &= \theta_{\infty}G^{2} - \frac{2G}{F} - 2 t G + 2, 
 	\end{aligned}
 	\right.\\[5pt]
  	(f,g):&\quad\left\{
 	\begin{aligned}
 		\frac{df}{dt} &= -f^{2} + 4 f g - 2 t f - 2\kappa_0,\\
		\frac{dg}{dt} &= -2 g^{2} + 2 f g + 2 t g - \theta_{\infty},			
 	\end{aligned}
 	\right. & \qquad\qquad
  	(F,g):&\quad\left\{
 	\begin{aligned}
 		\frac{dF}{dt} &= 2F^{2}\kappa_{0} - 4 F g + 2 t F + 1,\\
		\frac{dg}{dt} &= -2 g^{2} + \frac{2g}{F} + 2 t g - \theta_{\infty}.			
 	\end{aligned}
 	\right.	
\end{alignat*}
Consider, for example, our system in the $(F,g)$-chart. Solutions correspond to the flowlines of the vector field 
$\mathbf{V}(F,g) = (2F^{2}\kappa_{0} - 4 F g + 2 F t + 1)\partial_{F}+( -2 g^{2} + \frac{2g}{F} + 2 t g - \theta_{\infty})\partial_{g} + 
\partial_{t}$ that becomes undefined when $F=0$. Rescaling, 
$F \mathbf{V}(F,g) = (2F^{3}\kappa_{0} - 4 F^{2} g + 2 F^{2} t + F)\partial_{F} + 
( -2F g^{2} + 2g + 2 t F g - \theta_{\infty}F)\partial_{g} + F \partial_{t}$, we see that at the points 
$(F=0,g\neq0)$ the field becomes $2g \partial_{g}$, and so this flow is ``vertical''
(has a zero $\partial_{t}$-component) and hence such points do not parametrize solutions of the Painlev\'e equations 
(that are functions of $t$). Thus, we call 
the curve given by $F=0$ (in this chart) a \emph{vertical leaf} (or an \emph{inaccessible divisor}). However, at the point $(0,0)$ the rescaled field also vanishes (and we see the  indeterminacy $F/g=0/0$ in the original field $\mathbf{V}(F,g)$) and so a further adjustment is needed. 
This indeterminacy is resolved by the standard blowup procedure from algebraic geometry, see, e.g., \cite{Sha:2013:BAG1} for details.

In the two-dimensional case the blowup procedure is particularly simple and can be thought of as follows \cite{DzhFilSto:2020:RCDOPWHWDPE}.
\emph{Geometrically}, the blowup procedure ``separates'' the lines passing through the point $q_{i}$ (the \emph{center} of the blowup) 
by ``lifting'' them according to their ``slopes" (see the left picture on Figure~\ref{fig:blowup} for the local illustration of a blowup 
in the real-variable case). \emph{Topologically}, for complex surfaces, blowup is a surgery that creates a Riemann sphere ``bubble" (projectivized tangent
space) $S^{2}\simeq \mathbb{P}^{1}_{\mathbb{C}}$ in place of the center of the blowup $q_{i}$, thus adding a new spherical class to homology (and, via the 
Poincar\'e duality, cohomology) of the surface. \emph{Algebraically}, the blowup procedure is an introduction of two new charts $(u_{i},v_{i})$ and $(U_{i},V_{i})$ 
in the neighborhood of the blowup point $q_{i}(x_{i},y_{i})$, where the change of variables is given by 
$x = x_{i} + u_{i} = x_{i} + U_{i} V_{i}$ and $y = y_{i} + u_{i} v_{i} =y_{i} +  V_{i}$. 
This change of variables is a bijection away from $q_{i}$, but the point $q_{i}$ is replaced by the $\mathbb{P}^{1}$-line of all possible slopes, 
called the \emph{central fiber} or the \emph{exceptional divisor} of the blowup. We denote this central fiber by $F_{i}$ 
(and sometimes by $E_{i}$), it is given in the blowup charts by local equations $u_{i}=0$ and $V_{i} = 0$.
For these charts the upper/lower-case naming 
convention is only for convenience and, in contrast to the naming of affine charts, it does not hold that $U_{i} = 1/u_{i}$. 
However, it is true that $v_{i} = 1/U_{i}$ ---  these local coordinates on $\mathbb{P}^{1}$ represent all
possible slopes of lines passing through the point $q_{i}$, and so this variable change ``separates'' all curves passing through $q_{i}$ based on their
slopes.
\emph{Schematically}, it is convenient to illustrate the blowup on a
diagram as shown on the right on Figure~\ref{fig:blowup}. The notation $L-F_{i}$ denotes the \emph{proper transform} $\overline{\pi^{-1}(L-(x_{i},y_{i}))}$,
that needs to be distinguished from the \emph{total transform} $\pi^{-1}(L) = (L-F_{i}) + F_{i}$. Note that, despite the presence of the negative sign, $L-F_{i}$ is an actual geometric curve, i.e., an effective divisor. 

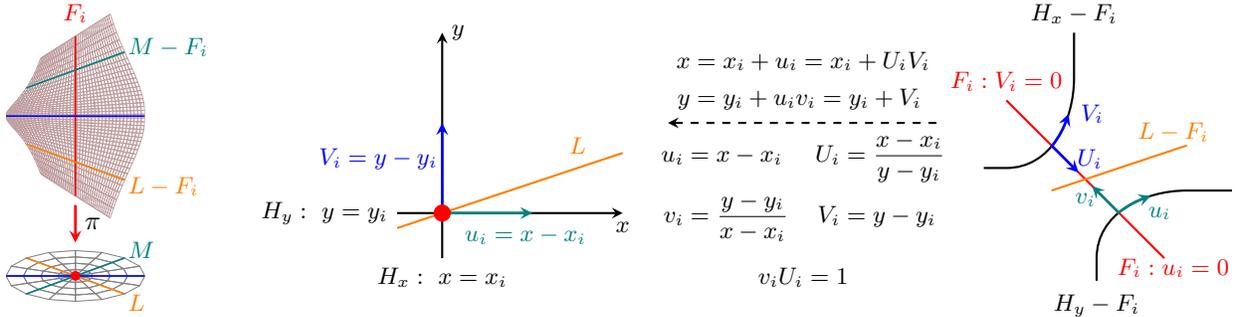
\begin{figure}[ht]
	 \begin{center}
		 \vskip-13pt
		\begin{tikzpicture}[>=stealth,scale=0.5]

		\begin{scope}[xshift=0cm,yshift=0cm]

		\begin{axis}[view={0}{10},axis lines=none, plot box ratio = 1 1 1, width=150pt, height=300pt]
		\addplot3[mesh,z buffer=sort, pink!50!gray, samples=40,domain=-1:1,y domain=-pi/3:pi/3] ({x * cos(deg(y))},{x * 	sin(deg(y))},{sqrt(3)*tan(deg(y))});
		\addplot3[mesh,z buffer=sort, gray, samples=9,domain=-1:1,y domain=0:pi] ({x * cos(deg(y))},{x * sin(deg(y))},-6);
		\addplot3[red,very thick] coordinates {(0,0,-3) (0,0,3)};
		\addplot3[orange,very thick] coordinates
		{({-1 * cos(deg(-pi/4))},{-1 * sin(deg(-pi/4))},{sqrt(3)*tan(deg(-pi/4))}) ({ cos(deg(-pi/4))},{ sin(deg(-pi/4))},{sqrt(3)*tan(deg(-pi/4))})};
		\addplot3[orange,very thick] coordinates
		{({-1 * cos(deg(-pi/4))},{-1 * sin(deg(-pi/4))},-6) ({ cos(deg(-pi/4))},{ sin(deg(-pi/4))},-6)};
		\addplot3[blue,very thick] coordinates
		{({-1 * cos(deg(0))},{-1 * sin(deg(0))},{sqrt(3)*tan(deg(0))}) ({ cos(deg(0))},{ sin(deg(0))},{sqrt(3)*tan(deg(0))})};
		\addplot3[blue,very thick] coordinates
		{({-1 * cos(deg(0))},{-1 * sin(deg(0))},-6) ({ cos(deg(0))},{ sin(deg(0))},-6)};
		\addplot3[teal,very thick] coordinates
		{({-1 * cos(deg(pi/4))},{-1 * sin(deg(pi/4))},{sqrt(3)*tan(deg(pi/4))}) ({ cos(deg(pi/4))},{ sin(deg(pi/4))},{sqrt(3)*tan(deg(pi/4))})};
		\addplot3[teal,very thick] coordinates
		{({-1 * cos(deg(pi/4))},{-1 * sin(deg(pi/4))},-6) ({ cos(deg(pi/4))},{ sin(deg(pi/4))},-6)};
		\end{axis}

		\node at (52.5pt,37.5pt) [circle, draw=red!100, fill=red!100, thick, inner sep=0pt,minimum size=1mm] {};
		\draw [red, line width = 1pt, ->] (52.5pt,3.2) -- (52.5pt,2.2);
		\draw (52.5pt,2.7) node [right] {$\pi$};
		\draw [red] (52.5pt,7.8) node [above] {\small $F_{i}$};
		\draw [orange] (3,3.7) node [right] {\small $L-F_{i}$};
		\draw [orange] (3,0.6) node [right] {\small $L$};
		\draw [teal] (3,7.4) node [right] {\small $M-F_{i}$};
		\draw [teal] (3,2) node [right] {\small $M$};

		\end{scope}
		\end{tikzpicture}
		\quad
		\begin{tikzpicture}[>=stealth,scale=0.6]
			\draw [black, line width = 0.8pt, ->] (-1,0) -- (4,0);
			\draw [teal, line width = 1pt, ->] (0,0) -- (2,0);
			\draw [black, line width = 0.8pt, ->] (0,-1) -- (0,4);
			\draw [blue, line width = 1pt, ->] (0,0) -- (0,2);
			\draw [orange, line width = 0.8pt] (-1,-1/3) -- (4,4/3);
			\node at (0,0) [circle, draw=red!100, fill=red!100, thick, inner sep=0pt,minimum size=2mm] {};
			\draw (0,-1) node[below] {\small $H_{x}:\ x=x_{i}$};
			\draw [teal] (0.3,-0.5) node[right] {\small $u_{i}=x-x_{i}$};
			\draw [blue] (0.1,1.2) node[left] {\small $V_{i}=y-y_{i}$};
			\draw (-1,0) node[left] {\small $H_{y}:\ y=y_{i}$};
			\draw (4,0) node[below] {\small $x$};
			\draw (0,4) node[right] {\small $y$};
			\node at (3,1.5) [orange] {\small $L$};
			\draw [black, line width = 0.8pt, dashed, <-] (5,2) -- (11,2);
			\draw (8,2) node [above] {\small $\begin{aligned} x &=x_{i} + u_{i} = x_{i} + U_{i} V_{i} \\ y & = y_{i} + u_{i}v_{i} = y_{i} + V_{i}\end{aligned}$};
			\draw (8,2) node [below] {\small $\begin{aligned} u_{i} &=x -x_{i} & U_{i}&=\frac{x-x_{i}}{y-y_{i}} \\ v_{i} & = \frac{y-y_{i}}{x-x_{i}} &
			V_{i} &= y-y_{i}\end{aligned} $};
			\draw (8,-1) node [below] {\small $v_{i} U_{i} = 1$};
			\begin{scope}[xshift=14cm]
			\draw [black, line width = 0.8pt, - ] (-2,1) -- (-1.5,1) to [out=0,in=-90]  	(0,3) -- (0,4);
			\draw [red, line width = 0.8pt] 	(-1.5,2.5) -- (2,-1);
			\draw [blue, line width = 1pt, ->] (-0.5,1.5) -- (0.1,0.9);
			\draw [blue, line width = 1pt, ->] (-0.51,1.46) to [out=52,in=-110] (-0.12,2.2);
			\draw [orange, line width = 0.8pt] (-0.5,0.5) -- (2.5,1.5);
			\begin{scope}[xshift=-0.5cm,yshift=+0.5cm]
			\draw [black, line width = 0.8pt, - ] (1,-2) -- (1,-1.5) to [out=90,in=-180]  	(3,0) -- (4,0);
			\draw [teal, line width = 1pt, ->] (1.5,-0.5) -- (0.9,0.1);
			\draw [teal, line width = 1pt, ->] (1.46,-0.51) to [out=38,in=-160] (2.2,-0.12);
			\end{scope}
			\draw [red] (-1.5,2.5) node[above] {\small $F_{i}:V_{i}=0$};
			\draw [red] (2.2,-0.8) node[below] {\small $F_{i}:u_{i}=0$};
			\draw (0,4) node[above] {\small $H_{x}-F_{i}$};
			\draw (0.5,-1.6) node[below] {\small $H_{y}-F_{i}$};
			\draw [orange] (2.2,1.8) node {\small $L-F_{i}$};
			\draw [teal] (1.9,0.0) node {\small $u_{i}$};
			\draw [teal] (0.25,0.3) node {\small $v_{i}$};
			\draw [blue] (0.4,2.2) node {\small $V_{i}$};
			\draw [blue] (0.35,1.2) node {\small $U_{i}$};
			\end{scope}
		\end{tikzpicture}
	\end{center}
	\caption{The Blowup Procedure}
	\label{fig:blowup}
\end{figure}

Thus, blowing up the point $q_{1}\left(F=0,g=0\right)$ amounts to introducing coordinate charts $(u_{1},v_{1})$ and $(U_{1},V_{1})$ via 
$F = u_{1} = U_{1}V_{1}$, $g=u_{1}v_{1} = V_{1}$. Since the change of variables is algebraic, it is easy to rewrite our system in those charts:
\begin{equation*}
	(u_{1}v_{1}):\quad \left\{
	\begin{aligned}
		\frac{d u_{1}}{dt} &= - 4 u_{1}^{2} v_{1} + 2 \kappa_{0} u_{1}^{2} + 2t u_{1} + 1,\\
		\frac{dv_{1}}{dt} &= 2 u_{1} v_{1}^{2} + \frac{v_{1} - \theta_{\infty}}{u_{1}}  - 2 \kappa_{0} u_{1} v_{1},
	\end{aligned}
	\right.\qquad 
	(U_{1},V_{1}):\quad\left\{
	\begin{aligned}
		\frac{dU_{1}}{dt} &= 2 \kappa_{0} U_{1}^{2} V_{1} + \frac{\theta_{\infty}U_{1}-1}{V_{1}} - 2 U_{1}V_{1},\\
		\frac{d V_{1}}{dt} &= - 2 V_{1}^{2}  + \frac{2}{U_{1}} + 2 t V_{1} - \theta_{\infty}.
	\end{aligned}
	\right.
\end{equation*}
Performing the same kind of analysis we see that there is another vertical leaf given by the equations $u_{1}=V_{1} = 0$ (that corresponds to 
the central fiber) that contains a new indeterminate point $q_{2}\left(u_{1} = 0, v_{1} = \theta_{\infty} \right)$ 
(equivalently, $U_{1} = 1/\theta_{\infty}$, $V_{1} = 0$); since all new base points appear on exceptional divisors, we 
omit the coordinates and write them, e.g., as $q_{2}(0,\theta_{\infty})$ or $q_{2}(1/\theta_{\infty},0)$. 
Introducing the new charts $(u_{2},v_{2})$ and $(U_{2}, V_{2})$ via 
$u_{1} = u_{2} = U_{2} V_{2}$, $v_{1} = \theta_{\infty} + u_{2} v_{2} = \theta_{\infty} + V_{2}$ and repeating this process we
see that there are no new base points in these charts and this cascade is resolved. 

\begin{figure}[ht]
	\begin{tikzpicture}[>=stealth,basept/.style={circle, draw=red!100, fill=red!100, thick, inner sep=0pt,minimum size=1.2mm}]
	\begin{scope}[xshift=0cm,yshift=0cm]
	\draw [black, line width = 1pt] (-0.2,0) -- (3.2,0)	node [pos=0,left] {\small $H_{g}$} node [pos=1,right] {\small $g=0$};
	\draw [black, line width = 1pt] (-0.2,2.5) -- (3.2,2.5) node [pos=0,left] {\small $H_{g}$} node [pos=1,right] {\small $g=\infty$};
	\draw [black, line width = 1pt] (0,-0.2) -- (0,2.7) node [pos=0,below] {\small $H_{f}$} node [pos=1,above] {\small $f=0$};
	\draw [black, line width = 1pt] (3,-0.2) -- (3,2.7) node [pos=0,below] {\small $H_{f}$} node [pos=1,above] {\small $f=\infty$};
	\node (p1) at (3,0) [basept,label={[xshift = -5pt, yshift=0pt] \small $q_{1}$}] {};
	\node (p2) at (3.5,0.5) [basept,label={[yshift=0pt] \small $q_{2}$}] {};
	\node (p3) at (0,2.5) [basept,label={[xshift = 8pt, yshift=-15pt] \small $q_{3}$}] {};
	\node (p4) at (-0.5,2) [basept,label={[yshift=-15pt] \small $q_{4}$}] {};
	\node (p5) at (3,2.5) [basept,label={[xshift = -5pt, yshift=-15pt] \small $q_{5}$}] {};
	\node (p6) at (3.7,2) [basept,label={[yshift=-15pt] \small $q_{6}$}] {};
	\node (p7) at (4.5,2) [basept,label={[yshift=-15pt] \small $q_{7}$}] {};
	\node (p8) at (5.3,2) [basept,label={[yshift=-15pt] \small $q_{8}$}] {};
	\draw [red, line width = 0.8pt, ->] (p2) -- (p1);
	\draw [red, line width = 0.8pt, ->] (p4) -- (p3);	
	\draw [red, line width = 0.8pt, ->] (p6) -- (p5);
	\draw [red, line width = 0.8pt, ->] (p7) -- (p6);
	\draw [red, line width = 0.8pt, ->] (p8) -- (p7);
	\end{scope}
	\draw [->] (8,1.5)--(6,1.5) node[pos=0.5, below] {$\operatorname{Bl}_{q_{1}\cdots q_{8}}$};
	\begin{scope}[xshift=9.5cm,yshift=0cm]
	\draw [blue, line width = 1pt] (-0.5,2.5) -- (3.5,2.5) node [pos=0,left] {\small $H_{g} - F_{3} - F_{5}$};
	\draw [blue, line width = 1pt] (-0.9,1.5) -- (0.3,2.7) node [pos=1,above] {\small $F_{3} - F_{4}$};		
	\draw [blue, line width = 1pt] (3.9,1.5) -- (2.7,2.7) node [pos=1,above] {\small $F_{5} - F_{6}$};		
	\draw [red, line width = 1pt] (-0.8,0) -- (3.5,0)	node [pos=0,left] {\small $H_{g}-F_{1}$};
	\draw [red, line width = 1pt] (-0.6,-0.2) -- (-0.6,2.1) node [pos=0,below] {\small $H_{f} - F_{5}$};
	\draw [red, line width = 1pt] (-0.4,2.3) -- (0.3,1.6)	node [pos=1,right] {\small $F_{4}$};
	\draw [red, line width = 1pt] (3.4,0.2) -- (2.7,0.9)	node [pos=1,left] {\small $F_{2}$};
	\draw [blue, line width = 1pt] (2.7,1.5) -- (3.9,2.7)	node [pos=1,right] {\small $F_{6}-F_{7}$};
	\draw [blue, line width = 1pt] (4.7,1.5) -- (3.5,2.7) node [pos=0,right] {\small $F_{7} - F_{8}$};		
	\draw [blue, line width = 1pt] (3.9,1) -- (2.7,-0.2) node [pos=1,below] {\small $F_{1} - F_{2}$};		
	\draw [blue, line width = 1pt] (3.6,0.4) -- (3.6,2.1) node [pos=0,right] {\small $H_{f} - F_{1} - F_{5}$};
	\draw [red, line width = 1pt] (4.4,1.5) -- (5.1,2.2) node [pos=1,right] {\small $F_{8}$};
	\end{scope}
	\end{tikzpicture}
	\caption{The Space of Initial Conditions for the $\Pain{IV}$ Okamoto Hamiltonian System \eqref{eq:Ok-Ham-4}}
	\label{fig:soic-Okamoto-4}
\end{figure}	

We now do the same procedure in the remaining charts to get the following cascades of base points: 
\begin{equation}\label{eq:basept-Okamoto-4}
\begin{aligned}
q_{1}(\infty,0) &\leftarrow q_{2}(0,\theta_{\infty}),\qquad q_{3}(0,\infty)\leftarrow q_{4}(\kappa_{0},0),\\
q_{5}(\infty,\infty) &\leftarrow q_{6}(0,2)\leftarrow q_{7}(0,-4t)
\leftarrow q_{8}(0,4(1 + 2 t^{2} + \theta_{\infty} - \kappa_{0})).	
\end{aligned}
\end{equation}

\begin{remark} Since the blowup points that appear in the cascades are on exceptional divisors, when we write 
	the coordinates of these points as $q_{i}(0,a)$, we work in the $(u,v)$ coordinate system, and points 
	$q_{i}(a,0)$ are in the $(U,V)$ coordinate system; for points $q_{i}(0,0)$ we would always specify which
	coordinate system is used.	
\end{remark}

The resulting configuration of the base points and vertical leaves is shown on 
Figure~\ref{fig:soic-Okamoto-4}. The \emph{Okamoto space of initial conditions} is then the resulting surface with the
configuration of vertical leaves removed. For the geometric analysis it is, however, more convenient to consider
the compact surface and just realize that points on vertical leaves do not correspond to initial conditions of 
our equation.

 Note that, modulo the coordinates of the base points, we get exactly the standard 
$E_{6}^{(1)}$ surface as given in \cite{KajNouYam:2017:GAPE}, see Figure~\ref{fig:surface-e61}. All of the 
standard geometric data, such as the choice of surface and root bases, for this surface is 
summarized in Section~\ref{sub:surface-e61} (note that we use the notation $F_{i}$ for exceptional 
divisors in the Okamoto case, and $E_{i}$ in the KNY case). However, the normalization 
of points $q_{6}(0,2)$ in the Okamoto case and $p_{6}(0,1)$ in the standard case is slightly different, 
which results in the need for some scaling adjustments in the variables. 
Specifically, the Hamiltonian system 
\begin{equation}\label{eq:KNY-Ham-4}
\left\{
\begin{aligned}
	\frac{dq}{d\tilde{t}} &= \frac{\partial \Ham{KNY}{IV}}{\partial p} = 2qp - q^{2} - \tilde{t} q  - a_{1},\\
	\frac{dp}{d\tilde{t}} &= -\frac{\partial \Ham{KNY}{IV}}{\partial q} = 2qp - p^{2} + \tilde{t}q -a_{2},			
\end{aligned}
\right.\qquad\text{where}\quad
\begin{aligned}
\Ham{KNY}{IV}(q,p;\tilde{t})&= qp^{2} - q^{2}p - \tilde{t}qp - a_{2}q - a_{1}p,\\	
\symp{KNY} &= dp \wedge dq,
\end{aligned}
\end{equation}
considered in \cite[Section 2.6]{KajNouYam:2017:GAPE} results in the following form of $\Pain{IV}$ equation:
\begin{equation*}
	\frac{d^{2} q}{d \tilde{t}^{2}} = \frac{1}{2q}\left(\frac{dq}{d\tilde{t}}\right)^{2} + \frac{3}{2} q^{3}
	+ 2 \tilde{t} q^{2} + \left(a_{2} - a_{0} + \frac{\tilde{t}^{2}}{2}\right)q - \frac{a_{1}^{2}}{2q}.
\end{equation*}
By the rescaling change of variables $\tilde{t} = \sqrt{2}t$, $q(\tilde{t}) = f(t(\tilde{t}))/\sqrt{2}$
this equation can be transformed to the standard form \eqref{eq:P4-std} with parameters 
$\alpha = 1 - a_{1} - 2 a_{2}$, $\beta = -2a_{1}^{2}$. The parameters $a_{1}$ and $a_{2}$ here are 
canonical geometric parameters called the \emph{root variables}, as explained in 
Section~\ref{ssub:period-map-a21}. For future computations, it is important to 
relate these parameters to the parameters $\kappa_{0}$ and $\theta_{\infty}$ of the Okamoto Hamiltonian.
Comparing
\begin{equation}\label{eq:rv-Ok2Std-4}
	\alpha = 1 - a_{1} - 2 a_{2} = 1 + 2 \theta_{\infty} - \kappa_{0},\qquad \beta = -2a_{1}^{2} = - 2 \kappa_{0}^{2}
\end{equation}
we immediately see that 
\begin{equation}\label{eq:rv-Ok2KNY-4}
	a_{1} = \kappa_{0},\qquad a_{2} = -\theta_{\infty},\qquad\text{and also,}\quad a_{0} = 1 - a_{2} - a_{2} = 1 + \theta_{\infty} - \kappa_{0}.
\end{equation}
Note that we could have also computed the root variables directly from the geometric data in the same way as 
outlined in Lemma~\ref{lem:period-map-a21} using the symplectic form $\symp{Ok} = dg\wedge df$.


\subsection{Geometry of the standard $E_{6}^{(1)}$ Surface} 
\label{sub:surface-e61}

In this section we describe the geometry of the model example of $E_{6}^{(1)}$ Sakai surface, following the standard reference 
\cite{KajNouYam:2017:GAPE}. This surface is the Okamoto space of initial conditions for $\Pain{IV}$. Its symmetry group is 
the extended affine Weyl group $\widetilde{W}\left(A_{2}^{(1)}\right)$. Here we describe the standard configuration of the blowup points, the choice of 
the surface and the symmetry root bases in the Picard lattice, and the birational representation of the symmetry group 
$\widetilde{W}\left(A_{2}^{(1)}\right)$. Recall that each nontrivial family of Sakai surfaces is obtained by blowing up 
$\mathbb{P}^{1} \times \mathbb{P}^{1}$ at $8$ points lying on the polar divisor of some symplectic form. The Picard lattice of the 
resulting rational algebraic surface $\mathcal{X}$ has rank $10$ and is 
$\operatorname{Pic}(\mathcal{X}) = \operatorname{Span}_{\mathbb{Z}}\{\mathcal{H}_{1}, \mathcal{H}_{2},\mathcal{E}_{1},\ldots,\mathcal{E}_{8}\}$,
where
$\mathcal{H}_{i}$ stand for the classes of coordinate divisors and $\mathcal{E}_{i}$ are the exceptional divisors, or the 
classes of central fibers of the blowup points. The anti-canonical
divisor class then is $-K_{\mathcal{X}} = 2 \mathcal{H}_{1} + 2 \mathcal{H}_{2} - \mathcal{E}_{1}-\cdots-\mathcal{E}_{8}$, and 
surfaces of different types correspond to different configurations of 
irreducible components of the anti-canonical divisor \cite{Sak:2001:RSAWARSGPE}.

\subsubsection{The point configuration} 
\label{ssub:point-conf-e61}
Consider the following decomposition of the anti-canonical divisor class $\delta = - \mathcal{K}_{\mathcal{X}}$ into classes 
$\delta_{i}$ (\emph{surface roots}) of the irreducible components $d_{i}$ of the anti-canonical divisor $-K_{\mathcal{X}}$:
\begin{align*}
	\delta &=  \delta_{0} + \delta_{1} + 2 \delta_{2} + 3 \delta_{3} + 2 \delta_{4} + \delta_{5} + 2 \delta_{6} \\
	&= (\mathcal{E}_{7} - \mathcal{E}_{6}) + (\mathcal{E}_{1} - \mathcal{E}_{2}) + 2 (\mathcal{H}_{q} - \mathcal{E}_{1} - \mathcal{E}_{5})
	+ 3 (\mathcal{E}_{5} - \mathcal{E}_{6}) + 2 (\mathcal{H}_{p} -  \mathcal{E}_{3} - \mathcal{E}_{5}) + 
	(\mathcal{E}_{3} - \mathcal{E}_{4}) + 2(\mathcal{E}_{6} - \mathcal{E}_{7}).
\end{align*}

The intersection configuration of those roots is given by the Dynkin diagram of type $E_{6}^{(1)}$, as shown on Figure~\ref{fig:d-roots-e61}.
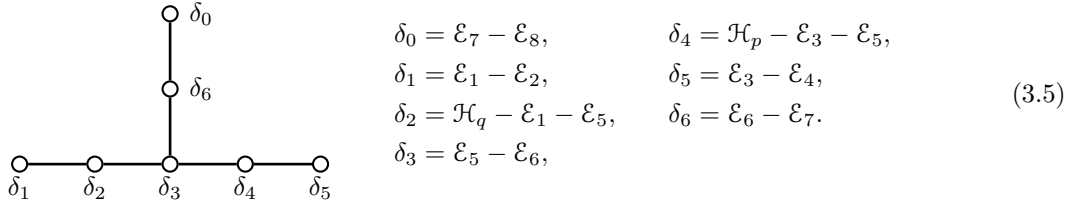
\begin{figure}[ht]
\begin{equation}\label{eq:d-roots-e61}			
	\raisebox{-40pt}{\begin{tikzpicture}[
			elt/.style={circle,draw=black!100,thick, inner sep=0pt,minimum size=2mm}]
		\path 	(-2,0) 	node 	(d1) [elt] {}
		        (-1,0) 	node 	(d2) [elt] {}
		        ( 0,0) 	node  	(d3) [elt] {}
		        ( 1,0) 	node  	(d4) [elt] {}
		        ( 2,0) 	node 	(d5) [elt] {}
		        ( 0,1)	node 	(d6) [elt] {}
		        ( 0,2)	node 	(d0) [elt] {};
		\draw [black,line width=1pt ] (d1) -- (d2) -- (d3) -- (d4) -- (d5)  (d3) -- (d6) -- (d0); 
			\node at ($(d1.south) + (0,-0.2)$) 	{$\delta_{1}$};
			\node at ($(d2.south) + (0,-0.2)$)  {$\delta_{2}$};
			\node at ($(d3.south) + (0,-0.2)$)  {$\delta_{3}$};
			\node at ($(d4.south) + (0,-0.2)$)  {$\delta_{4}$};		
			\node at ($(d5.south) + (0,-0.2)$)  {$\delta_{5}$};		
			\node at ($(d6.east) + (0.3,0)$) 	{$\delta_{6}$};		
			\node at ($(d0.east) + (0.3,0)$) 	{$\delta_{0}$};		
			\end{tikzpicture}} \qquad
			\begin{alignedat}{2}
			\delta_{0} &= \mathcal{E}_{7} - \mathcal{E}_{8}, &\qquad  \delta_{4} &= \mathcal{H}_{p} -  \mathcal{E}_{3} - \mathcal{E}_{5},\\
			\delta_{1} &= \mathcal{E}_{1} - \mathcal{E}_{2}, &\qquad  \delta_{5} &= \mathcal{E}_{3} - \mathcal{E}_{4},\\
			\delta_{2} &= \mathcal{H}_{q} - \mathcal{E}_{1} - \mathcal{E}_{5}, &\qquad  \delta_{6} &= \mathcal{E}_{6} - \mathcal{E}_{7}.\\
			\delta_{3} &= \mathcal{E}_{5} - \mathcal{E}_{6}, 			
			\end{alignedat}
\end{equation}
	\caption{The Surface Root Basis for the standard $E_{6}^{(1)}$ Sakai surface}
	\label{fig:d-roots-e61}	
\end{figure}

Consider the complex projective plane $\mathbb{P}^{1}\times \mathbb{P}^{1}$ covered by four coordinate
charts $(q,p)$, $(Q,p)$, $(q,P)$, and $(Q,P)$, where $Q = 1/q$ and $P=1/p$.  
Using the action of the gauge group $\mathbf{PGL}_{2}(\mathbb{C})\times \mathbf{PGL}_{2}(\mathbb{C})$ 
of M\"obius transformations, we can, without loss of generality, put the divisors $d_{2}$ and $d_{4}$, where $\delta_{i} = [d_{i}]$,
to be the lines at infinity, $d_{2} = V(Q) = \{q=\infty\}$ and $d_{4} = V(P) = \{p=\infty\}$ and put the base points $p_{1}$, $p_{3}$, and $p_{5}$
to be on the intersection of coordinate lines: $p_{1}(\infty,0)$, $p_{3}(0,\infty)$, and $p_{5}(\infty,\infty)$. We then get the blowup 
diagram for the standard example of Sakai surface of type $E_{6}^{(1)}$ as shown on Figure~\ref{fig:surface-e61}.

\begin{remark} Here we use the same coordinates $(q,p)$ as \cite{KajNouYam:2017:GAPE}, but these coordinates should not be confused with 
	the different coordinates $(q,p)$ for the Its-Prokhorov Hamiltonian system in  \eqref{eq:IP-Ham-4} and 	
	Section~\ref{sub:IP-soic-4}.
\end{remark}

\begin{figure}[ht]
	\begin{tikzpicture}[>=stealth,basept/.style={circle, draw=red!100, fill=red!100, thick, inner sep=0pt,minimum size=1.2mm}]
	\begin{scope}[xshift=0cm,yshift=0cm]
	\draw [black, line width = 1pt] (-0.2,0) -- (3.2,0)	node [pos=0,left] {\small $H_{p}$} node [pos=1,right] {\small $p=0$};
	\draw [black, line width = 1pt] (-0.2,2.5) -- (3.2,2.5) node [pos=0,left] {\small $H_{p}$} node [pos=1,right] {\small $p=\infty$};
	\draw [black, line width = 1pt] (0,-0.2) -- (0,2.7) node [pos=0,below] {\small $H_{q}$} node [pos=1,above] {\small $q=0$};
	\draw [black, line width = 1pt] (3,-0.2) -- (3,2.7) node [pos=0,below] {\small $H_{q}$} node [pos=1,above] {\small $q=\infty$};
	\node (p1) at (3,0) [basept,label={[xshift = -5pt, yshift=-15pt] \small $p_{1}$}] {};
	\node (p2) at (3.5,0.5) [basept,label={[yshift=0pt] \small $p_{2}$}] {};
	\node (p3) at (0,2.5) [basept,label={[xshift = 8pt, yshift=-15pt] \small $p_{3}$}] {};
	\node (p4) at (-0.5,2) [basept,label={[yshift=-15pt] \small $p_{4}$}] {};
	\node (p5) at (3,2.5) [basept,label={[xshift = -5pt, yshift=-15pt] \small $p_{5}$}] {};
	\node (p6) at (3.7,2) [basept,label={[yshift=-15pt] \small $p_{6}$}] {};
	\node (p7) at (4.5,2) [basept,label={[yshift=-15pt] \small $p_{7}$}] {};
	\node (p8) at (5.3,2) [basept,label={[yshift=-15pt] \small $p_{8}$}] {};
	\draw [line width = 0.8pt, ->] (p2) -- (p1);
	\draw [line width = 0.8pt, ->] (p4) -- (p3);	
	\draw [line width = 0.8pt, ->] (p6) -- (p5);
	\draw [line width = 0.8pt, ->] (p7) -- (p6);
	\draw [line width = 0.8pt, ->] (p8) -- (p7);
	\end{scope}
	\draw [->] (8,1.5)--(6,1.5) node[pos=0.5, below] {$\operatorname{Bl}_{p_{1}\cdots p_{8}}$};
	\begin{scope}[xshift=9.5cm,yshift=0cm]
	\draw [blue, line width = 1pt] (-0.5,2.5) -- (3.5,2.5) node [pos=0,left] {\small $H_{p} - E_{3} - E_{5}$};
	\draw [blue, line width = 1pt] (-0.9,1.5) -- (0.3,2.7) node [pos=1,above] {\small $E_{3} - E_{4}$};		
	\draw [blue, line width = 1pt] (3.9,1.5) -- (2.7,2.7) node [pos=1,above] {\small $E_{5} - E_{6}$};		
	\draw [red, line width = 1pt] (-0.8,0) -- (3.5,0)	node [pos=0,left] {\small $H_{p}-E_{1}$};
	\draw [red, line width = 1pt] (-0.6,-0.2) -- (-0.6,2.1) node [pos=0,below] {\small $H_{q} - E_{3}$};
	\draw [red, line width = 1pt] (-0.4,2.3) -- (0.3,1.6)	node [pos=1,right] {\small $E_{4}$};
	\draw [red, line width = 1pt] (3.4,0.2) -- (2.7,0.9)	node [pos=1,left] {\small $E_{2}$};
	\draw [blue, line width = 1pt] (2.7,1.5) -- (3.9,2.7)	node [pos=1,right] {\small $E_{6}-E_{7}$};
	\draw [blue, line width = 1pt] (4.7,1.5) -- (3.5,2.7) node [pos=0,right] {\small $E_{7} - E_{8}$};		
	\draw [blue, line width = 1pt] (3.9,1) -- (2.7,-0.2) node [pos=1,below] {\small $E_{1} - E_{2}$};		
	\draw [blue, line width = 1pt] (3.6,0.4) -- (3.6,2.1) node [pos=0,right] {\small $H_{q} - E_{1} - E_{5}$};
	\draw [red, line width = 1pt] (4.4,1.5) -- (5.1,2.2) node [pos=1,right] {\small $E_{8}$};
	\end{scope}
	\end{tikzpicture}
	\caption{The standard $E_{6}^{(1)}$ Sakai surface}
	\label{fig:surface-e61}
\end{figure}
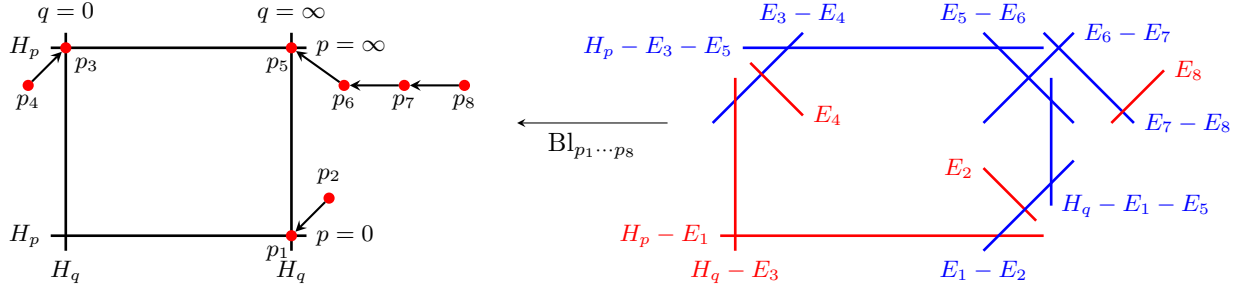	

This point configuration can be parameterized by five parameters $b_{2}, b_{4}, b_{6}, b_{7}$, and $b_{8}$ as follows:
\begin{equation}\label{eq:basept-e61}
\begin{aligned}
p_{1}(\infty,0) &\leftarrow p_{2}(\infty,0;qp=b_{2}),\qquad p_{3}(0,\infty)\leftarrow p_{4}(\infty,0;qp=b_{4}),\\
p_{5}(\infty,\infty) &\leftarrow p_{6}\left(\infty,\infty; \frac{q}{p}=b_{6}\right)\leftarrow
p_{7}\left(\infty,\infty; \frac{q}{p}=b_{6}; \frac{q(q-b_{6}p)}{p}=b_{7}\right)\\		
&\leftarrow p_{8}\left(\infty,\infty; \frac{q}{p}=b_{6}; \frac{q(q-b_{6}p)}{p}=b_{7}; 
\frac{q(q(q - b_{6}p)- b_{7}p)}{p}=b_{8}\right).	
\end{aligned}
\end{equation}
The residual two-parameter gauge group acts on these configurations via rescaling,
	\begin{equation}\label{eq:gauge-e61}
			\left(\begin{matrix}
				b_{2} & b_{4} & \\
				b_{6} & b_{7} & b_{8}
			\end{matrix}; \   \begin{matrix}
				f \\ g
			\end{matrix}\right) \sim \left(\begin{matrix}
			\lambda \mu b_{2} & \lambda \mu b_{4} \\
			\frac{\lambda}{\mu} b_{6} & \frac{\lambda^{2}}{\mu} b_{7} & \frac{\lambda^{3}}{\mu} b_{8}
			\end{matrix};\  \begin{matrix}
				\lambda f \\ \mu g
			\end{matrix}\right),\,\lambda,\mu\neq0,
	\end{equation}
and so the true number of parameters is three. A canonical choice of such parameters is known as the 
\emph{root variables}, as we explain next.

\subsubsection{The period map and the root variables} 
\label{ssub:period-map-a21}
To define the root variables we need to choose a root basis in the 
\emph{symmetry sub-lattice} $Q = \Pi(R^{\perp}) \triangleleft \operatorname{Pic}(\mathcal{X})$ 
and find a symplectic form $\omega$
whose polar divisor $-K_{\mathcal{X}}$ is the configuration of $-2$-curves shown on Figure~\ref{fig:surface-e61}. 
As usual, we take the same basis as in \cite{KajNouYam:2017:GAPE}, see Figure~\ref{fig:a-roots-a21}.
A symplectic form $\omega$ with $[(\omega)]\in \mathcal{K}_{\mathcal{X}}$ such that 
$-[(\omega)] = \delta_{0} + \delta_{1} + 2 \delta_{2} + 3\delta_{3} + 2\delta_{4} + \delta_{5} + 2 \delta_{6}$ 
can be written in the main affine $(q,p)$-chart as $\omega = k dq\wedge dp$.

\begin{figure}[ht]
\begin{equation}\label{eq:a-roots-a21}			
	\raisebox{-32.1pt}{\begin{tikzpicture}[
			elt/.style={circle,draw=black!100,thick, inner sep=0pt,minimum size=2mm}]
		\path 	(0,1.73) 	node 	(a0) [elt, label={[xshift=0pt, yshift = 0 pt] $\alpha_{0}$} ] {}
		        (-1,0) node 	(a1) [elt, label={[xshift=-10pt, yshift = -10 pt] $\alpha_{1}$} ] {}
		        ( 1,0) 	node  	(a2) [elt, label={[xshift=10pt, yshift = -10 pt] $\alpha_{2}$} ] {};
\		\draw [black,line width=1pt ] (a0) -- (a1) -- (a2) -- (a0);
	\end{tikzpicture}} \qquad
			\begin{aligned}
			\alpha_{0} &= \mathcal{H}_{q} + \mathcal{H}_{p} - \mathcal{E}_{5} - \mathcal{E}_{6} - \mathcal{E}_{7} - \mathcal{E}_{8}, \\
			\alpha_{1} &= \mathcal{H}_{q} - \mathcal{E}_{3} - \mathcal{E}_{4}, \\
			\alpha_{2} &= \mathcal{H}_{p} - \mathcal{E}_{1} - \mathcal{E}_{2},\\[5pt]
			\delta & = \mathrlap{\alpha_{0} + \alpha_{1} + \alpha_{2}.} 
			\end{aligned}
\end{equation}
	\caption{The Symmetry Root Basis for the standard d-$P\left(A_{2}^{(1)}\right)$ case}
	\label{fig:a-roots-a21}	
\end{figure}

The \emph{period map} $\chi: Q\to \mathbb{C}$ is first defined on the simple roots $\alpha_{i}$ 
and then extended to the full symmetry sub-lattice by linearity. The \emph{root variables} are the 
values of the period map on the roots $\alpha_{i}$, $a_{i}:= \chi(\alpha_{i})$ that can be 
computed as follows, see \cite{Sak:2001:RSAWARSGPE} and 
\cite{DzhFilSto:2020:RCDOPWHWDPE} for details.
	\begin{itemize}
		\item First, we represent $\alpha_{i}$ as a difference of two effective divisors, 
		$\alpha_{i} = [C_{i}^{1}] - [C_{1}^{0}]$;
		\item second, note that there exists a \emph{unique} component $d_{k}$ of $-K_{\mathcal{X}}$ such that 
		$d_{k}\bullet C_{i}^{1} = d_{k}\bullet C_{i}^{0} = 1$, put $P_{i} =  d_{k}\cap C_{i}^{0}$ and 
		$Q_{i} =  d_{k}\cap C_{i}^{1}$:
		
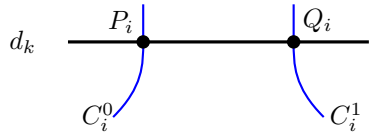
\begin{center}
			\begin{tikzpicture}[>=stealth, 
					elt/.style={circle,draw=black!100, fill=black!100, thick, inner sep=0pt,minimum size=1.5mm}]
					\draw[black, very thick] (0,0) -- (4,0);
					\draw[blue, thick] (1,0) -- (1,0.5);
					\draw[blue,thick] (1,0) .. controls (1,-0.3) and (1,-0.6) .. (0.6,-1);
					\draw[blue, thick] (3,0) -- (3,0.5);
					\draw[blue,thick] (3,0) .. controls (3,-0.3) and (3,-0.6) .. (3.4,-1);
					\node[style=elt] (P) at (1,0) {}; 		\node [above left] at (P) {$P_{i}$};
					\node[style=elt] (Q) at (3,0) {};		\node [above right] at (Q) {$Q_{i}$};
					\node at (-0.6,0) {$d_{k}$};
					\node at (0.4,-1) {$C_{i}^{0}$}; \node at (3.7,-1) {$C_{i}^{1}$};
					\end{tikzpicture}
		\end{center}
		\item then 
		\begin{equation*}
			\chi(\alpha_{i}) = \chi\left([C_{i}^{1}] - [C_{i}^{0}]\right) = 
			\int_{P_{i}}^{Q_{i}} \frac{ 1 }{ 2 \pi \mathfrak{i} }\oint_{d_{k}} \omega
			= \int_{P_{i}}^{Q_{i}} \operatorname{res}_{d_{k}} \omega,
		\end{equation*}
		where $\omega$ is our symplectic form.
	\end{itemize}

We then get the following result.

\begin{lemma}\label{lem:period-map-a21} 
	 The root variables $a_{i}$ are given by 
		\begin{equation}\label{eq:a2-root_vars}
			a_{0} = k\, \frac{b_{7}^{2} - b_{6} b_{8}}{b_{6}^{3}},\qquad 
			a_{1} = -k b_{4},\qquad 
			a_{2} = k b_{2}.
		\end{equation}
		It is convenient to put $k=-1$, so that the symplectic form is the standard symplectic form $\omega=dp \wedge dq$.	
		Using the gauge action \eqref{eq:gauge-e61} we can normalize $b_{6} = 1$. It is easy to see that under this gauge action 
		each root variable scales by $\lambda \mu$, so we can use the remaining scaling freedom to ensure the
		normalization condition $\chi(\delta) = a_{0} + a_{1} + a_{2} = 1$. Finally, in view of the relation of this example to $\Pain{IV}$,
		we denote $b_{7}:=-t$. This then gives us the point configuration in terms of root variables. Using the notation of \cite{KajNouYam:2017:GAPE}
		$p_{12}\left(\frac{1}{\varepsilon},-a_{2} \varepsilon\right)_{2}$, $p_{34}\left(a_{1} \varepsilon,\frac{1}{\varepsilon}\right)$,
		$p_{5678}\left(\frac{1}{\varepsilon},\frac{1}{\varepsilon} + t - a_{0}\varepsilon \right)_{4}$, we get exactly the the same 
		parameterization of the base points as in \cite[section 8.2.22]{KajNouYam:2017:GAPE}.
\end{lemma}


\subsubsection{The extended affine Weyl symmetry group} 
\label{ssub:Weyl-group-a21}

For completeness, we also include here the description of the birational representation of the extended
affine Weyl symmetry group $\widetilde{W}\left(A_{2}^{(1)}\right) = \operatorname{Aut}\left(A_{2}^{(1)}\right) \ltimes W\left(A_{2}^{(1)}\right)$.
The computations here follow the same steps as in \cite{KajNouYam:2017:GAPE} or in \cite{DzhFilSto:2020:RCDOPWHWDPE}
and are omitted, we only state the final result.

The affine Weyl group $W\left(A_{2}^{(1)}\right)$ is defined in terms of generators $w_{i} = w_{\alpha_{i}}$ and relations that 
are encoded by the affine Dynkin diagram $A_{2}^{(1)}$,
\begin{equation*}
	W\left(A_{2}^{(1)}\right) = W\left(\raisebox{-20pt}{\begin{tikzpicture}[
			elt/.style={circle,draw=black!100,thick, inner sep=0pt,minimum size=2mm},scale=0.5]
		\path 	(0,1.73) 	node 	(a0) [elt, label={[xshift=0pt, yshift = 0 pt] $\alpha_{0}$} ] {}
		        (-1,0) node 	(a1) [elt, label={[xshift=-10pt, yshift = -10 pt] $\alpha_{1}$} ] {}
		        ( 1,0) 	node  	(a2) [elt, label={[xshift=10pt, yshift = -10 pt] $\alpha_{2}$} ] {};
		\draw [black,line width=1pt ] (a0) -- (a1) -- (a2) -- (a0);
	\end{tikzpicture}} \right)
	=
	\left\langle w_{0},\dots, w_{2}\ \left|\ 
	\begin{alignedat}{2}
    w_{i}^{2} = e,\quad  w_{i}\circ w_{j} &= w_{j}\circ w_{i}& &\text{ when 
   				\raisebox{-0.08in}{\begin{tikzpicture}[
   							elt/.style={circle,draw=black!100,thick, inner sep=0pt,minimum size=1.5mm}]
   						\path   ( 0,0) 	node  	(ai) [elt] {}
   						        ( 0.5,0) 	node  	(aj) [elt] {};
   						\draw [black] (ai)  (aj);
   							\node at ($(ai.south) + (0,-0.2)$) 	{$\alpha_{i}$};
   							\node at ($(aj.south) + (0,-0.2)$)  {$\alpha_{j}$};
   							\end{tikzpicture}}}\\
    w_{i}\circ w_{j}\circ w_{i} &= w_{j}\circ w_{i}\circ w_{j}& &\text{ when 
   				\raisebox{-0.17in}{\begin{tikzpicture}[
   							elt/.style={circle,draw=black!100,thick, inner sep=0pt,minimum size=1.5mm}]
   						\path   ( 0,0) 	node  	(ai) [elt] {}
   						        ( 0.5,0) 	node  	(aj) [elt] {};
   						\draw [black] (ai) -- (aj);
   							\node at ($(ai.south) + (0,-0.2)$) 	{$\alpha_{i}$};
   							\node at ($(aj.south) + (0,-0.2)$)  {$\alpha_{j}$};
   							\end{tikzpicture}}}
	\end{alignedat}\right.\right\rangle. 
\end{equation*} 
The natural action of this group on $\operatorname{Pic}(\mathcal{X})$ is given by reflections in the 
roots $\alpha_{i}$, 
\begin{equation}\label{eq:root-refl}
	w_{i}(\mathcal{C}) = w_{\alpha_{i}}(\mathcal{C}) = \mathcal{C} - 2 
	\frac{\mathcal{C}\bullet \alpha_{i}}{\alpha_{i}\bullet \alpha_{i}}\alpha_{i}
	= \mathcal{C} + \left(\mathcal{C}\bullet \alpha_{i}\right) \alpha_{i},\qquad \mathcal{C}\in \operatorname{Pic(\mathcal{X})},
\end{equation}
which can be extended to an action on point configurations by elementary birational maps (which lifts to 
isomorphisms $w_{i}$ on the family of Sakai's surfaces),
this is known as a birational representation of $W\left(A_{2}^{(1)}\right)$.
The group of Dynkin diagram automorphisms $\operatorname{Aut}\left(D_{4}^{(1)}\right)\simeq \mathbb{D}_{3}$, where 
$\mathbb{D}_{3}$ is the usual dihedral group. Thus we only describe two transpositions $\sigma_{1}$, $\sigma_{2}$ that generate the whole group.

\begin{theorem}\label{thm:bir-weyl-a21}
	Reflections $w_{i}$ on $\operatorname{Pic}(\mathcal{X})$ are induced by the elementary 
	birational mappings
	\begin{align*}
		w_{0}:	\left(\begin{matrix} a_{0} & a_{1} \\ a_{2} & t \end{matrix}; \  \begin{matrix} q \\ p \end{matrix}\right) &\mapsto
		\left(\begin{matrix} -a_{0} & a_{0} + a_{1} \\ a_{0} + a_{2} & t \end{matrix}; \  
		\begin{matrix} q - \frac{a_{0}}{q-p+t} \\ p - \frac{a_{0}}{q-p+t} \end{matrix}\right),\\
		w_{1}:	\left(\begin{matrix} a_{0} & a_{1} \\ a_{2} & t \end{matrix}; \  \begin{matrix} q \\ p \end{matrix}\right) &\mapsto
		\left(\begin{matrix} a_{0} + a_{1} & - a_{1} \\ a_{1} + a_{2} & t \end{matrix}; \  
		\begin{matrix} q \\ p - \frac{a_{1}}{q} \end{matrix}\right),\\
		w_{2}:	\left(\begin{matrix} a_{0} & a_{1} \\ a_{2} & t \end{matrix}; \  \begin{matrix} q \\ p \end{matrix}\right) &\mapsto
		\left(\begin{matrix} a_{0} + a_{2} & a_{1} + a_{2} \\ - a_{2} & t \end{matrix}; \  
		\begin{matrix} q  + \frac{a_{2}}{p}\\ p  \end{matrix}\right).
	\end{align*}
	For the automorphisms $\sigma_{i}$ we choose the following two generators whose induced action on $\operatorname{Pic}(\mathcal{X})$ 
	can be represented as a composition of reflections in roots (but no longer \emph{symmetry roots}), and that act on 
	the surface and symmetry root bases as follows:
	\begin{align*}
		\sigma_{1}&=w_{\mathcal{H}_{q}-\mathcal{H}_{p}}w_{\mathcal{E}_{1} - \mathcal{E}_{3}}w_{\mathcal{E}_{2} - \mathcal{E}_{4}} \sim 
		(\alpha_{1}\alpha_{2})\sim (\delta_{1}\delta_{5})(\delta_{2}\delta_{4}),\\
		\sigma_{2}&=w_{\mathcal{H}_{q}-\mathcal{E}_{5} - \mathcal{E}_{6}}w_{\mathcal{E}_{1} - \mathcal{E}_{7}}w_{\mathcal{E}_{2} - \mathcal{E}_{8}} \sim 
		(\alpha_{0}\alpha_{2})\sim (\delta_{0}\delta_{1})(\delta_{2}\delta_{6}).
	\end{align*}
	The corresponding birational mappings then are
	\begin{align*}
		\sigma_{1}:	\left(\begin{matrix} a_{0} & a_{1} \\ a_{2} & t \end{matrix}; \  \begin{matrix} q \\ p \end{matrix}\right) &\mapsto
		\left(\begin{matrix} -a_{0}  & - a_{2} \\ - a_{1} & t \end{matrix}; \  
		\begin{matrix} -p \\ -q \end{matrix}\right),\\
		\sigma_{2}:	\left(\begin{matrix} a_{0} & a_{1} \\ a_{2} & t \end{matrix}; \  \begin{matrix} q \\ p \end{matrix}\right) &\mapsto
		\left(\begin{matrix} - a_{2} & -a_{1}  \\ - a_{0} & t \end{matrix}; \  
		\begin{matrix} q  \\ q - p + t  \end{matrix}\right).
	\end{align*}	
\end{theorem}


\subsection{Non-autonomous Hamiltonian systems, symplectic transformations and 2-forms} 
\label{sub:non-auto-Hams}
Suppose we have two Hamiltonian systems, with time-dependent Hamiltonians, and 
we found a birational change of coordinates that transforms one system into the other. In this section we address the question of
how, given this change of coordinates, to find the relationship between the Hamiltonians themselves. The difficulty here stems from the fact that
for time-dependent coordinate systems the above change of variables is in general also time-dependent. As a result,
the direct change of  coordinates in the Hamiltonian for one system does not, as a rule, give the Hamiltonian for another system -- some additional terms appear.
 To understand the nature of these terms we need some facts about non-autonomous Hamiltonian systems that are relevant to the 
 global Hamiltonian structures of Painlev\'e equations on Okamoto's spaces \cite{ShiTak:1997:SHSPS, MatMatTak:1999:SHSPS, Mat:1997:SHSPS}. 
 This point is explained in detail in \cite{DzhFilSto:2022:DERCSOPTRPEGA}, so
here we only briefly illustrate it, locally, 
using the Its-Prokhorov Hamiltonian system \eqref{eq:IP-Ham-4} and the Jimbo-Miwa Hamiltonian system \eqref{eq:JM-Ham-4},
since the change of coordinates \eqref{eq:JMtoIP-4} from  $(q,p) \mapsto (y,z)$ is $t$-dependent, and we can clearly see 
the additional correction terms.


Consider two copies of $\mathbb{C}^{2}$ with coordinates $(x,y)$ and $(X,Y)$ equipped with rational 
symplectic forms $\omega = F(x,y) dy \wedge dx$ and $\tilde{\omega}  = G(X,Y) dY \wedge dX$, i.e., 
$F(x,y)$ and $G(X,Y)$ are rational functions of their arguments. Suppose that we have a time-dependent 
birational change of variables $\varphi$ that we can consider as a birational transformation
on the extended phase space $\mathbb{C}^{3}$,
\begin{equation*}
\varphi : \mathbb{C}^3 \ni (x,y,t) \mapsto \left(X(x,y,t),Y(x,y,t),t \right) \in \mathbb{C}^3.
\end{equation*}
Then, for each \emph{fixed} $t$, we have a birational transformation $\varphi_{t}:(x,y)\mapsto \left( X(x,y ; t),Y (x,y ; t) \right)$ 
and we say that $\varphi_{t}$ is \emph{symplectic} with respect to $\omega$ and $\tilde{\omega}$ if, under $\varphi_t$, we have 
\begin{equation}
\omega = F(x,y) dy \wedge dx = \varphi_{t}^{*}(\tilde{\omega}) =  G(X(x,y;t),Y(x,y,;t)) d_{t}Y \wedge d_{t}X,
\end{equation}
where $\varphi_{t}^{*}$ is the usual pull-back map and $d_{t}$ indicates the exterior derivative on $\mathbb{C}^2$, 
so $t$ is treated as a constant in the calculation. We usually
omit the pull-back symbol and simply write $\omega = \tilde{\omega}$. Suppose that we now have two time-dependent
Hamiltonian functions  $H(x,y,t)$, $K(X,Y,t)$. Then we can define $2$-forms 
$\Omega = \omega - dH \wedge dt$ and $\tilde{\Omega} = \tilde{\omega} - dK \wedge dt$ on the 
two copies of the extended phase space $\mathbb{C}^{3}$. Then, if under the transformation $\varphi$ we have the equality of $2$-forms 
\begin{equation} \label{2forms}
\Omega = F(x,y) dy \wedge dx - dH \wedge dt =  G(X,Y) dY \wedge dX - dK \wedge dt = \tilde{\Omega},
\end{equation}
where $d$ here is the exterior derivative on the extended phase space $\mathbb{C}^{3}$ and so $t$ is treated as a variable in this calculation,
then the Hamiltonian system 
\begin{equation*}
	\left\{
	\begin{aligned}
		F(x,y) \frac{dx}{dt} &= \frac{\partial H}{\partial y}, \\ 
		F(x,y) \frac{dy}{dt} &= - \frac{\partial H}{\partial x},
	\end{aligned}
	\right. \qquad\text{is transformed into}\qquad
	\left\{
	\begin{aligned}
		G(X,Y) \frac{dX}{dt} &= \frac{\partial K}{\partial Y}, \\ 
		G(X,Y) \frac{dY}{dt} &= - \frac{\partial K}{\partial X}.
	\end{aligned}
	\right.	
\end{equation*}
This equality of $2$-forms dictates the correction between the Hamiltonians modulo purely $t$-dependent functions, with
\begin{equation} \label{eq:correction}
d(H - K)\wedge dt = G(X,Y) \left( \frac{\partial X}{\partial t} \frac{\partial Y}{\partial y} - \frac{\partial Y}{\partial t} \frac{\partial Y}{\partial x} \right) dy \wedge dt + G(X,Y) \left( \frac{\partial X}{\partial t} \frac{\partial Y}{\partial x} - \frac{\partial Y}{\partial t} \frac{\partial X}{\partial x} \right) dx \wedge dt.
\end{equation}
Returning now to our example, note that the change of variables \eqref{eq:JMtoIP-4} is symplectic with respect to the symplectic forms 
$\symp{IP}$ and $\symp{JM}$, 
and so the equality of $2$-forms
\begin{equation*}
\Symp{IP} = dp \wedge dq - d \Ham{IP}{IV} \wedge dt = \frac{1}{y} dy \wedge dz - d \Ham{JM}{IV} \wedge dt = \Symp{JM}
\end{equation*}
gives
\begin{equation}
d( \Ham{IP}{IV} - \Ham{JM}{IV}) \wedge dt =  - \frac{1}{2} dy \wedge dt = \frac{1}{2} dq \wedge dt.
\end{equation}
Thus we get
\begin{equation*}
\Ham{JM}{IV}(y,z;t) = \Ham{IP}{IV}\left( q(y,z,t), p(y,z,t) ; t \right) - \frac{y}{2} \quad \text{and} \quad 
\Ham{IP}{IV}(q,p;t) =  \Ham{JM}{IV}\left( y(q,p,t), z(q,p,t) ; t \right) + \frac{q}{2}.
\end{equation*}
\begin{remark}
	Sometimes after such changes of variables in the Hamiltonian there may be additional terms present depending only on $t$ that,
	as far as the dynamics is concerned, can be simply ignored.  
	It is important to note, however, that if a transformation $\varphi$ is such that $\tilde{\omega} = (\varphi_t^{-1})^*(\omega)$ has explicit dependence on $t$ in the coefficient $F$ or $G$, then the above discussion does not apply and a Hamiltonian structure of a differential system will in general not be preserved by the transformation. This is because, given $H(x,y,t)$, the pair of partial differential equations that must be solved to find $K(X,Y,t)$ satisfying the equality of two-forms \eqref{2forms} will in general not be compatible, and the transformed system for $X$, $Y$ will not be Hamiltonian with respect to $G(X,Y,t) d_tX \wedge d_tY$.	
\end{remark}


\section{Reducing Hamiltonian Systems to the Canonical Form} 
\label{sec:P4-matching}
In this chapter we construct spaces of initial conditions for each of the Hamiltonian systems described in the Introduction.
We then follow the identification procedure, described below, to match the spaces of initial conditions to some reference case,
that we take to be the Okamoto Hamiltonian system \eqref{eq:Ok-Ham-4}, both on the level of the Picard lattice, and on the level of 
a birational change of coordinates.
\subsection{The identification procedure} 
\label{sub:ident-proc}
\begin{enumerate}[(Step 1)]
	\item \textbf{Construct the space of initial conditions for the system}. We carefully review such construction in Section~\ref{sub:Ok-soic-4}.
	We need to remark here that this step is potentially more involved than in the discrete case \cite{DzhFilSto:2020:RCDOPWHWDPE}, 
	since indeterminacies of the vector field do not necessarily constitute singularities which need to be resolved through blowups.	
	In particular one must identify when a singularity is inaccessible and should be removed along with the vertical leaves rather 
	than blown up, since there is the possibility that a system is indeterminate at a point, but it does not require blowing up 
	since there is not a family of local solutions passing through, but rather a single one. In such cases more detailed inspection 
	of solutions is required along the lines of classical Painlevé analysis, but we do not encounter such examples here --- 
	see \cite{DzhFilSto:2022:DERCSOPTRPEGA} for an example where this is required. 
	%
	%
	
	\item \textbf{Determine the surface type, according to Sakai's classification scheme.} Recall that almost all surfaces in the Sakai classification 
	scheme are obtained as blowups of $\mathbb{P}^{1}\times \mathbb{P}^{1}$ at eight base points (the only 
	exception of the $E^{(1)}_{8}$ surface corresponding to $\Pain{I}$ that is obtained as a blowup of $\mathbb{P}^{2}$ at nine points, and we do 
	not consider that special case). 
	These base points lie on a biquadratic curve on $\mathbb{P}^{1}\times \mathbb{P}^{1}$ that is a \emph{polar divisor} of the symplectic form
	$\omega$ defining the Hamiltonian structure. The surface type of the system is determined by the configuration of the irreducible components of 
	this curve. Each such component should have self-intersection index $-2$ and 
	is associated with a node of an \emph{affine Dynkin diagram}; nodes are connected when the corresponding components intersect. 
	The type of this Dynkin diagram is called the \emph{surface type} of the equation. 
	This description assumes that the surfaces $\mathcal{X}_{n}$ are \emph{minimal}, and this may not be the case in general. In that case 
	some of the $-1$ curves, that should be inaccessible divisors for the system, need to be blown down, as in example in 
	Sections~\ref{sub:IP-soic-4} and \ref{sub:Kek-soic-4}.
	\item \textbf{Find a preliminary change of basis of $\operatorname{Pic}(\mathcal{X})$.} 
	At this step, we only need to ensure that this change of basis 
	identifies the \emph{surface roots} (or nodes of the Dynkin diagrams of our surface) with the standard example. 
	\item \textbf{Adjust this change of basis using root variables and parameter matching.} 
	The matching obtained in the previous step is far from unique, since we have the whole extended affine Weyl group of 
	B\"acklund transformations acting on the Picard lattice preserving the symmetry sub-lattice. On the level of the equations, 
	B\"acklund transformations act on the parameters of the equation. It may happen that we already know how the parameters should match,
	e.g., by finding a reduction of our Hamiltonian system to the standard Painlev\'e equation. This parameter matching 
	would fix the B\"acklund transformation ambiguity, but it may require adjusting the initial change of basis. 
	For that, we compute the \emph{canonical parameters} for the surface,
	known as the \emph{root variables}, using the symmetry roots that we get from the  preliminary change of basis in (Step 3), in terms of 
	parameters of the system, and compare it with the one obtained from the root variables of the reference surface using that 
	a priori parameter matching. These two sets of root variables would either match or
	differ by some symmetry. The importance of the root variables is that they are compatible with the action of the symmetry group, so we can just read off the 
	expression of that symmetry in terms of the generators of the group. 
	Then we need to act by that symmetry transformation on our preliminary choice of basis to ensure
	that the final change of basis will result in matching Hamiltonian systems on the level of parameters as well. We illustrate this 
	step in the next section, see \eqref{eq:IP-rv-prelim}.
	\item \textbf{Find the change of variables reducing the applied problem to the standard example.} 
	At this point we have the identification between the two surfaces on the level of the basis change on the Picard lattice, and so  
	next we need to find the actual change of variables that induces
	that linear change of basis. For that, identify the curves that form the basis for the corresponding coordinate pencils. Those curves then are our 
	projective coordinates, up to a M\"obius transformation. To fix the M\"obius transformations, use the mapping of coordinate divisors. An important part
	of this computation is the identification of the parameters between the two systems.
\end{enumerate}

\subsection{The Its-Prokhorov Hamiltonian system} 
\label{sub:IP-soic-4}

\begin{notation*}
For the Its-Prokhorov system we use the following notation: coordinates $(q,p)$, parameters $\Theta_{0}$ and $\Theta_{\infty}$; time variable $s$;
base points $y_{i}$, exceptional divisors $K_{i}$.	
\end{notation*}

In this section we construct the space of initial conditions for the Hamiltonian system \eqref{eq:IP-Ham-4}. Note that, if we only 
look at the Hamiltonians, we can not see the correspondence between the Okamoto parameters $\{\theta_{\infty}, \kappa_{0}\}$ 
in \eqref{eq:IP-Ham-4} and the Its-Prokhorov parameters $\{\Theta_{0}, \Theta_{\infty}\}$ in \eqref{eq:IP-Ham-4}. However, since both 
equations reduce to the same standard form \eqref{eq:P4-std}  of $\Pain{IV}$ with
\begin{equation*}
	\alpha = 2 \Theta_{\infty} - 1 = 1 + 2 \theta_{\infty} - \kappa_{0},\qquad \beta = - 8 \Theta_{0}^{2} = - 2 \kappa_{0}^{2},
\end{equation*}
we see that the relationship between parameters is given by 
\begin{equation}\label{eq:pars-Ok2IP-4}
	\Theta_{0}^{2} = \frac{\kappa_{0}^{2}}{4},\qquad \Theta_{\infty} = \theta_{\infty} - \frac{\kappa_{0}}{2} + 1.
\end{equation}
The time variable $t$ does not have to be the same either, so for now we denote it by $s$, but from $\Pain{IV}$ we know that we can take $s=t$. 

The space of initial conditions for this system is constructed in the same way as in Section~\ref{sub:Ok-soic-4}.
However, this time we get \emph{ten} base points (and so the space of initial conditions is not minimal), 
\begin{equation}\label{eq:IP-pts}
	\begin{tikzpicture}
	\node (y1) at (0,0) {$y_{1}(0,\infty)$}; 
	\node (y2) at (2,0.7) {$y_{2}(-\Theta_{0},0)$}; \node (y3) at (2,-0.7) {$y_{3}(\Theta_{0},0)$}; 
	\draw[->] (y2)--(1,0)--(y1);	\draw[->] (y3)--(1,0)--(y1); 
	\node (y4) at (5,0) {$y_{4}(\infty,\infty)$};
	\node (y5) at (7,0.7) {$y_{5}(0,-4)$}; \node (y8) at (7,-0.7) {$y_{8}(0,4)$};
	\draw[->] (y5)--(6.5,0)--(y4); \draw[->] (y8)--(6.5,0)--(y4);
	\node (y6) at (9.5,0.7) {$y_{6}(0,8s)$}; \node (y9) at (9.5,-0.7) {$y_{9}(0,-8s)$};
	\draw[->] (y6)--(y5); \draw[->] (y9)--(y8);
	\node (y7) at (12.5,0.7) {$y_{7}(0,16(1-s^{2} - \Theta_{\infty}))$}; \node (y10) at (12.7,-0.7) {$y_{10}(0,16(s^{2} + \Theta_{\infty}))$};
	\draw[->] (y7)--(y6); \draw[->] (y10)--(y9);
	\end{tikzpicture}
\end{equation}
whose configuration is shown on Figure~\ref{fig:IP-soic-4} (left). After blowing them up and denoting the exceptional divisors
for the blowup points $y_{i}$ by $K_{i}$, we get the configuration of 
vertical leaves that includes two curves with self-intersection index $-3$, $K_{1} - K_{2} - K_{3}$ and $K_{4} - K_{5} - K_{8}$.
These points lie on the polar divisor of a symplectic form $\omega = k dq\wedge dp$.

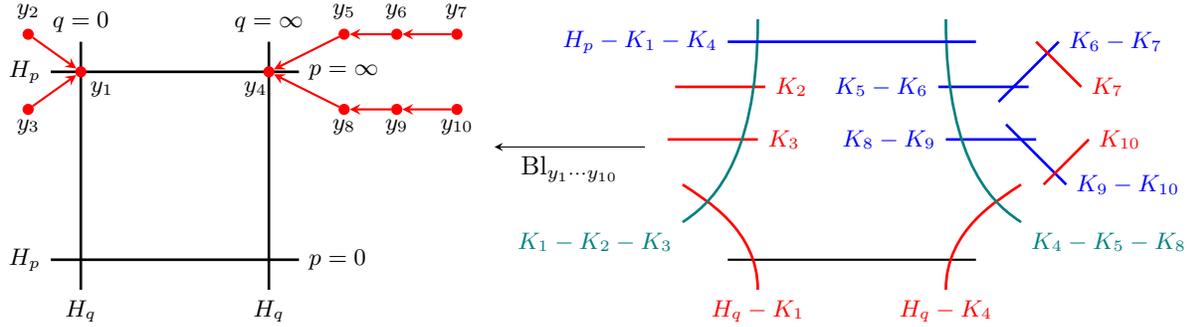
\begin{figure}[ht]
	\begin{tikzpicture}[>=stealth,basept/.style={circle, draw=red!100, fill=red!100, thick, inner sep=0pt,minimum size=1.2mm}]
	\begin{scope}[xshift=0cm,yshift=0cm]
	\draw [black, line width = 1pt] (-0.4,0) -- (2.9,0)	node [pos=0,left] {\small $H_{p}$} node [pos=1,right] {\small $p=0$};
	\draw [black, line width = 1pt] (-0.4,2.5) -- (2.9,2.5) node [pos=0,left] {\small $H_{p}$} node [pos=1,right] {\small $p=\infty$};
	\draw [black, line width = 1pt] (0,-0.4) -- (0,2.9) node [pos=0,below] {\small $H_{q}$} node [pos=1,above] {\small $q=0$};
	\draw [black, line width = 1pt] (2.5,-0.4) -- (2.5,2.9) node [pos=0,below] {\small $H_{q}$} node [pos=1,above] {\small $q=\infty$};
	\node (p1) at (0,2.5) [basept,label={[xshift = 8pt, yshift=-15pt] \small $y_{1}$}] {};
	\node (p2) at (-0.7,3) [basept,label={[yshift=0pt] \small $y_{2}$}] {};
	\node (p3) at (-0.7,2) [basept,label={[yshift=-15pt] \small $y_{3}$}] {};
	\node (p4) at (2.5,2.5) [basept,label={[xshift = -5pt, yshift=-15pt] \small $y_{4}$}] {};
	\node (p5) at (3.5,3) [basept,label={[yshift=0pt] \small $y_{5}$}] {};
	\node (p6) at (4.2,3) [basept,label={[yshift=0pt] \small $y_{6}$}] {};
	\node (p7) at (5.0,3) [basept,label={[yshift=0pt] \small $y_{7}$}] {};
	\node (p8) at (3.5,2) [basept,label={[yshift=-15pt] \small $y_{8}$}] {};
	\node (p9) at (4.2,2) [basept,label={[yshift=-15pt] \small $y_{9}$}] {};
	\node (p10) at (5.0,2) [basept,label={[yshift=-15pt] \small $y_{10}$}] {};
	\draw [red, line width = 0.8pt, ->] (p2) -- (p1);
	\draw [red, line width = 0.8pt, ->] (p3) -- (p1);
	\draw [red, line width = 0.8pt, ->] (p5) -- (p4);	
	\draw [red, line width = 0.8pt, ->] (p6) -- (p5);
	\draw [red, line width = 0.8pt, ->] (p7) -- (p6);
	\draw [red, line width = 0.8pt, ->] (p8) -- (p4);	
	\draw [red, line width = 0.8pt, ->] (p9) -- (p8);
	\draw [red, line width = 0.8pt, ->] (p10) -- (p9);
	\end{scope}
	\draw [->] (7.5,1.5)--(5.5,1.5) node[pos=0.5, below] {$\operatorname{Bl}_{y_{1}\cdots y_{10}}$};
	\begin{scope}[xshift=9cm,yshift=0cm]
	\draw [black, line width = 0.8pt] (-0.4,0) -- (2.9,0)	node [pos=0,left] {};	
	\draw [red, line width = 1pt] 		(-1,1) .. controls (-0.7,0.8) and (0,0.4) .. (0,-0.4) node[below] {\small $H_{q} - K_{1}$};
	\draw [red, line width = 1pt] 		(3.5,1) .. controls (3.2,0.8) and (2.5,0.4) .. (2.5,-0.4) node[below] {\small $H_{q} - K_{4}$};
	\draw [red, line width = 1pt] 		(-1.1,2.3) -- (0.1,2.3) node[right] {\small $K_{2}$};
	\draw [blue, line width = 1pt] 		(2.4,2.3) -- (3.6,2.3) node[left,pos=0] {\small $K_{5}-K_{6}$};
	\draw [blue, line width = 1pt] 		(3.2,2.1) -- (4.0,2.9) node[right,pos=1] {\small $K_{6}-K_{7}$};
	\draw [red, line width = 1pt] 		(3.7,2.9) -- (4.3,2.3) node[right,pos=1] {\small $K_{7}$};	
	\draw [red, line width = 1pt] 		(-1.2,1.6) -- (0.0,1.6) node[right] {\small $K_{3}$};
	\draw [blue, line width = 1pt] 		(2.5,1.6) -- (3.7,1.6) node[left,pos=0] {\small $K_{8}-K_{9}$};
	\draw [blue, line width = 1pt] 		(3.3,1.8) -- (4.1,1) node[right,pos=1] {\small $K_{9}-K_{10}$};
	\draw [red, line width = 1pt] 		(3.8,1.0) -- (4.4,1.6) node[right,pos=1] {\small $K_{10}$};
	\draw [teal, line width = 1pt] 	(-1,0.5) .. controls (-0.7,0.7) and (0,1.2) .. (0,3.2) node[pos=0,below left] {\small $K_{1} - K_{2} - K_{3}$};
	\draw [teal, line width = 1pt] 	(3.5,0.5) .. controls (3.2,0.7) and (2.5,1.2) .. (2.5,3.2) node[pos=0,below right] {\small $K_{4} - K_{5} - K_{8}$};
	\draw [blue, line width = 1pt] (-0.4,2.9) -- (2.9,2.9) node [pos=0,left] {\small $H_{p} - K_{1} - K_{4}$};
	\end{scope}
	\end{tikzpicture}
	\caption{The Space of Initial Conditions for the Its-Prokhorov Hamiltonian System \eqref{eq:IP-Ham-4}}
	\label{fig:IP-soic-4}
\end{figure}

To match this space with the space of initial conditions for the Okamoto system shown on Figure~\ref{fig:soic-Okamoto-4}, we need to blow down 
two $-1$ curves, and we should choose them among the curves intersecting with above $-3$ curves to reduce the index; to avoid going back
we should not choose the exceptional divisors. Thus, we choose the curves $H_{q}-K_{1}$ and $H_{q}-K_{4}$. Note that these curves are also 
\emph{inaccessible divisors}, and so blowing them down does not change the space of initial conditions. It is also more convenient  
instead of blowing down those two curves to blow up two points on the Okamoto surface. Looking at the intersection 
diagram \eqref{eq:d-roots-e61}	and  Figure~\ref{fig:soic-Okamoto-4}, we see that one point should be on the curve 
$d_{3}=F_{5} - F_{6}$ corresponding to the central node of the Dynkin diagram $E_{6}^{(1)}$, and the other one should be taken on 
a divisor corresponding to one of the boundary nodes. At this point we do not have enough information to make the right choice, so we just choose one of them, say 
$d_{1} = F_{1} - F_{2}$. Using the \emph{root variables} we can later see whether this choice is correct and make adjustments, if necessary.
The same applies to the matching of the remaining two legs of the $E_{6}^{(1)}$ diagram. Thus, we want to find an identification of 
two bases of the Picard lattice of two non-minimal surfaces
\begin{equation*}
	\operatorname{Pic}^{\mathrm{IP}}(\mathcal{X}) = \operatorname{Span}_{\mathbb{Z}}\{\mathcal{H}_{q},\mathcal{H}_{p},\mathcal{K}_{1},\ldots
	\mathcal{K}_{10}\}\simeq
	\operatorname{Pic}^{\mathrm{Ok}}_{\mathrm{ext}}(\mathcal{X}) = \operatorname{Span}_{\mathbb{Z}}\{\mathcal{H}_{f},\mathcal{H}_{g},\mathcal{F}_{1},\ldots
	\mathcal{F}_{10}\},
\end{equation*}
that matches the irreducible components of the anti-canonical divisor as follows:
\begin{equation}\label{eq:geom-IP-Ok-prelim-4}
	\begin{aligned}
		\delta_{0} &= \mathcal{F}_{7} - \mathcal{F}_{8} = \mathcal{K}_{9} - \mathcal{K}_{10}, &\qquad  
			\delta_{4} &= \mathcal{H}_{g} - \mathcal{F}_{3} - \mathcal{F}_{5} = \mathcal{K}_{3} - \mathcal{K}_{6}, \\
		\delta_{1} &= \mathcal{F}_{1} - \mathcal{F}_{2} - \mathcal{F}_{9} = \mathcal{K}_{1} - \mathcal{K}_{2} - \mathcal{K}_{3}, &\qquad  
			\delta_{5} &= \mathcal{F}_{3} - \mathcal{F}_{4} = \mathcal{K}_{6} - \mathcal{K}_{7}, \\
		\delta_{2} &= \mathcal{H}_{f} - \mathcal{F}_{1} - \mathcal{F}_{5} = \mathcal{H}_{p} - \mathcal{K}_{1} - \mathcal{K}_{4}, &\qquad  
			\delta_{6} &= \mathcal{F}_{6} - \mathcal{F}_{7} = \mathcal{K}_{8} - \mathcal{K}_{9}. \\
		\delta_{3} &= \mathcal{F}_{5} - \mathcal{F}_{6} - \mathcal{F}_{10} = \mathcal{K}_{4} - \mathcal{K}_{5} - \mathcal{K}_{8}, &\qquad  
	\end{aligned}
\end{equation}
We summarize such preliminary identification in the following Lemma.

\begin{lemma} A preliminary change of basis resulting in the identification \eqref{eq:geom-IP-Ok-prelim-4} is given by
	\begin{equation}\label{eq:basis-IP-Ok-prelim-4}
		\begin{aligned}
			\mathcal{H}_{f} &= \mathcal{H}_{p} + 2 \mathcal{H}_{q} - \mathcal{K}_{1} - \mathcal{K}_{3} - \mathcal{K}_{4} - \mathcal{K}_{5}, &\qquad 
				\mathcal{H}_{q} & = \mathcal{H}_{g},\\
			\mathcal{H}_{g} & = \mathcal{H}_{q}, &\qquad 	
				\mathcal{H}_{p} &= \mathcal{H}_{f} + 2\mathcal{H}_{g} - \mathcal{F}_{1} - \mathcal{F}_{5} - \mathcal{F}_{9} - \mathcal{F}_{10}, \\
			\mathcal{F}_{1}	&= \mathcal{H}_{q} - \mathcal{K}_{3}, &\qquad 
				\mathcal{K}_{1} &= \mathcal{H}_{g} - \mathcal{F}_{9},\\ 
			\mathcal{F}_{2}	&= \mathcal{K}_{2}, &\qquad 
				\mathcal{K}_{2} &= \mathcal{F}_{2},\\ 
			\mathcal{F}_{3}	&= \mathcal{K}_{6}, &\qquad 
				\mathcal{K}_{3} &= \mathcal{H}_{g} - \mathcal{F}_{1},\\ 
			\mathcal{F}_{4}	&= \mathcal{K}_{7}, &\qquad 
				\mathcal{K}_{4} &= \mathcal{H}_{g} - \mathcal{F}_{10},\\ 
			\mathcal{F}_{5}	&= \mathcal{H}_{q} - \mathcal{K}_{5}, &\qquad 
				\mathcal{K}_{5} &= \mathcal{H}_{g} - \mathcal{F}_{5},\\ 
			\mathcal{F}_{6}	&= \mathcal{K}_{8}, &\qquad 
				\mathcal{K}_{6} &= \mathcal{F}_{3},\\ 
			\mathcal{F}_{7}	&= \mathcal{K}_{9}, &\qquad 
				\mathcal{K}_{7} &= \mathcal{F}_{4},\\ 
			\mathcal{F}_{8}	&= \mathcal{K}_{10}, &\qquad 
				\mathcal{K}_{8} &= \mathcal{F}_{6},\\ 
			\mathcal{F}_{9}	&= \mathcal{H}_{q} - \mathcal{K}_{1}, &\qquad 
				\mathcal{K}_{9} &= \mathcal{F}_{7},\\ 
			\mathcal{F}_{10}	&= \mathcal{H}_{q} - \mathcal{K}_{4}, &\qquad 
				\mathcal{K}_{10} &=\mathcal{F}_{8}.
		\end{aligned}
	\end{equation}
	The corresponding symmetry roots then become 
	\begin{equation*}
		\begin{aligned}
			\alpha_{0} &= 2\mathcal{H}_{q} + \mathcal{H}_{p} - \mathcal{K}_{1} - \mathcal{K}_{3} - \mathcal{K}_{4}
				 - \mathcal{K}_{8} - \mathcal{K}_{9} - \mathcal{K}_{10},\\
			\alpha_{1} &= 2\mathcal{H}_{q} + \mathcal{H}_{p} - \mathcal{K}_{1} - \mathcal{K}_{3} - \mathcal{K}_{4}
				 - \mathcal{K}_{5} - \mathcal{K}_{6} - \mathcal{K}_{7},	\\
			\alpha_{2} &= \mathcal{K}_{3} - \mathcal{K}_{2},				 					
		\end{aligned}
	\end{equation*}
	and using the \emph{Period Map} for the symplectic form $\omega = k dq \wedge dp$ 
	we get the following root variables:
	\begin{equation*}
		a_{0}=-k(\Theta_{0} + \Theta_{\infty}),\quad a_{1} = -k(1 + \Theta_{0} - \Theta_{\infty}),\quad a_{2} = 2 k \Theta_{0}.
	\end{equation*}
	Imposing the normalization condition $a_{0} + a_{1} + a_{2} = 1$ gives $k=-1$, and so we get back the standard symplectic form
	$\symp{IP} = dp\wedge dq$.
\end{lemma}

Now we can use the root variables to check whether we got the correct bases identification. We know the correspondence 
between the Okamoto parameters $\kappa_{0}$ and $\theta_{\infty}$, Painlev\'e parameters $\alpha$, $\beta$, and root variables $a_{i}$ from 
\eqref{eq:rv-Ok2Std-4} and \eqref{eq:rv-Ok2KNY-4}. We also 
have the correspondence between the Its-Prokhorov parameters $\Theta_{0}$ and $\Theta_{\infty}$ 
and Painlev\'e parameters $\alpha$, $\beta$ from \eqref{eq:IP-pars-4}, and between  $\Theta_{0}$, $\Theta_{\infty}$ and 
$\kappa_{0}$ and $\theta_{\infty}$ from \eqref{eq:pars-Ok2IP-4}. In particular, we have $\Theta_{0}^{2} = \frac{\kappa_{0}^{2}}{4}$, and so we 
need to choose the correct sign, which turns out to be a bit delicate, since the Its-Prokhorov system only depends on $\Theta^{2}$. However, 
the same parameters are also used in the Jimbo-Miwa system \eqref{eq:JM-Ham-4} considered in the next section, and from that parameter matching
we see that we need to take $\kappa_{0} = 2 \Theta_{0}$. 
Then we get
\begin{equation}\label{eq:IP-rv-prelim}
	\begin{aligned}
			a_{0}^{\mathrm{IP}}&= \Theta_{0} + \Theta_{\infty}  = 1 + \theta_{\infty}  
				 = a_{0}^{\mathrm{Ok}} + a_{1}^{\mathrm{Ok}},\\
			a_{1}^{\mathrm{IP}}&= 1+ \Theta_{0} - \Theta_{\infty} = \kappa_{0} - \theta_{\infty}
				=  a_{1}^{\mathrm{Ok}}+a_{2}^\mathrm{Ok},\\
			a_{2}^{\mathrm{IP}}&= -2 \Theta_{0} = - \kappa_{0} = -a_{1}^\mathrm{Ok},			
	\end{aligned}	
\end{equation}
where we chose $\sqrt{\kappa_{0}^{2}} = - \kappa_{0}$. Thus, we see that our root variables (and hence the bases identification) 
differ by a composition of a reflection $w_{2}$ and an automorphism $\sigma_{1}$ described in Theorem~\ref{thm:bir-weyl-a21}. 

\begin{remark} We are off by a sign in the action of $\sigma_{1}$ on $a_{i}$, but that's due to some normalization choices and can be ignored
	at this point. 
\end{remark}

Acting by $\sigma_{1}\circ w_{2}=w_{\mathcal{H}_{f}-\mathcal{H}_{g}}w_{\mathcal{F}_{1} - \mathcal{F}_{3}}w_{\mathcal{F}_{2} - \mathcal{F}_{4}}
w_{\mathcal{H}_{g} - \mathcal{F}_{1} - \mathcal{F}_{2}}$ 
on the identification \eqref{eq:basis-IP-Ok-prelim-4} we arrive at the final bases identification.

\begin{lemma}\label{lem:IP-to-Ok-4} The change of bases for Picard lattices between the Its-Prokhorov and the Okamoto 
	(with two additional blowup points) surfaces is given by 
	\begin{equation}\label{eq:basis-IP-Ok}
		\begin{aligned}
			\mathcal{H}_{f} &= \mathcal{H}_{q}, &\qquad 
				\mathcal{H}_{q} & = \mathcal{H}_{f},\\
			\mathcal{H}_{g} & = 2 \mathcal{H}_{q} + \mathcal{H}_{p} - \mathcal{K}_{1} - \mathcal{K}_{2} - \mathcal{K}_{4} - \mathcal{K}_{5},  &\qquad 	
				\mathcal{H}_{p} &= 2\mathcal{H}_{f} + \mathcal{H}_{g} - \mathcal{F}_{3} - \mathcal{F}_{5} - \mathcal{F}_{9} - \mathcal{F}_{10}, \\
			\mathcal{F}_{1}	&= \mathcal{K}_{6}, &\qquad 
				\mathcal{K}_{1} &= \mathcal{H}_{f} - \mathcal{F}_{9},\\ 
			\mathcal{F}_{2}	&= \mathcal{K}_{7}, &\qquad 
				\mathcal{K}_{2} &= \mathcal{H}_{f} - \mathcal{F}_{3},\\ 
			\mathcal{F}_{3}	&= \mathcal{H}_{q} - \mathcal{K}_{2}, &\qquad 
				\mathcal{K}_{3} &= \mathcal{F}_{4},\\ 
			\mathcal{F}_{4}	&= \mathcal{K}_{3}, &\qquad 
				\mathcal{K}_{4} &= \mathcal{H}_{f} - \mathcal{F}_{10},\\ 
			\mathcal{F}_{5}	&= \mathcal{H}_{q} - \mathcal{K}_{5}, &\qquad 
				\mathcal{K}_{5} &= \mathcal{H}_{f} - \mathcal{F}_{5},\\ 
			\mathcal{F}_{6}	&= \mathcal{K}_{8}, &\qquad 
				\mathcal{K}_{6} &= \mathcal{F}_{1},\\ 
			\mathcal{F}_{7}	&= \mathcal{K}_{9}, &\qquad 
				\mathcal{K}_{7} &= \mathcal{F}_{2},\\ 
			\mathcal{F}_{8}	&= \mathcal{K}_{10}, &\qquad 
				\mathcal{K}_{8} &= \mathcal{F}_{6},\\ 
			\mathcal{F}_{9}	&= \mathcal{H}_{q} - \mathcal{K}_{1}, &\qquad 
				\mathcal{K}_{9} &= \mathcal{F}_{7},\\ 
			\mathcal{F}_{10}	&= \mathcal{H}_{q} - \mathcal{K}_{4}, &\qquad 
				\mathcal{K}_{10} &=\mathcal{F}_{8}.
		\end{aligned}
	\end{equation}
	This results in the following correspondences between the surface roots (note that we need to move an additional blowup point from 
	the divisor $d_{1}$ to the divisor $d_{5}=F_{3} - F_{4}$),
	\begin{equation}\label{eq:geom-IP-Ok}
	\begin{aligned}
		\delta_{0} &= \mathcal{F}_{7} - \mathcal{F}_{8} = \mathcal{K}_{9} - \mathcal{K}_{10}, &\qquad  
			\delta_{4} &= \mathcal{H}_{g} - \mathcal{F}_{3} - \mathcal{F}_{5} = \mathcal{H}_{p} - \mathcal{K}_{1} - \mathcal{K}_{4}, \\
		\delta_{1} &= \mathcal{F}_{1} - \mathcal{F}_{2} = \mathcal{K}_{6} - \mathcal{K}_{7}, &\qquad  
			\delta_{5} &= \mathcal{F}_{3} - \mathcal{F}_{4} - \mathcal{F}_{9} = \mathcal{K}_{1} - \mathcal{K}_{2} - \mathcal{K}_{3}, \\
		\delta_{2} &= \mathcal{H}_{f} - \mathcal{F}_{1} - \mathcal{F}_{5} = \mathcal{K}_{5} - \mathcal{K}_{6}, &\qquad  
			\delta_{6} &= \mathcal{F}_{6} - \mathcal{F}_{7} = \mathcal{K}_{8} - \mathcal{K}_{9}; \\
		\delta_{3} &= \mathcal{F}_{5} - \mathcal{F}_{6} - \mathcal{F}_{10} = \mathcal{K}_{4} - \mathcal{K}_{5} - \mathcal{K}_{8}, &\qquad  
	\end{aligned}
	\end{equation}
	and the symmetry roots,
	\begin{equation}\label{eq:sym-IP-Ok-4}
	\begin{aligned}
			\alpha_{0} &= \mathcal{H}_{f} + \mathcal{H}_{g} - \mathcal{F}_{5} - \mathcal{F}_{6} - \mathcal{F}_{7} - \mathcal{F}_{8}
			= 2 \mathcal{H}_{q} + \mathcal{H}_{p} - \mathcal{K}_{1} - \mathcal{K}_{2} - \mathcal{K}_{4} - \mathcal{K}_{8} - \mathcal{K}_{9} - \mathcal{K}_{10}, \\
			\alpha_{1} &= \mathcal{H}_{f} - \mathcal{F}_{3} - \mathcal{F}_{4} =  \mathcal{K}_{2} - \mathcal{K}_{3}, \\
			\alpha_{2} &= \mathcal{H}_{g} - \mathcal{F}_{1} - \mathcal{F}_{2} = 
			2 \mathcal{H}_{q} + \mathcal{H}_{p} - \mathcal{K}_{1} - \mathcal{K}_{2} - \mathcal{K}_{4} - \mathcal{K}_{5} - \mathcal{K}_{6} - \mathcal{K}_{7}.
		\end{aligned}
	\end{equation}
	The symplectic form is the standard one, $\symp{IP} = dp\wedge dq$, and the root variables match, 
	\begin{equation*}
		a_{0}^{\mathrm{IP}} = \Theta_{\infty} - \Theta_{0} = 1 + \theta_{\infty} - \kappa_{0} = a_{0}^{\mathrm{Ok}},\quad 
		a_{1}^{\mathrm{IP}} =  2 \Theta_{0} = 2 \kappa_{0} = a_{1}^\mathrm{Ok},\quad 
		a_{2}^{\mathrm{IP}} = 1 - \Theta_{0} - \Theta_{\infty} = - \theta_{\infty} = a_{2}^\mathrm{Ok}.				
	\end{equation*}
\end{lemma}

It remains to find the birational change of variables that corresponds to this change of bases of the Picard lattice. It is given in the following Lemma.

\begin{lemma}\label{lem:coords-IP-Ok-4} The change of coordinates and parameter matching between the Its-Prokhorov and Okamoto Hamiltonian systems is given by
    \begin{equation*}
   	 \left\{\begin{aligned}
   	 	q(f,g,t)&=f,\\
   		p(f,g,t)&= g - \frac{f}{4} - \frac{t}{2} - \frac{\kappa_{0}}{2f},\\
		\Theta_{0}&=\frac{\kappa_{0}}{2},\\  
		\Theta_{\infty} &= 1 + \theta_{\infty} - \frac{\kappa_{0}}{2},
   	 \end{aligned}\right.
    \qquad\text{and conversely} \qquad 
    	\left\{\begin{aligned}
   	 	f(q,p,t)&=q,\\
   		g(q,p,t)&=p + \frac{q}{4} + \frac{t}{2} + \frac{\Theta_{0}}{q},\\
		\kappa_{0} &= 2\Theta_{0},\\  
		\theta_{\infty} &= -1 + \Theta_{0} + \Theta_{\infty}.
    	\end{aligned}\right.
    \end{equation*}
\end{lemma}

\begin{proof} The proof here is by now a standard computation, see detailed examples in \cite{DzhTak:2018:SASGTDPE,DzhFilSto:2020:RCDOPWHWDPE}, but to 
	make this paper self-contained, we briefly outline the argument here as well. From the change of basis we see that the coordinate classes are  
	\begin{equation*}
		\mathcal{H}_{f} = \mathcal{H}_{q},\qquad 
		\mathcal{H}_{g} = 2 \mathcal{H}_{q} + \mathcal{H}_{p} - \mathcal{K}_{1} - \mathcal{K}_{2} - \mathcal{K}_{4} - \mathcal{K}_{5}.
	\end{equation*}
	Thus, up to a M\"obius transformation, the coordinates $f$ and $q$ coinside, 
	\begin{equation*}
		f = \frac{A q + B}{C q + D},
	\end{equation*}
	where $A,\ldots D$ are the parameters of the M\"obius transformation, defined up to a common multiple, that we still need to find. The $g$-coordinate
	is more interesting --- it is a (projective) coordinate on a pencil of $(2,1)$-curves passing through the points $y_{1}$, $y_{2}$, $y_{4}$, and $y_{5}$.
	A generic $(2,1)$-curve, written in the affine $(q,p)$ chart, has the equation 
	\begin{align*}
		a_{21}q^{2}p + a_{20}q^{2} + a_{11} qp + a_{10}q + a_{01}p + a_{00} &= 0,\\
		\intertext{and imposing the conditions given by points $y_{1}$, $y_{2}$, $y_{4}$, $y_{5}$ reduces this equation to}
		a_{20}(q^{2} + 4 qp + 4 \Theta_{0}) + a_{10}q = 0.
	\end{align*}
	Thus, equations $q=0$ and $q^{2} + 4 qp + 4 \Theta_{0}=0$ define two basis curves in this pencil, and so 
	\begin{equation*}
		g(q,p) = \frac{K q + L  (q^{2} + 4 qp + 4 \Theta_{0})}{M q + N  (q^{2} + 4 qp + 4 \Theta_{0})},
	\end{equation*}
	where $K,L,M,N$ are again some paramaters to be determined. Let now $\varphi: (q,p)\to (f,g)$ be our change of variables, and consider the 
	induced forward mapping on the (unique irreducible) divisors corresponding to classes of $-3$, $-2$, and $-1$ curves. For example, 
	$\varphi_{*}(\mathcal{H}_{p}-\mathcal{K}_{1}-\mathcal{K}_{4}) = \mathcal{H}_{g}-\mathcal{F}_{3}-\mathcal{F}_{5}$ means that the 
	$-2$ curve $H_{p} - K_{1} - K_{4}$, whose projection down to $\mathbb{P}^{1} \times \mathbb{P}^{1}$ is  given in the $(q,P)$-chart by 
	the equation $P=0$ (and parameterized by $q$) should map on the $-2$ curve $H_{g}-F_{3}-F_{5}$ whose projection in the $(f,G)$-chart
	is given by the equation $G=0$ (and parameterized by $g$). Thus,
	\begin{equation*}
		(f,G)(q,P=0) = \left(\frac{A q + B}{Cq + D},\frac{P(Mq + N (q^{2}  + 4 \Theta_{0})) + 4 N q}{P(Kq + L(q^{2} + 4 \Theta_{0})) + 4 Lq}\right){\Big|_{P=0}}
		= \left(\frac{A q + B}{Cq + D},\frac{N}{L}\right)
	\end{equation*}
	implies that $N=0$ (and hence $M\neq 0$ and we can take $M=1$). So  $g(q,p) = K + L \left(q + 4 p + \frac{4 \Theta_{0} }{q}\right)$. Similarly,	
	$\varphi_{*}(\mathcal{K}_{9} - \mathcal{K}_{10}) = \mathcal{F}_{7} - \mathcal{F}_{8}$ means that the $-2$-curve $K_{9}-K_{10}$, given in the 
	$(u_{9},v_{9})$ chart in the domain by the equation $u_{9}=0$ and parameterized by $v_{9}$, should collapse on the $q_{5}\leftarrow q_{6}\leftarrow q_{7}$
	cascade. This results in $C=0$, $D=1$, $A = 4L$, $B = 2(K+ 2 L s - t)$ and so $f(q,p) = 2(K + 2L(q + s) - t)$. Similar computations for 
	$\varphi_{*}(\mathcal{K}_{6}-\mathcal{K}_{7}) = \mathcal{F}_{1}-\mathcal{F}_{2}$ imply that $K = 2 L s$, 
	and $\varphi_{*}(\mathcal{K}_{1} - \mathcal{K}_{2} - \mathcal{K}_{3}) = \mathcal{F}_{3} - \mathcal{F}_{4} - \mathcal{F}_{9}$ gives $L = t/(4s)$. Finally,
	$\varphi_{*}(\mathcal{K}_{10}) = \mathcal{K}_{8}$ gives $s = \pm t$, as expected, and we take $s = t$ (we also expect this from the Painlev\'e equation)
	to get the final change of variables.	The inverse change of variables is then immediate.
\end{proof}


\subsection{The Jimbo-Miwa Hamiltonian system} 
\label{sub:JM-soic-4}

\begin{notation*}
For the Jimbo-Miwa system we use the following notation: coordinates $(y,z)$, parameters $\Theta_{0}$ and $\Theta_{\infty}$ (same as in 
Section~\ref{sub:IP-soic-4}); time variable $s$;
base points $x_{i}$, exceptional divisors $M_{i}$.	
\end{notation*}

Consider now the system \eqref{eq:JM-Ham-4}. It has eight base points, 
\begin{align*}
&x_{1}(0,0),\quad x_{2}(0,2\Theta_{0}),\quad x_{3}(\infty,\Theta_{0} + \Theta_{\infty}),\quad \\
&x_{4}(\infty,\infty)\leftarrow x_{5}(u_{4}=0,v_{4}=0)\leftarrow x_{6}(0,2)\leftarrow x_{7}(0,-4s)
\leftarrow x_{8}(0,4(-1 + 2s^{2}- \Theta_{0} + \Theta_{\infty}))	
\end{align*}
that lie on the polar divisor of a symplectic form $\omega = k \frac{dy\wedge dz}{y}$. 
Blowing them up we obtain the space of initial conditions for the Jimbo-Miwa system shown on Figure~\ref{fig:soic-JM}.
Proceeding in the same way as in Section~\ref{sub:JM-soic-4}, but omitting the preliminary 
basis identification, we get the following result.

\begin{figure}[ht]
	\begin{tikzpicture}[>=stealth,basept/.style={circle, draw=red!100, fill=red!100, thick, inner sep=0pt,minimum size=1.2mm}]
	\begin{scope}[xshift=0cm,yshift=0cm]
	\draw [black, line width = 1pt] (-0.4,0) -- (2.9,0)	node [pos=0,left] {\small $H_{z}$} node [pos=1,right] {\small $z=0$};
	\draw [black, line width = 1pt] (-0.4,2.5) -- (2.9,2.5) node [pos=0,left] {\small $H_{z}$} node [pos=1,right] {\small $z=\infty$};
	\draw [black, line width = 1pt] (0,-0.4) -- (0,2.9) node [pos=0,below] {\small $H_{y}$} node [pos=1,above] {\small $y=0$};
	\draw [black, line width = 1pt] (2.5,-0.4) -- (2.5,2.9) node [pos=0,below] {\small $H_{y}$} node [pos=1,above] {\small $y=\infty$};
	\node (p1) at (0,0) [basept,label={[xshift = 8pt, yshift=-15pt] \small $x_{1}$}] {};
	\node (p2) at (0,1.7) [basept,label={[xshift = -8pt,yshift=-10pt] \small $x_{2}$}] {};
	\node (p3) at (2.5,1) [basept,label={[xshift = 8pt,yshift=-10pt] \small $x_{3}$}] {};
	\node (p4) at (2.5,2.5) [basept,label={[xshift = -5pt, yshift=-15pt] \small $x_{4}$}] {};
	\node (p5) at (3.5,3) [basept,label={[yshift=0pt] \small $x_{5}$}] {};
	\node (p6) at (4.2,3) [basept,label={[yshift=0pt] \small $x_{6}$}] {};
	\node (p7) at (4.9,3) [basept,label={[yshift=0pt] \small $x_{7}$}] {};
	\node (p8) at (5.6,3) [basept,label={[yshift=0pt] \small $x_{8}$}] {};
	\draw [red, line width = 0.8pt, ->] (p5) -- (2.8,2.5) -- (p4);	
	\draw [red, line width = 0.8pt, ->] (p6) -- (p5);
	\draw [red, line width = 0.8pt, ->] (p7) -- (p6);
	\draw [red, line width = 0.8pt, ->] (p8) -- (p7);	
	\end{scope}
	\draw [->] (7,1.5)--(5,1.5) node[pos=0.5, below] {$\operatorname{Bl}_{x_{1}\cdots x_{8}}$};
	\begin{scope}[xshift=8.5cm,yshift=0cm]
	\draw [blue, line width = 1pt] (-0.6,0.2) -- (-0.6,2.7) node [pos=1,above] {\small $H_{y} - M_{1}-M_{2}$};
	\draw [red, line width = 1pt] (-0.9,1) -- (0.3,-0.2)	node [pos=1,below right] {\small $M_{1}$};
	\draw [red, line width = 1pt] (-0.4,0) -- (3.5,0)	node [pos=0,below left] {\small $H_{z}-M_{1}$};
	\draw [blue, line width = 1pt] (3.9,1) -- (2.7,-0.2) node [pos=1,below] {\small $H_{y} - M_{3} - M_{4}$};	
	\draw [red, line width = 1pt] (3.4,0.2) -- (2.7,0.9)	node [pos=1,left] {\small $M_{3}$};		
	\draw [red, line width = 1pt] (-0.8,2.1) -- (-0.1,1.4)	node [pos=1,right] {\small $M_{2}$};
	\draw [blue, line width = 1pt] (3.6,0.4) -- (3.6,2.1) node [pos=0,right] {\small $M_{4} - M_{5}$};
	\draw [blue, line width = 1pt] (3.9,1.5) -- (2.7,2.7) node [pos=1,above] {\small $M_{5} - M_{6}$};		
	\draw [blue, line width = 1pt] (-0.8,2.5) -- (3.5,2.5) node [pos=0,left] {\small $H_{x} - M_{4} - M_{5}$};
	\draw [blue, line width = 1pt] (3,1.7) -- (4.1,2.9)	node [pos=1,right] {\small $M_{6}-M_{7}$};
	\draw [blue, line width = 1pt] (5,1.7) -- (3.8,2.9) node [pos=0,right] {\small $M_{7} - M_{8}$};		
	\draw [red, line width = 1pt] (4.7,1.7) -- (5.4,2.4) node [pos=1,right] {\small $M_{8}$};
	\end{scope}
	\end{tikzpicture}
	\caption{The Space of Initial Conditions for the Jimbo-Miwa Hamiltonian System \eqref{eq:JM-Ham-4}}
	\label{fig:JM-soic-4}
\end{figure}
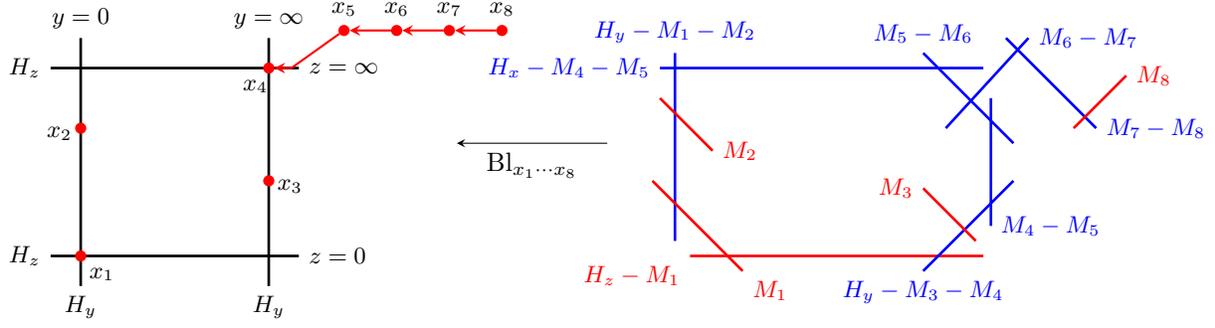

\begin{lemma}\label{lem:JMtoOk-basis-4} The change of bases for Picard lattices between the Jimbo-Miwa and the Okamoto surfaces 
	is given by 
	\begin{equation}\label{eq:basis-JM-Ok-4}
		\begin{aligned}
			\mathcal{H}_{f} &= \mathcal{H}_{y}, &\qquad 
				\mathcal{H}_{y} & = \mathcal{H}_{f},\\
			\mathcal{H}_{g} & = 2 \mathcal{H}_{y} + \mathcal{H}_{z} - \mathcal{M}_{2} - \mathcal{M}_{4} - \mathcal{M}_{5} - \mathcal{M}_{6},  &\qquad 	
				\mathcal{H}_{z} &= 2\mathcal{H}_{f} + \mathcal{H}_{g} - \mathcal{F}_{3} - \mathcal{F}_{5} - \mathcal{F}_{6} - \mathcal{F}_{7}, \\
			\mathcal{F}_{1}	&= \mathcal{M}_{7}, &\qquad 
				\mathcal{M}_{1} &= \mathcal{F}_{4},\\ 
			\mathcal{F}_{2}	&= \mathcal{M}_{8}, &\qquad 
				\mathcal{M}_{2} &= \mathcal{H}_{f} - \mathcal{F}_{3},\\ 
			\mathcal{F}_{3}	&= \mathcal{H}_{y} - \mathcal{M}_{2}, &\qquad 
				\mathcal{M}_{3} &=\mathcal{F}_{8},\\ 
			\mathcal{F}_{4}	&= \mathcal{M}_{1}, &\qquad 
				\mathcal{M}_{4} &= \mathcal{H}_{f} - \mathcal{F}_{7},\\ 
			\mathcal{F}_{5}	&= \mathcal{H}_{y} - \mathcal{M}_{6}, &\qquad 
				\mathcal{M}_{5} &= \mathcal{H}_{f} - \mathcal{F}_{6},\\ 
			\mathcal{F}_{6}	&= \mathcal{H}_{y} - \mathcal{M}_{5}, &\qquad 
				\mathcal{M}_{6} &=\mathcal{H}_{f} -  \mathcal{F}_{5},\\ 
			\mathcal{F}_{7}	&= \mathcal{H}_{f} -  \mathcal{M}_{4}, &\qquad 
				\mathcal{M}_{7} &= \mathcal{F}_{1},\\ 
			\mathcal{F}_{8}	&= \mathcal{M}_{3}, &\qquad 
				\mathcal{M}_{8} &= \mathcal{F}_{2}.
		\end{aligned}
	\end{equation}
	This results in the following correspondences between the surface roots,
	\begin{equation}\label{eq:geom-JM-Ok-4}
	\begin{aligned}
		\delta_{0} &= \mathcal{F}_{7} - \mathcal{F}_{8} = \mathcal{H}_{y} - \mathcal{M}_{3} - \mathcal{M}_{4}, &\qquad  
			\delta_{4} &= \mathcal{H}_{g} - \mathcal{F}_{3} - \mathcal{F}_{5} = \mathcal{H}_{z} - \mathcal{M}_{4} - \mathcal{M}_{5}, \\
		\delta_{1} &= \mathcal{F}_{1} - \mathcal{F}_{2} = \mathcal{M}_{7} - \mathcal{M}_{8}, &\qquad  
			\delta_{5} &= \mathcal{F}_{3} - \mathcal{F}_{4}  = \mathcal{H}_{y} - \mathcal{M}_{1} - \mathcal{M}_{2}, \\
		\delta_{2} &= \mathcal{H}_{f} - \mathcal{F}_{1} - \mathcal{F}_{5} = \mathcal{M}_{6} - \mathcal{M}_{7}, &\qquad  
			\delta_{6} &= \mathcal{F}_{6} - \mathcal{F}_{7} = \mathcal{M}_{4} - \mathcal{M}_{5}; \\
		\delta_{3} &= \mathcal{F}_{5} - \mathcal{F}_{6}  = \mathcal{M}_{5} - \mathcal{M}_{6}, &\qquad  
	\end{aligned}
	\end{equation}
	and the symmetry roots,
	\begin{equation}\label{eq:sym-JM-Ok-4}
	\begin{aligned}
			\alpha_{0} &= \mathcal{H}_{f} + \mathcal{H}_{g} - \mathcal{F}_{5} - \mathcal{F}_{6} - \mathcal{F}_{7} - \mathcal{F}_{8}
			= \mathcal{H}_{z} - \mathcal{M}_{2} - \mathcal{M}_{3}, \\
			\alpha_{1} &= \mathcal{H}_{f} - \mathcal{F}_{3} - \mathcal{F}_{4} =  \mathcal{M}_{2} - \mathcal{M}_{1}, \\
			\alpha_{2} &= \mathcal{H}_{g} - \mathcal{F}_{1} - \mathcal{F}_{2} = 
			2 \mathcal{H}_{y} + \mathcal{H}_{z} - \mathcal{M}_{2} - \mathcal{M}_{4} - \mathcal{M}_{5} - \mathcal{M}_{6} - \mathcal{M}_{7} - \mathcal{M}_{8}.
		\end{aligned}
	\end{equation}
	The normalization $a_{0}+a_{1}+a_{2}=1$  then results in $k=1$ for the symplectic form, and so we recover the symplectic form 
	$\symp{JM} = (1/y) dy\wedge dz$. The root variables then match the Okamoto (and the Its-Prokhorov) ones,
	\begin{equation*}
		a_{0}^{\mathrm{JM}} = \Theta_{\infty} -\Theta_{0} = a_{0}^{\mathrm{IP}}  = a_{0}^{\mathrm{Ok}},\quad 
		a_{1}^{\mathrm{JM}}= 2 \Theta_{0}  = a_{1}^{\mathrm{IP}}  = a_{1}^\mathrm{Ok},\quad 
		a_{2}^{\mathrm{JM}}= 1 - \Theta_{0} - \Theta_{\infty} = a_{1}^{\mathrm{IP}} = a_{2}^\mathrm{Ok}.
	\end{equation*}
\end{lemma}

The corresponding birational change of variables is given in the following Lemma.

\begin{lemma}\label{lem:JMtoOk-coords-4} The change of coordinates and parameter matching between the Jimbo-Miwa and Okamoto Hamiltonian systems is given by
    \begin{equation*}
   	 \left\{\begin{aligned}
   	 	y(f,g,t)&=f,\\
   		z(f,g,t)&= \frac{f^{2}}{2} - fg + ft + \kappa_{0},\\
		\Theta_{0}&=\frac{\kappa_{0}}{2},\\  
		\Theta_{\infty} &= 1 + \theta_{\infty} - \frac{\kappa_{0}}{2},
   	 \end{aligned}\right.
    \qquad\text{and conversely} \qquad 
    	\left\{\begin{aligned}
   	 	f(y,z,t)&=y,\\
   		g(y,z,t)&=\frac{y}{2}- \frac{z}{y}+ t + \frac{2\Theta_{0}}{y},\\
		\kappa_{0} &= 2\Theta_{0},\\  
		\theta_{\infty} &= -1 + \Theta_{0} + \Theta_{\infty}.
    	\end{aligned}\right.
    \end{equation*}
\end{lemma}

\begin{remark} 
	We also note that knowing the configuration of the base points amounts to knowing the correct symplectic structure, 
	from which we can then directly obtain the 
	Hamiltonian, as was observed in \cite{DzhFilLigSto:2021:HSDSFMLWGMTPE}. Indeed, equations (\ref{eq:JM-system-4-1} -- \ref{eq:JM-system-4-2})
	are clearly \emph{not} Hamiltonian w.r.t.~the ``standard'' symplectic form $dz\wedge dy$. However, using the correct form
	$\symp{JM} = (1/y) dy\wedge dz$ obtained in Lemma~\ref{lem:JMtoOk-basis-4}, we get 
	 \begin{align*}
		 \frac{\partial \Ham{}{}}{\partial z}& = -\frac{dy}{y dt} = \frac{4z}{y} - y - 2t - \frac{4\Theta_0}{y}, \\
		 \frac{\partial \Ham{}{}}{\partial y}& = \frac{dz}{y dt} = -\frac{2z^{2}}{y^{2}} - z + \frac{4 \Theta_{0}z}{y^{2}} + (\Theta_{0} + \Theta_{\infty}),\\
		 \intertext{and therefore}
		 \Ham{}{}(y,z;t) &= \frac{2z^{2}}{y} - yz - 2t z - \frac{4 \Theta_{0} z}{y} + (\Theta_{0} + \Theta_{\infty})y = 
		 \Ham{JM}{IV}(y,z;t) -(\Theta_{0} + \Theta_{\infty})t,
	 \end{align*}
	and so we recovered the Jimbo-Miwa Hamiltonian, up to $t$-dependent terms.
\end{remark}


\subsection{The Kecker Hamiltonian system} 
\label{sub:Kek-soic-4}
\begin{notation*}
For the Kecker system we use the following notation: coordinates $(x,y)$, parameters $\tilde{\alpha}$ and $\tilde{\beta}$; time variable $z$;
base points $w_{i}$, exceptional divisors $N_{i}$. 
\end{notation*}

The space of initial conditions for this system is again not minimal and has \emph{ten}
base points, but it looks quite different from the other examples. The base points come in three cascades originating from a single point
$w_{1}(\infty,\infty)$, reflecting the cubic nature of $\Ham{Kek}{IV}(x,y;z)$. Because of that, it is convenient to introduce the cubic roots of unity,
\begin{equation}\label{eq:cube-roots}
	\xi_{0} = 1,\qquad \xi_{1} = e^{\frac{2 \pi \mathfrak{i}}{3}} = \frac{-1 + \mathfrak{i} \sqrt{3}}{2},\qquad
	\xi_{2}= \xi_{1}^{2} = e^{\frac{4 \pi \mathfrak{i}}{3}} = \frac{-1 - \mathfrak{i} \sqrt{3}}{2}.
\end{equation}
Then the base point cascades are
\begin{equation*}
	\begin{tikzpicture}
	\node (w1) at (0,0) [left] {$w_{1}(\infty,\infty)$}; 
	\node (w2) at (2,1) {$w_{2}(-\xi_{0},0)$};  \node (w3) at (4.5,1) {$w_{3}(-z,0)$}; 
		\node (w4) at (8.5,1) {$w_{4}(-\tilde{\beta} - (1 + z^{2}) \xi_{0}  + \tilde{\alpha}\xi_{0},0)$}; 
	\node (w5) at (2,0) {$w_{5}(-\xi_{2},0)$}; \node (w6) at (4.5,0) {$w_{6}(-z,0)$}; 
		\node (w7) at (8.5,0) {$w_{7}(-\tilde{\beta}  - (1 + z^{2}) \xi_{1} + \tilde{\alpha} \xi_{2},0)$};
	\node (w8) at (2,-1) {$w_{8}(-\xi_{1},0)$}; \node (w9) at (4.5,-1) {$w_{9}(-z,0)$}; 
		\node (w10) at (8.5,-1) {$w_{10}(-\tilde{\beta} - (1 + z^{2}) \xi_{2} + \tilde{\alpha} \xi_{1},0)$};
	\draw[->] (w2) -- (0.5,1) -- (w1);	\draw[->] (w3) -- (w2); \draw[->] (w4) -- (w3); 
	\draw[->] (w5) -- (w1); \draw[->] (w6) -- (w5); \draw[->] (w7) -- (w6); 
	\draw[->] (w8) -- (0.5,-1) -- (w1); \draw[->] (w9) -- (w8); \draw[->] (w10) -- (w9); 
	\end{tikzpicture}
\end{equation*}
that lie on the polar divisor of a symplectic form $\omega = k dx\wedge dy$. The space of initial conditions is shown on
Figure~\ref{fig:Kek-soic-4}. Note that this time we have a curve $N_{1} - N_{2} - N_{5} - N_{8}$  with self-intersection index $-4$, and 
it is clear that we should blow down two $-1$ curves intersecting with it, with the obvious choice being the 
inaccessible divisors $H_{x} - N_{1}$ and $H_{y} - N_{1}$.

\begin{figure}[ht]
	\begin{tikzpicture}[>=stealth,basept/.style={circle, draw=red!100, fill=red!100, thick, inner sep=0pt,minimum size=1.2mm}]
	\begin{scope}[xshift=0cm,yshift=0cm]
	\draw [black, line width = 1pt] (-0.4,0) -- (2.9,0)	node [pos=0,left] {\small $H_{y}$} node [pos=1,right] {\small $y=0$};
	\draw [black, line width = 1pt] (-0.4,2.5) -- (2.9,2.5) node [pos=0,left] {\small $H_{y}$} node [pos=1,right] {\small $y=\infty$};
	\draw [black, line width = 1pt] (0,-0.4) -- (0,2.9) node [pos=0,below] {\small $H_{x}$} node [pos=1,above] {\small $x=0$};
	\draw [black, line width = 1pt] (2.5,-0.4) -- (2.5,2.9) node [pos=0,below] {\small $H_{x}$} node [pos=1,above] {\small $x=\infty$};
	\node (w1) at (2.5,2.5) [basept,label={[xshift = -8pt, yshift=-1pt] \small $w_{1}$}] {};
	\node (w5) at (3.5,3) [basept,label={[above] \small $w_{5}$}] {};
	\node (w6) at (4.2,3) [basept,label={[above] \small $w_{6}$}] {};
	\node (w7) at (5.0,3) [basept,label={[above] \small $w_{7}$}] {};
	\node (w2) at (1.8,2) [basept,label={[left] \small $w_{2}$}] {};
	\node (w3) at (1.8,1.3) [basept,label={[left] \small $w_{3}$}] {};
	\node (w4) at (1.8,0.6) [basept,label={[left] \small $w_{4}$}] {};
	\node (w8) at (3.2,2) [basept,label={[right] \small $w_{8}$}] {};
	\node (w9) at (3.2,1.3) [basept,label={[right] \small $w_{9}$}] {};
	\node (w10) at (3.2,0.6) [basept,label={[right] \small $w_{10}$}] {};
	\draw [red, line width = 0.8pt, ->] (w2) -- (w1);	
	\draw [red, line width = 0.8pt, ->] (w3) -- (w2);
	\draw [red, line width = 0.8pt, ->] (w4) -- (w3);
	\draw [red, line width = 0.8pt, ->] (w5) -- (w1);	
	\draw [red, line width = 0.8pt, ->] (w6) -- (w5);
	\draw [red, line width = 0.8pt, ->] (w7) -- (w6);
	\draw [red, line width = 0.8pt, ->] (w8) -- (w1);	
	\draw [red, line width = 0.8pt, ->] (w9) -- (w8);
	\draw [red, line width = 0.8pt, ->] (w10) -- (w9);
	\end{scope}
	\draw [->] (7,1.5)--(5,1.5) node[pos=0.5, below] {$\operatorname{Bl}_{w_{1}\cdots w_{10}}$};
	\begin{scope}[xshift=8cm,yshift=0cm]
	\draw [black, line width = 1pt] (-0.4,0) -- (3.1,0);
	\draw [black, line width = 1pt] (0,-0.4) -- (0,3.1);
	\draw [red, line width = 1pt] 		(3.7,1) .. controls (3.4,0.8) and (2.7,0.4) .. (2.7,-0.4) node[below] {\small $H_{y} - N_{1}$};
	\draw [red, line width = 1pt] 		(1,3.7) .. controls (0.8,3.4) and (0.4,2.7) .. (-0.4,2.7) node[left] {\small $H_{x} - N_{1}$};
	\draw [teal, line width = 1pt] 		(0.5,3.5) -- (3.5,0.5) node[below right] {\small $N_{1} - N_{2} - N_{5} - N_{8}$};
	\draw [blue, line width = 1pt] 		(1,2.5) -- (2,3.5) node[pos = 0, xshift=-8pt, yshift=-6pt] {\small $N_{2} - N_{3}$};
	\draw [blue, line width = 1pt] 		(1.3,3.3) -- (2.8,3.3) node[pos = 1, right] {\small $N_{3} - N_{4}$};
	\draw [red, line width = 1pt] 		(2.3,3.1) -- (2.9,3.7) node[pos = 0, yshift=-3pt, xshift=-5pt] {\small $N_{4}$};
	\draw [blue, line width = 1pt] 		(1.8,1.7) -- (2.8,2.7) node[pos = 0, xshift=-5pt, yshift=-5pt] {\small $N_{5} - N_{6}$};
	\draw [blue, line width = 1pt] 		(2.1,2.5) -- (3.6,2.5) node[pos = 1, right] {\small $N_{6} - N_{7}$};
	\draw [red, line width = 1pt] 		(3.1,2.3) -- (3.7,2.9) node[pos = 0, yshift=-3pt, xshift=-5pt] {\small $N_{7}$};
	\draw [blue, line width = 1pt] 		(2.6,0.9) -- (3.6,1.9) node[pos = 0, xshift=-8pt, yshift=-6pt] {\small $N_{8} - N_{9}$};
	\draw [blue, line width = 1pt] 		(2.9,1.7) -- (4.4,1.7) node[pos = 1, right] {\small $N_{9} - N_{10}$};
	\draw [red, line width = 1pt] 		(3.9,1.5) -- (4.5,2.1) node[pos = 0, yshift=-3pt, xshift=-4pt ] {\small $N_{10}$};
	\end{scope}
	\end{tikzpicture}
	\caption{The Space of Initial Conditions for the Kecker Hamiltonian System \eqref{eq:Kek-Ham-4}}
	\label{fig:Kek-soic-4}
\end{figure}
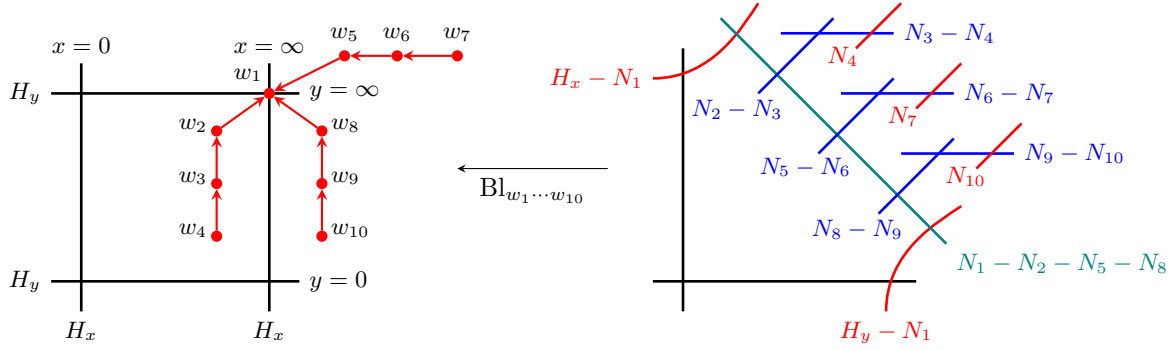

The process here is analogous to the one described in Section~\ref{sub:IP-soic-4}. Instead of blowing down, we add two additional 
blowup points to the standard $E_{6}^{(1)}$ surface shown on Figure~\ref{fig:surface-e61}. Since the $-4$-curve on Figure~\ref{fig:Kek-soic-4}
clearly corresponds to the central node of the surface roots Dynkin diagram, we add those additional  points to the 
divisor $d_{3}=F_{5} - F_{6}$ on Figure~\ref{fig:surface-e61}  and then find the change of basis of the
Picard lattices that matches the irreducible components of the anti-canonical divisor, using the root variables and the period map to find the correct
identification between the three branches of the $E_{6}^{(1)}$ diagrams. We only state the result. 

\begin{lemma}\label{lem:Kek-to-Ok} The change of bases for Picard lattices between the Kecker and the Okamoto (with two additional blowup points) surfaces 
	is given by 
	\begin{equation}\label{eq:basis-Kek-Ok}
		\begin{aligned}
			\mathcal{H}_{f} &= \mathcal{H}_{x} + \mathcal{H}_{y} - \mathcal{N}_{1} - \mathcal{N}_{2}, &\qquad 
				\mathcal{H}_{x} & = \mathcal{H}_{f} + \mathcal{H}_{g} - \mathcal{F}_{5} - \mathcal{F}_{10},\\
			\mathcal{H}_{g} & = \mathcal{H}_{x} + \mathcal{H}_{y} - \mathcal{N}_{1} - \mathcal{N}_{5},  &\qquad 	
				\mathcal{H}_{y} &= \mathcal{H}_{f} + \mathcal{H}_{g} - \mathcal{F}_{5} - \mathcal{F}_{9}, \\
			\mathcal{F}_{1}	&= \mathcal{N}_{6}, &\qquad 
				\mathcal{N}_{1} &= \mathcal{H}_{f} + \mathcal{H}_{g} -  \mathcal{F}_{5} - \mathcal{F}_{9} - \mathcal{F}_{10},\\ 
			\mathcal{F}_{2}	&= \mathcal{N}_{7}, &\qquad 
				\mathcal{N}_{2} &= \mathcal{H}_{g} - \mathcal{F}_{5},\\ 
			\mathcal{F}_{3}	&=  \mathcal{N}_{3}, &\qquad 
				\mathcal{N}_{3} &= \mathcal{F}_{3},\\ 
			\mathcal{F}_{4}	&= \mathcal{N}_{4}, &\qquad 
				\mathcal{N}_{4} &= \mathcal{F}_{4},\\ 
			\mathcal{F}_{5}	&= \mathcal{H}_{x} +  \mathcal{H}_{y} - \mathcal{N}_{1} - \mathcal{N}_{2} - \mathcal{N}_{5}, &\qquad 
				\mathcal{N}_{5} &= \mathcal{H}_{f} - \mathcal{F}_{5},\\ 
			\mathcal{F}_{6}	&= \mathcal{N}_{8}, &\qquad 
				\mathcal{N}_{6} &= \mathcal{F}_{1},\\ 
			\mathcal{F}_{7}	&= \mathcal{N}_{9}, &\qquad 
				\mathcal{N}_{7} &= \mathcal{F}_{2},\\ 
			\mathcal{F}_{8}	&= \mathcal{N}_{10}, &\qquad 
				\mathcal{N}_{8} &= \mathcal{F}_{6},\\ 
			\mathcal{F}_{9}	&= \mathcal{H}_{x} - \mathcal{N}_{1}, &\qquad 
				\mathcal{N}_{9} &= \mathcal{F}_{7},\\ 
			\mathcal{F}_{10}	&= \mathcal{H}_{y} - \mathcal{N}_{1}, &\qquad 
				\mathcal{N}_{10} &=\mathcal{F}_{8}.
		\end{aligned}
	\end{equation}
	This results in the following correspondences between the surface roots,
	\begin{equation}\label{eq:geom-Kek-Ok}
	\begin{aligned}
		\delta_{0} &= \mathcal{F}_{7} - \mathcal{F}_{8} = \mathcal{N}_{9} - \mathcal{N}_{10}, &\qquad  
			\delta_{4} &= \mathcal{H}_{g} - \mathcal{F}_{3} - \mathcal{F}_{5} = \mathcal{N}_{2} - \mathcal{N}_{3}, \\
		\delta_{1} &= \mathcal{F}_{1} - \mathcal{F}_{2} = \mathcal{N}_{6} - \mathcal{N}_{7}, &\qquad  
			\delta_{5} &= \mathcal{F}_{3} - \mathcal{F}_{4} = \mathcal{N}_{3} - \mathcal{N}_{4}, \\
		\delta_{2} &= \mathcal{H}_{f} - \mathcal{F}_{1} - \mathcal{F}_{5} = \mathcal{N}_{5} - \mathcal{N}_{6}, &\qquad  
			\delta_{6} &= \mathcal{F}_{6} - \mathcal{F}_{7} = \mathcal{N}_{8} - \mathcal{N}_{9}, \\
		\delta_{3} &= \mathcal{F}_{5} - \mathcal{F}_{6} - \mathcal{F}_{9} -  \mathcal{F}_{10} = 
		\mathcal{N}_{1} - \mathcal{N}_{2} - \mathcal{N}_{5} - \mathcal{N}_{8}, &\qquad  
	\end{aligned}
	\end{equation}
	and the symmetry roots,
	\begin{equation}\label{eq:sym-Kek-Ok}
	\begin{aligned}
			\alpha_{0} &= \mathcal{H}_{f} + \mathcal{H}_{g} - \mathcal{F}_{5} - \mathcal{F}_{6} - \mathcal{F}_{7} - \mathcal{F}_{8}
			= \mathcal{H}_{x} + \mathcal{H}_{y} - \mathcal{N}_{1} - \mathcal{N}_{8} - \mathcal{N}_{9} - \mathcal{N}_{10}, \\
			\alpha_{1} &= \mathcal{H}_{f} - \mathcal{F}_{3} - \mathcal{F}_{4} =  
			\mathcal{H}_{x} + \mathcal{H}_{y} - \mathcal{N}_{1} - \mathcal{N}_{2} - \mathcal{N}_{3} - \mathcal{N}_{4},\\
			\alpha_{2} &= \mathcal{H}_{g} - \mathcal{F}_{1} - \mathcal{F}_{2} = 
			\mathcal{H}_{x} + \mathcal{H}_{y} - \mathcal{N}_{1} - \mathcal{N}_{5} - \mathcal{N}_{6} - \mathcal{N}_{7}.
		\end{aligned}
	\end{equation}
	The symplectic form here is $\sympt{Kek} = (1/3)dy\wedge dx$, and the root variables match, 
	\begin{align*}
		a_{0}^{\mathrm{Kek}} &= \frac{\xi_{0} + \tilde{\beta} \xi_{1} - \tilde{\alpha} \xi_{2} }{3} 
		= \frac{1}{3} + \frac{\tilde{\alpha}- \tilde{\beta}}{6} + \frac{\mathfrak{i}(\tilde{\alpha} + \tilde{\beta})}{2 \sqrt{3}} 
		= 1 + \theta_{\infty} - \kappa_{0} = a_{0}^{\mathrm{Ok}},\\
		a_{1}^{\mathrm{Kek}} &=  \frac{\xi_{0} + \tilde{\beta} \xi_{0} - \tilde{\alpha} \xi_{0} }{3} 
		= \frac{1 - \tilde{\alpha} + \tilde{\beta}}{3}  = 2 \kappa_{0} = a_{1}^\mathrm{Ok},\\
		a_{2}^{\mathrm{Kek}} &= \frac{\xi_{0} + \tilde{\beta} \xi_{2} - \tilde{\alpha} \xi_{1} }{3} 
		= \frac{1}{3} + \frac{\tilde{\alpha}- \tilde{\beta}}{6} - \frac{\mathfrak{i}(\tilde{\alpha} + \tilde{\beta})}{2 \sqrt{3}}
		= - \theta_{\infty} = a_{2}^\mathrm{Ok}.				
	\end{align*}
\end{lemma}

To find the corresponding birational change of variables, we need to keep in mind that we also have to find the 
relationship between the independent variables that enter geometrically as coordinates of the blowup points, so we include some
details in the proof.

\begin{lemma}\label{lem:coords-Kek-Ok-4} The change of coordinates and parameter matching between the Kecker and Okamoto Hamiltonian systems is given by
    \begin{align}\label{eq:KektoOk-4}
   	 &\left\{\begin{aligned}
   	 	x(f,g,t)&=\frac{(1 + \mathfrak{i})}{4(3)^{3/4}}\Big( 3(\sqrt{3}-\mathfrak{i})f + 12 \mathfrak{i} g + 2 (\sqrt{3} - 3 \mathfrak{i})t \Big),\\
   	 	y(f,g,t)&=\frac{(1 + \mathfrak{i})}{4(3)^{3/4}}\Big( 3(\sqrt{3}+\mathfrak{i})f - 12 \mathfrak{i} g + 2 (\sqrt{3} + 3 \mathfrak{i})t \Big),\\
		z(t) &= \left(-\frac{4}{3}\right)^{\frac{1}{4}}t,
   	 \end{aligned}\right.\quad 
	 \begin{aligned}
		\tilde{\alpha}&=\frac{(\sqrt{3} - 3 \mathfrak{i}) - 6 \mathfrak{i} \theta_{\infty} - 3(\sqrt{3} - \mathfrak{i})\kappa_{0}}{2 \sqrt{3}},\\  
		\tilde{\beta}&=\frac{(-\sqrt{3} - 3 \mathfrak{i}) - 6 \mathfrak{i} \theta_{\infty} + 3(\sqrt{3} + \mathfrak{i})\kappa_{0}}{2 \sqrt{3}},	 	
	 \end{aligned}\\
	 \intertext{and converseley,}
    	&\left\{\begin{aligned}
   	 	f(x,y,z)&=\frac{1-\mathfrak{i}}{3^{3/4}}(x + y - z),\\
   		g(x,y,z)&= - \frac{(1 + \mathfrak{i})}{4(3)^{3/4}}\Big( (\sqrt{3}+\mathfrak{i})x - (\sqrt{3} - \mathfrak{i})y + 2 \mathfrak{i} z  \Big),\\
		t(z) &= \left(-\frac{3}{4}\right)^{\frac{1}{4}}z,
    	\end{aligned}\right.\quad
		\begin{aligned}
		\kappa_{0} &= \frac{1 - \tilde{\alpha} + \tilde{\beta}}{3},\\
		\theta_{\infty} &= \frac{-2 + (\sqrt{3}\mathfrak{i} - 1)\tilde{\alpha} + (\sqrt{3}\mathfrak{i} + 1)\tilde{\beta} }{6}.			
		\end{aligned}\label{eq:OktoKek-4}
    \end{align}
\end{lemma}

\begin{proof}
	 Let $\varphi: (q,p)\to (f,g)$ be the change of variables that induces the above change of bases of the Picard lattice. 
	In the $(x,y)$-chart, the equations of the base curves for the pencil $|\mathcal{H}_{f}| = |\mathcal{H}_{x} + \mathcal{H}_{y} - \mathcal{N}_{1} - \mathcal{N}_{2}|$ can be taken to be $y + \xi_{0} x$ and $1$, and for the pencil $|\mathcal{H}_{g}| = |\mathcal{H}_{x} + \mathcal{H}_{y} - \mathcal{N}_{1} - \mathcal{N}_{5}|$ we can take $y + \xi_{2} x$ and $1$. Thus, up to the M\"obius transformations, we get
	\begin{equation*}
		f(x,y) = \frac{A (y + \xi_{0}x) + B}{C (y + \xi_{0}x) + D}, \qquad 
		g(x,y) = \frac{K (y + \xi_{2}x) + L}{M (y + \xi_{2}x) + N}.
	\end{equation*}
	Using the correspondence between the $-4$ and $-2$ curves $\delta_{i}$ given by \eqref{eq:geom-Kek-Ok} we can evaluate the coefficients 
	$A,\ldots,N$ to get 
	\begin{equation*}
		f(x,y) = \frac{-2t \xi_{1}\xi_{2}(x + \xi_{0}^{2} y - \xi_{0}z)}{\xi_{0} z(\xi_{0} - \xi_{2})(\xi_{1} - \xi_{0})}, \qquad 
		g(x,y) = \frac{-2t \xi_{1}\xi_{2}(x + \xi_{0}^{2} y - \xi_{0}z)}{\xi_{0} z(\xi_{0} - \xi_{2})(\xi_{1} - \xi_{0})},
	\end{equation*}
	or, after some simplification with $\xi_{i}$,
	\begin{equation*}
		f(x,y) = \frac{2t (x +  y - z)}{ z(\xi_{1} - 1)(\xi_{2} - 1)}, \qquad 
		g(x,y) = \frac{-t \xi_{1}(x + \xi_{1} y - \xi_{2}z)}{z(\xi_{1} - 1)(\xi_{2} - 1)}.
	\end{equation*}
	Next, from the divisor matching $\varphi_{*}(N_{7})=F_{2}$, $\varphi_{*}(N_{4}) = F_{4}$, and $\varphi_{*}(N_{10}) = F_{8}$ we get a 
	correspondence between the coordinates of the blowup points, which in turn gives us the following relationship between parameters:
	\begin{equation*}
		\theta_{\infty} = \frac{2t^{2}(\tilde{\beta} + \xi_{1} - \xi_{2} \tilde{\alpha})}{3z^{2} (\xi_{2} - 1)},\quad 
		\kappa_{0} = - \frac{2t^{2}\xi_{1}(\tilde{\beta} + 1 - \tilde{\alpha})}{3z^{2} (\xi_{2} - 1)},\quad 
		1 + 2 t^{2} + \theta_{\infty} - \kappa_{0} = 
		\frac{2t^{2} (3 z^{2} (\xi_{2} - 1) + \tilde{\alpha} - \xi_{1} - \xi_{2} \tilde{\beta})}{3z^{2} (\xi_{2} - 1)}.
	\end{equation*}
	These equations then give us the required parameter matching, confirming the root variable matching in 
	Lemma~\ref{eq:basis-Kek-Ok} and  the parameter matching  
	in \eqref{eq:Kek-pars}, as well as the relation between the different time variables $t$ (for the Okamoto System)
	and $z$ (for the Kecker system),  $\sqrt{3} z^{2} = 2 \mathfrak{i} t^{2}$, i.e., $z = \left(-\frac{4}{3}\right)^{\frac{1}{4}} t$,
	as expected.
	
	Substituting expressions \eqref{eq:cube-roots} for $\xi_{i}$ then gives \eqref{eq:OktoKek-4} and \eqref{eq:KektoOk-4} is established
	directly by inverting the change of coordinates.
\end{proof}

\section{The Fifth Painlev\'e Equation $\Pain{V}$} 
\label{sec:P5}

For the standard form of the fifth Painlev\'e equations we take the one in \cite{Oka:1987:SPEISPEP},
\begin{equation}\label{eq:P5-std}
	\frac{d^{2} w}{dt^2} = \left(\frac{1}{2w} + \frac{1}{w-1}\right)\left(\frac{dw}{dt}\right)^{2} - \frac{1}{t} \frac{dw}{dt} + 
		\frac{(w-1)^{2}}{t^{2}}\left(\alpha w + \frac{\beta}{w}\right) + \frac{\gamma}{t} w + \delta \frac{w(w+1)}{w-1},
\end{equation}
where $t$ is an independent (complex) variable, $w(t)$ is the dependent variable, and $\alpha,\beta,\gamma,\delta\in \mathbb{C}$ are some 
parameters. Following Okamoto \cite{Oka:1987:SPEIFPEP}\footnote{Note that there is a  sign typo in the form of $\Pain{V}$ in that paper, it should be $\frac{d^{2} q}{dt^{2}} = \left(\frac{1}{2q} + \frac{1}{q-1}\right)\left(\frac{dq}{dt}\right)^{2} +\cdots$}, it is convenient to introduce new parameters 
$\kappa_{0}$, $\kappa_{\infty}$, $\theta$, and $\eta$ that we call \emph{Okamoto parameters}
\begin{equation}\label{eq:Ok-pars-PV}
	\alpha = \frac{\kappa_{\infty}^{2}}{2},\quad \beta = -\frac{\kappa_{0}^{2}}{2},\quad \gamma = - \eta (\theta + 1), \quad \delta = -	\frac{1}{2}\eta^{2},
	\quad\text{as well as } \kappa = \frac{(\kappa_{0}+\theta)^{2} - \kappa_{\infty}^{2}}{4},
\end{equation}
where the parameter notation is coming from a certain isomonodromy problem. Note that for $\delta=0$, 
$\Pain{V}$ reduces to $\Pain{III}$ and, excluding this case,
it is possible to rescale the variables and put $\eta=-1$. Using these parameters, one possible Hamiltonian for $\Pain{V}$
is given in \cite{Oka:1987:SPEISPEP},
\begin{equation}\label{eq:Ok-Ham5}
	\Ham{Ok}{V}(f,g;t) = \frac{1}{t}\Big(f(f-1)^{2}g^{2} - \left(\kappa_{0}(f-1)^{2} + \theta f(f-1) + \eta t f  \right)g + \kappa(f-1)\Big).
\end{equation}
Using the standard symplectic form $\symp{Ok} = dg\wedge df$, we get the Hamiltonian system
  \begin{equation}\label{eq:Ok-Ham5-sys}
  	\left\{
 	\begin{aligned}
 		\frac{df}{dt} &= \frac{\partial \Ham{Ok}{V}}{\partial g} 
			=\frac{1}{t}\Big((f-1)^{2} (2 f g - \kappa_{0}) - f(\theta(f-1) + t \eta)\Big),\\
		\frac{dg}{dt} &= -\frac{\partial \Ham{Ok}{V}}{\partial f} 
			=-\frac{1}{t}\Big(g \left((f-1)(3 f g- g - 2 \kappa_{0} - \theta) - (\theta f+ t \eta)\right) + \kappa\Big),
 	\end{aligned}
	\right.
  \end{equation}
and eliminating the function $g=g(t)$ from these equations we get $\Pain{V}$ equation \eqref{eq:P5-std} for the function $f=f(t)$ with parameters given by 
\eqref{eq:Ok-pars-PV}.
We take the Okamoto form of the Hamiltonian system for $\Pain{V}$ to be the reference one. However, from the geometric point of view it is better to consider the Hamiltonian given in \cite{KajNouYam:2017:GAPE},
\begin{equation}\label{eq:KNY-Ham5}
	\Ham{KNY}{V}(q,p;t) = \frac{1}{t}\Big(q (q-1) p(p+t) - (a_{1} + a_{3})q p + a_{1}p + a_{2} t q\Big),	
\end{equation}
since the Okamoto space of initial conditions for the resulting Hamiltonian system  w.r.t.~the standard symplectic form $\symp{KNY} = dp \wedge dq$,
\begin{equation}\label{eq-KNY-Ham5-sys}
	\left\{
	\begin{aligned}
		\frac{dq}{dt} &= \frac{\partial \Ham{KNY}{V}}{\partial p} =\frac{1}{t}\Big(q(q-1)(2p+t) - a_{1}(q-1) - a_{3}q\Big),\\
		\frac{dp}{dt} &= -\frac{\partial \Ham{KNY}{V}}{\partial q} =\frac{1}{t}\Big(p(p+t)(1-2q) + (a_{1} + a_{3})p - a_{2}t\Big)
	\end{aligned}
	\right.
\end{equation}
coincides with the standard $D_{5}^{(1)}$ Sakai surface, where the parameters $a_{i}$ satisfying the standard normalization condition 
$a_{0}+a_{1}+a_{2}+a_{3}=1$ are the \emph{root variables}. 
The system \eqref{eq-KNY-Ham5-sys} reduces to the standard Painlev\'e equation \eqref{eq:P5-std} for the variable 
$w(t) = 1 -  \frac{1}{q(t)}$ with the parameter values
\begin{equation}\label{eq:root-pars-PV}
	\alpha = \frac{a_{1}^{2}}{2},\quad \beta = -\frac{a_{3}^{2}}{2},\quad \gamma = a_{0}-a_{2}, \quad \delta = -	\frac{1}{2}.
\end{equation}

There are a few other Hamiltonians for $\Pain{V}$ known in the literature. In \cite{JimMiw:1981:MPDLODEWRC-II} M.~Jimbo and T.~Miwa gave the 
Hamiltonian
\begin{equation}\label{eq:JM-Ham5}
	\begin{aligned}
	\Ham{JM}{V}(y,z;t) &= -\frac{1}{t}\left(z - \frac{1}{y}\left(z + \frac{\Theta_{0} + \Theta_{1} + \Theta_{\infty}}{2}\right)\right)	
			\left(z + \Theta_{0} - y\left(z + \frac{\Theta_{0} - \Theta_{1} + \Theta_{\infty}}{2}\right)\right)	- z - 
			\frac{\Theta_{0} + \Theta_{\infty}}{2}.
	\end{aligned} 
\end{equation}
For this Hamiltonian the symplectic form $\symp{JM}$ is logarithmic, $\symp{JM} = (1/y) dy\wedge dz$, and the resulting system 
is
%
\begin{equation}\label{eq:JM-Ham5-sys}
\left\{
\begin{aligned}
-\frac{1}{y}\frac{dy}{dt} &= \frac{\partial \Ham{JM}{V}}{\partial z} =-1 + \frac{1}{2t y}\left(
(y-1)\Big(\Theta_{0}(y-3) - \Theta_{1}(y+1) + (\Theta_{\infty} + 4z)(y-1)\Big)\right)\\
-\frac{1}{y}\frac{dz}{dt} &= -\frac{\partial \Ham{JM}{V}}{\partial y} = - \frac{1}{2 t y^2}\Big(
z\Big(\Theta_{0}(y^{2}-3) - \Theta_{1} (y^{2} + 1) + (\Theta_{\infty} + 2z)(y^{2}-1) \Big) -
\Theta_{0}(\Theta_{0} + \Theta_{1} + \Theta_{\infty})
\Big).
\end{aligned}
\right.
\end{equation}

This system reduces to the standard $\Pain{V}$ equation \eqref{eq:P5-std} for the variable $y(t)$ with the parameter matching
\begin{equation}\label{eq:JM-pars-PV}
	\alpha = \frac{1}{2}\left(\frac{\Theta_{0}-\Theta_{1}+\Theta_{\infty}}{2}\right)^{2},\quad 
	\beta = -\frac{1}{2}\left(\frac{\Theta_{0}-\Theta_{1}-\Theta_{\infty}}{2}\right)^{2},\quad 
	\gamma = 1 - \Theta_{0} - \Theta_{1}, \quad \delta = -\frac{1}{2}.
\end{equation}

In \cite{ZolFil:2015:PEEIEF} G.~Filipuk and H.~\.{Z}o\l\c{a}dek introduced the Hamiltonian
\begin{equation*}
	\widetilde{\Ham{FZ}{V}}(x,y;t) = \frac{x(x-1)^{2}y^{2}}{2t} - \frac{\alpha x}{t} + \frac{\beta}{t x} + \frac{\gamma}{x-1} + \frac{\delta t x}{(x-1)^{2}},
\end{equation*}
where parameters $\alpha,\beta,\gamma,\delta$ are the same as in \eqref{eq:P5-std}. 
Rescaling the variables slightly, similar to \eqref{eq:FZ-IP-Ham-4}, we consider instead the Hamiltonian
\begin{equation}\label{eq:FZ-Ham5}
	\Ham{FZ}{V}(x,y;t) = \frac{1}{2}\widetilde{\Ham{FZ}{V}}(x,2y;t) = 
	\frac{x(x-1)^{2}y^{2}}{t} + \frac{1}{2}\left(- \frac{\alpha x}{t} + \frac{\beta}{t x} + \frac{\gamma}{x-1} + \frac{\delta t x}{(x-1)^{2}}\right).
\end{equation}
Using the standard symplectic form $\symp{FZ} = dy \wedge dx$
and the Okamoto parameters \eqref{eq:Ok-pars-PV}, we get the following system
\begin{equation}\label{eq:FZ-Ham5-sys}
	\left\{
	\begin{aligned}
		\frac{dx}{dt} &= \frac{\partial \Ham{FZ}{V}}{\partial y} =\frac{2x(x-1)^{2}y}{t},\\
		\frac{dy}{dt} &= -\frac{\partial \Ham{FZ}{V}}{\partial y} =-\frac{1}{4}\left(
		\frac{2 (\theta+1) \eta}{(x-1)^{2}} + \frac{\eta^{2} t(x+1)}{(x-1)^{3}} + 
			\frac{\kappa_{0}^{2}}{t x^{2}} + \frac{4(3x^{2} - 4x + 1)y^{2} - \kappa_{\infty}^{2}}{t} 
		\right).
	\end{aligned}
	\right.
\end{equation}

Finally, the Hamiltonian given in \cite{ItsPro:2018:SHPIF} (that we write in variables $(f,g)$ instead of the original variables $(q,p)$),
\begin{equation}\label{eq:IP-Ham5}
	\begin{aligned}
	\Ham{IP}{V}(f,g;t) &= \frac{1}{t}\Big(f(f-1)^{2}g^{2} + (\Theta_{0} + 3 \Theta_{1} + \Theta_{\infty})f^{2}g + (t - 2\Theta_{\infty} - 4\Theta_{1})fg
	+ (\Theta_{\infty} +  \Theta_{1} - \Theta_{0})g \\
	&\qquad + 2\Theta_{1}(\Theta_{\infty} + \Theta_{1} + \Theta_{0})f + \Theta_{0}^{2} - (\Theta_{1} + \Theta_{\infty})^{2} + \Theta_{1}t \Big)		
	\end{aligned}
\end{equation}
coincides, up to purely $t$-dependent terms, with the Okamoto Hamiltonian \eqref{eq:Ok-Ham5} under the parameter matching
\begin{equation}\label{eq:IP-pars-Ok}
\Theta_{0} = \frac{\kappa_{0} - \kappa_{\infty} - \theta}{4},\quad \Theta_{1} = - \frac{\kappa_{0} - \kappa_{\infty} + \theta}{4},\quad 
\Theta_{\infty} = -\frac{\kappa_{0} + \kappa_{\infty}}{2}, 
\end{equation}
and the normalization $\eta=-1$, and so we do not consider it in any detail. 
Note that here the parameters $\Theta_{i}$ are, up to some simple scaling, the same as in the Jimbo-Miwa case. 

In the next sections we give the spaces of initial conditions for each of those Hamiltonian systems and give the change of coordinates
reducing them to the standard Okamoto case. These computations are standard and so we only summarize the geometric data. 
We begin, however, with system~\eqref{eq-KNY-Ham5-sys} and the standard $D_{5}^{(1)}$ surface.

\subsection{The Kajiwara-Noumi-Yamada Hamiltonian system} 
\label{sub:KNY-5}

\begin{notation*}
For the Kajiwara-Noumi-Yamada system we use the following notation: coordinates $(q,p)$, 
parameters $a_{0},a_{1},a_{2},a_{3}$ (the root variables); time variable $t$;
base points $p_{i}$, exceptional divisors $E_{i}$.	
\end{notation*}

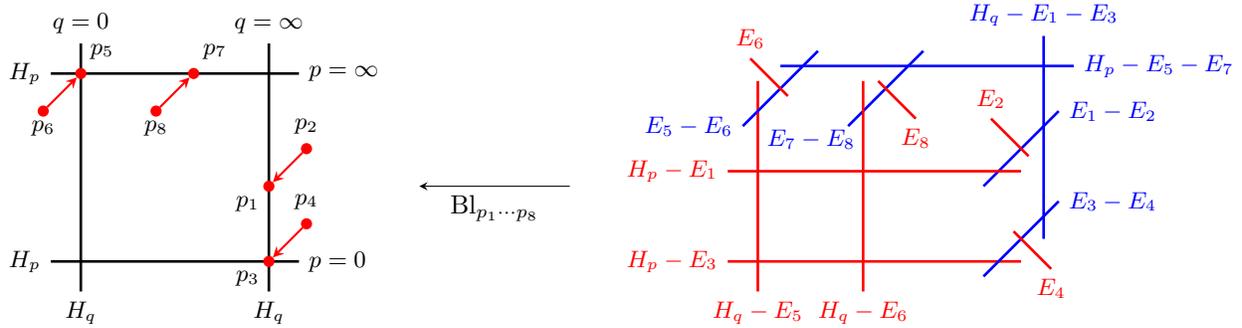
\begin{figure}[ht]
	\begin{tikzpicture}[>=stealth,basept/.style={circle, draw=red!100, fill=red!100, thick, inner sep=0pt,minimum size=1.2mm}]
	\begin{scope}[xshift=0cm,yshift=0cm]
	\draw [black, line width = 1pt] (-0.4,0) -- (2.9,0)	node [pos=0,left] {\small $H_{p}$} node [pos=1,right] {\small $p=0$};
	\draw [black, line width = 1pt] (-0.4,2.5) -- (2.9,2.5) node [pos=0,left] {\small $H_{p}$} node [pos=1,right] {\small $p=\infty$};
	\draw [black, line width = 1pt] (0,-0.4) -- (0,2.9) node [pos=0,below] {\small $H_{q}$} node [pos=1,above] {\small $q=0$};
	\draw [black, line width = 1pt] (2.5,-0.4) -- (2.5,2.9) node [pos=0,below] {\small $H_{q}$} node [pos=1,above] {\small $q=\infty$};
	\node (p3) at (2.5,0) [basept,label={[xshift = -8pt, yshift=-15pt] \small $p_{3}$}] {};
	\node (p4) at (3,0.5) [basept,label={[yshift=0pt] \small $p_{4}$}] {};
	\node (p1) at (2.5,1) [basept,label={[xshift = -8pt, yshift=-15pt] \small $p_{1}$}] {};
	\node (p2) at (3,1.5) [basept,label={[yshift=0pt] \small $p_{2}$}] {};
	\node (p5) at (0,2.5) [basept,label={[xshift = 8pt, yshift=0pt] \small $p_{5}$}] {};
	\node (p6) at (-.5,2) [basept,label={[yshift=-15pt] \small $p_{6}$}] {};
	\node (p7) at (1.5,2.5) [basept,label={[xshift = 8pt, yshift=0pt] \small $p_{7}$}] {};
	\node (p8) at (1,2) [basept,label={[yshift=-15pt] \small $p_{8}$}] {};
	\draw [red, line width = 0.8pt, ->] (p2) -- (p1);
	\draw [red, line width = 0.8pt, ->] (p4) -- (p3);
	\draw [red, line width = 0.8pt, ->] (p6) -- (p5);
	\draw [red, line width = 0.8pt, ->] (p8) -- (p7);	
	\end{scope}
	\draw [->] (6.5,1)--(4.5,1) node[pos=0.5, below] {$\operatorname{Bl}_{p_{1}\cdots p_{8}}$};
	\begin{scope}[xshift=9cm,yshift=0cm]
	\draw [red, line width = 1pt] (-0.4,0) -- (3.5,0)	node [pos=0, left] {\small $H_{p}-E_{3}$};
	\draw [red, line width = 1pt] (0,-0.4) -- (0,2.4) node [pos=0, below] {\small $H_{q}-E_{5}$};
	\draw [blue, line width = 1pt] (-0.2,1.8) -- (0.8,2.8) node [pos=0, left] {\small $E_{5}-E_{6}$};
	\draw [red, line width = 1pt] (-0.1,2.7) -- (0.4,2.2) node [pos=0, above] {\small $E_{6}$};
	\draw [blue, line width = 1pt] (1.2,1.8) -- (2.2,2.8) node [pos=0, xshift=-14pt, yshift=-5pt] {\small $E_{7}-E_{8}$};
	\draw [red, line width = 1pt] (1.6,2.4) -- (2.1,1.9) node [pos=1, below] {\small $E_{8}$};
	\draw [blue, line width = 1pt] (0.3,2.6) -- (4.2,2.6) node [pos=1,right] {\small $H_{p} - E_{5} - E_{7}$};
	\draw [blue, line width = 1pt] (3,-0.2) -- (4,0.8) node [pos=1,right] {\small $E_{3} - E_{4}$};
	\draw [red, line width = 1pt] (3.4,0.4) -- (3.9,-0.1) node [pos=1, below] {\small $E_{4}$};
	\draw [blue, line width = 1pt] (3.8,0.3) -- (3.8,3) node [pos=1, above] {\small $H_{q}-E_{1} - E_{3}$};	
	\draw [blue, line width = 1pt] (3,1) -- (4,2) node [pos=1,right] {\small $E_{1} - E_{2}$};
		\draw [red, line width = 1pt] (3.1,1.9) -- (3.6,1.4) node [pos=0, above] {\small $E_{2}$};
	\draw [red, line width = 1pt] (-0.4,1.2) -- (3.5,1.2)	node [pos=0, left] {\small $H_{p}-E_{1}$};
	\draw [red, line width = 1pt] (1.4,-0.4) -- (1.4,2.4) node [pos=0, below] {\small $H_{q}-E_{6}$};
	\end{scope}
	\end{tikzpicture}
	\caption{The Space of Initial Conditions for the Kajiwara-Noumi-Yamada Hamiltonian System (standard $D_{5}^{(1)}$ surface)}
	\label{fig:KNY-soic-5}
\end{figure}

The Okamoto space of initial conditions for system~\eqref{eq-KNY-Ham5-sys} is given on  Figure~\ref{fig:KNY-soic-5}. It is the standard realization of the
$D_{5}^{(1)}$ Sakai surface and the coordinates of the basepoints are 
given in terms of root variables satisfying the usual normalization condition $a_{0} + a_{1} + a_{2} + a_{3} = 1$ by
\begin{equation*}
	p_{1}(\infty,-t)\leftarrow p_{2}(0,-a_{0}),\quad p_{3}(\infty,0)\leftarrow p_{4}(0,-a_{2}),\quad p_{5}(0,\infty)\leftarrow p_{6}(a_{1},0),\quad 
	p_{7}(1,\infty) \leftarrow p_{8}(a_{3},0).
\end{equation*}
This is the same parameterization of the point configuration as in section 8.2.18 of \cite{KajNouYam:2017:GAPE}.

The surface and symmetry root bases for this standard realization of the $D_{5}^{(1)}$ surface are given on Fig.~\ref{fig:d-roots-d5-KNY}	
and Fig.~\ref{fig:a-roots-a3-KNY} respectively. The birational representation of the extended affine Weyl symmetry group 
$\widetilde{W}(D_{5}^{(1)})$ is given in \cite{KajNouYam:2017:GAPE} and \cite{HuDzhChe:2020:PLUEDPE}, and we do not reproduce it here. 

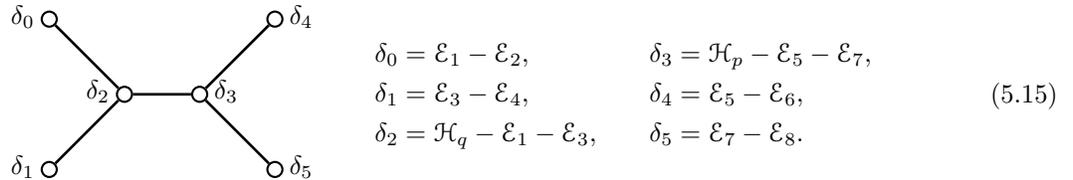
\begin{figure}[ht]
\begin{equation}\label{eq:d-roots-d51}			
	\raisebox{-32.1pt}{\begin{tikzpicture}[
			elt/.style={circle,draw=black!100,thick, inner sep=0pt,minimum size=2mm}]
		\path 	(-1,1) 	node 	(d0) [elt, label={[xshift=-10pt, yshift = -10 pt] $\delta_{0}$} ] {}
		        (-1,-1) node 	(d1) [elt, label={[xshift=-10pt, yshift = -10 pt] $\delta_{1}$} ] {}
		        ( 0,0) 	node  	(d2) [elt, label={[xshift=-10pt, yshift = -10 pt] $\delta_{2}$} ] {}
		        ( 1,0) 	node  	(d3) [elt, label={[xshift=10pt, yshift = -10 pt] $\delta_{3}$} ] {}
		        ( 2,1) 	node  	(d4) [elt, label={[xshift=10pt, yshift = -10 pt] $\delta_{4}$} ] {}
		        ( 2,-1) node 	(d5) [elt, label={[xshift=10pt, yshift = -10 pt] $\delta_{5}$} ] {};
		\draw [black,line width=1pt ] (d0) -- (d2) -- (d1)  (d2) -- (d3) (d4) -- (d3) -- (d5);
	\end{tikzpicture}} \qquad
			\begin{alignedat}{2}
			\delta_{0} &= \mathcal{E}_{1} - \mathcal{E}_{2}, &\qquad  \delta_{3} &= \mathcal{H}_{p} - \mathcal{E}_{5} - \mathcal{E}_{7},\\
			\delta_{1} &= \mathcal{E}_{3} - \mathcal{E}_{4}, &\qquad  \delta_{4} &= \mathcal{E}_{5} - \mathcal{E}_{6},\\
			\delta_{2} &= \mathcal{H}_{q} - \mathcal{E}_{1} - \mathcal{E}_{3}, &\qquad  \delta_{5} &= \mathcal{E}_{7} - \mathcal{E}_{8}.
			\end{alignedat}
\end{equation}
	\caption{The Surface Root Basis for the standard $D_{5}^{(1)}$ Sakai surface point configuration}
	\label{fig:d-roots-d5-KNY}	
\end{figure}

\begin{figure}[ht]
\begin{equation}\label{eq:a-roots-a31}			
	\raisebox{-32.1pt}{\begin{tikzpicture}[
			elt/.style={circle,draw=black!100,thick, inner sep=0pt,minimum size=2mm}]
		\path 	(-1,1) 	node 	(a0) [elt, label={[xshift=-10pt, yshift = -10 pt] $\alpha_{0}$} ] {}
		        (-1,-1) node 	(a1) [elt, label={[xshift=-10pt, yshift = -10 pt] $\alpha_{1}$} ] {}
		        ( 1,-1) node  	(a2) [elt, label={[xshift=10pt, yshift = -10 pt] $\alpha_{2}$} ] {}
		        ( 1,1) 	node 	(a3) [elt, label={[xshift=10pt, yshift = -10 pt] $\alpha_{3}$} ] {};
		\draw [black,line width=1pt ] (a0) -- (a1) -- (a2) --  (a3) -- (a0); 
	\end{tikzpicture}} \qquad
			\begin{alignedat}{2}
			\alpha_{0} &= \mathcal{H}_{p} - \mathcal{E}_{1} - \mathcal{E}_{2}, &\qquad  \alpha_{2} &= \mathcal{H}_{p} - \mathcal{E}_{3} - \mathcal{E}_{4},\\
			\alpha_{1} &= \mathcal{H}_{q} - \mathcal{E}_{5} - \mathcal{E}_{6}, &\qquad  \alpha_{3} &= \mathcal{H}_{q} - \mathcal{E}_{7} - \mathcal{E}_{8}.
			\\[5pt]
			\delta & = \mathrlap{\alpha_{0} + \alpha_{1} +  \alpha_{2} + \alpha_{3}.} 
			\end{alignedat}
\end{equation}
	\caption{The Symmetry Root Basis for the standard $A_{3}^{(1)}$ symmetry sub-lattice}
	\label{fig:a-roots-a3-KNY}	
\end{figure}

\subsection{The Okamoto Hamiltonian system} 
\label{sub:OK-5}
\begin{notation*}
For the Okamoto system we use the following notation: coordinates $(f,g)$, parameters are the Okamoto parameters \eqref{eq:Ok-pars-PV};
base points $q_{i}$, exceptional divisors $F_{i}$.	
\end{notation*}

The space of initial conditions for the system \eqref{eq:Ok-Ham5-sys} is given on Figure~\ref{fig:Ok-soic-5}, where the basepoints are
\begin{equation}\label{eq:Ok-pts-5}
	\begin{tikzpicture}
	\node (q1) at (0,0) {$q_{1}(\infty,0)$}; 
	\node (q2) at (2,0.7) {$q_{2}\left(0,\frac{\theta + \kappa_{0} - \kappa_{\infty}}{2}\right),$}; 
	\node (q3) at (2,-0.7) {$q_{3}\left(0,\frac{\theta + \kappa_{0} + \kappa_{\infty}}{2}\right),$}; 
	\draw[->] (q2)--(1,0)--(q1);	\draw[->] (q3)--(1,0)--(q1); 
	\node (q4) at (5,0.7) {$q_{4}(0,\infty)$};
	\node (q5) at (7,0.7) {$q_{5}(\kappa_{0},0),$};
	\draw[->] (q5)--(q4);
	\node (q6) at (5,-0.7) {$q_{6}(1,\infty)$}; \node (q7) at (7.5,-0.7) {$q_{7}(u_{6}=0,v_{6}=0)$};
	\node (q8) at (10,-0.7) {$q_{8}(t \eta,0)$}; \node (q9) at (12.5,-0.7) {$q_{9}(t(1 + \theta)\eta,0).$};
	\draw[->] (q7)--(q6);	\draw[->] (q8)--(q7); \draw[->] (q9)--(q8);
	\end{tikzpicture}
\end{equation}
Note that we have nine basepoints and two $-3$ curves, so this surface is not minimal and we need to blow down the $-1$-curve $H_{f} - F_{1}$.

\begin{figure}[ht]
	\begin{tikzpicture}[>=stealth,basept/.style={circle, draw=red!100, fill=red!100, thick, inner sep=0pt,minimum size=1.2mm}]
	\begin{scope}[xshift=0cm,yshift=0cm]
	\draw [black, line width = 1pt] (-0.4,0) -- (2.9,0)	node [pos=0,left] {\small $H_{g}$} node [pos=1,right] {\small $g=0$};
	\draw [black, line width = 1pt] (-0.4,2.5) -- (2.9,2.5) node [pos=0,left] {\small $H_{g}$} node [pos=1,right] {\small $g=\infty$};
	\draw [black, line width = 1pt] (0,-0.4) -- (0,2.9) node [pos=0,below] {\small $H_{f}$} node [pos=1,above] {\small $f=0$};
	\draw [black, line width = 1pt] (2.5,-0.4) -- (2.5,2.9) node [pos=0,below] {\small $H_{f}$} node [pos=1,above] {\small $f=\infty$};
	\node (q1) at (2.5,0) [basept,label={[xshift = -8pt, yshift=-15pt] \small $q_{1}$}] {};
	\node (q2) at (3.2,0.5) [basept,label={[yshift=0pt] \small $q_{2}$}] {};
	\node (q3) at (3.2,-0.5) [basept,label={[yshift=-15pt] \small $q_{3}$}] {};
	\node (q4) at (0,2.5) [basept,label={[xshift = 8pt, yshift=-15pt] \small $q_{4}$}] {};
	\node (q5) at (-.5,2) [basept,label={[yshift=-15pt] \small $q_{5}$}] {};
	\node (q6) at (0.5,2.5) [basept,label={[yshift=0pt] \small $q_{6}$}] {};
	\node (q7) at (1.2,2) [basept,label={[yshift=-15pt] \small $q_{7}$}] {};
	\node (q8) at (1.7,2) [basept,label={[yshift=-15pt] \small $q_{8}$}] {};
	\node (q9) at (2.2,2) [basept,label={[yshift=-15pt] \small $q_{9}$}] {};
	\draw [red, line width = 0.8pt, ->] (q2) -- (q1);
	\draw [red, line width = 0.8pt, ->] (q3) -- (q1);
	\draw [red, line width = 0.8pt, ->] (q5) -- (q4);	
	\draw [red, line width = 0.8pt, ->] (q7) -- (0.8,2.5) -- (q6);
	\draw [red, line width = 0.8pt, ->] (q8) -- (q7);	
	\draw [red, line width = 0.8pt, ->] (q9) -- (q8);
	\end{scope}
	\draw [->] (6.5,1)--(4.5,1) node[pos=0.5, below] {$\operatorname{Bl}_{q_{1}\cdots q_{9}}$};
	\begin{scope}[xshift=8cm,yshift=0cm]
	\draw [red, line width = 1pt] (-0.4,0) -- (3.2,0)	node [pos=0, left] {\small $H_{g}-F_{1}$};
	\draw [red, line width = 1pt] (0,-0.4) -- (0,2.4) node [pos=0, below] {\small $H_{f}-F_{4}$};
	\draw [blue, line width = 1pt] (-0.4,1.6) -- (1,3) node [pos=0, left] {\small $F_{4}-F_{5}$};
	\draw [red, line width = 1pt] (0.7,2) -- (-0.1,2.8) node [pos=1, left] {\small $F_{5}$};
	\draw [teal, line width = 1pt] (0.3,2.7) -- (4,2.7) node [pos=1,right] {\small $H_{g} - F_{4} - F_{6} - F_{7}$};
	\draw [teal, line width = 1pt] (2.5,-0.2) -- (4,1.3) node [pos=1,right] {\small $F_{1} - F_{2} - F_{3}$};
	\draw [red, line width = 1pt] (3.8,0.4) -- (3.8,3) node [pos=1, above] {\small $H_{f}-F_{1}$};
	\draw [red, line width = 1pt] (2.7,0.7) -- (3.5,-0.1) node [pos=1, below] {\small $F_{2}$};
	\draw [red, line width = 1pt] (3,1) -- (3.8,0.2) node [pos=1, right] {\small $F_{3}$};
	\draw [red, line width = 1pt] (1.5,-0.4) -- (1.5,1) node [pos=0, below] {\small $H_{f}-F_{6}$};
	\draw [blue, line width = 1pt] (0.9,1.4) -- (1.7,0.6) node [pos=1, xshift = 12pt, yshift=-5pt ] {\small $F_{6}-F_{7}$};	
	\draw [blue, line width = 1pt] (0.8,1) -- (2.8,3) node [pos=1,yshift=9pt, xshift = -8pt] {\small $F_{7}-F_{8}$};	
	\draw [blue, line width = 1pt] (2,2.5) -- (2.8,1.7) node [pos=1, xshift = 5pt, yshift=-5pt ] {\small $F_{8}-F_{9}$};	
	\draw [red, line width = 1pt] (2,1.5) -- (2.8,2.3) node [pos=0, below] {\small $F_{9}$};
	\end{scope}
	\end{tikzpicture}
	\caption{The Space of Initial Conditions for the Okamoto Hamiltonian System}
	\label{fig:Ok-soic-5}
\end{figure}

We then get the following Lemma.

\begin{lemma}\label{lem:KNY-to-Ok-5} The change of bases for Picard lattices between the standard Kajiwara-Noumi-Yamada (with an additional blowup point) 
	and the Okamoto surfaces is given by 
	\begin{equation}\label{eq:basis-KNY-Ok-5}
		\begin{aligned}
			\mathcal{H}_{q} & = \mathcal{H}_{f}, &\qquad 
				\mathcal{H}_{f} &= \mathcal{H}_{q},\\
			\mathcal{H}_{p} &= 2\mathcal{H}_{f} + \mathcal{H}_{g} - \mathcal{F}_{1} - \mathcal{F}_{2} - \mathcal{F}_{6} - \mathcal{F}_{7},  &\qquad 	
				\mathcal{H}_{g} & = 2 \mathcal{H}_{q} + \mathcal{H}_{p} - \mathcal{E}_{3} - \mathcal{E}_{4} - \mathcal{E}_{5} - \mathcal{E}_{9}, \\
			\mathcal{E}_{1} &= \mathcal{F}_{8}, &\qquad 
				\mathcal{F}_{1}	&= \mathcal{H}_{q} - \mathcal{E}_{9},\\ 
			\mathcal{E}_{2} &= \mathcal{F}_{9}, &\qquad 
				\mathcal{F}_{2}	&= \mathcal{H}_{q} - \mathcal{E}_{5},\\ 
			\mathcal{E}_{3} &= \mathcal{H}_{f} - \mathcal{F}_{7}, &\qquad 
				\mathcal{F}_{3}	&= \mathcal{E}_{6},\\ 
			\mathcal{E}_{4} &= \mathcal{H}_{f} - \mathcal{F}_{6}, &\qquad 
				\mathcal{F}_{4}	&= \mathcal{E}_{7},\\ 
			\mathcal{E}_{5} &= \mathcal{H}_{f} - \mathcal{F}_{2}, &\qquad 
				\mathcal{F}_{5}	&= \mathcal{E}_{8},\\ 
			\mathcal{E}_{6} &= \mathcal{F}_{3}, &\qquad 
				\mathcal{F}_{6}	&= \mathcal{H}_{q} - \mathcal{E}_{4},\\ 
			\mathcal{E}_{7} &= \mathcal{F}_{4}, &\qquad 
				\mathcal{F}_{7}	&= \mathcal{H}_{q} - \mathcal{E}_{3},\\ 
			\mathcal{E}_{8} &= \mathcal{F}_{5}, &\qquad 
				\mathcal{F}_{8}	&= \mathcal{E}_{1},\\ 
			\mathcal{E}_{9} &= \mathcal{H}_{f} - \mathcal{F}_{1}, &\qquad 
				\mathcal{F}_{9}	&= \mathcal{E}_{2}.
		\end{aligned}
	\end{equation}
	This results in the following correspondences between the surface roots with additional blowups (so strictly speaking these are no longer roots but 
	 classes of curves with indices $-2$ and $-3$), 
	\begin{equation}\label{eq:geom-KNY-Ok-5}
	\begin{aligned}
		\delta_{0} &= \mathcal{E}_{1} - \mathcal{E}_{2} = \mathcal{F}_{8} - \mathcal{F}_{9}, &\qquad
			\delta_{3} &=\mathcal{H}_{p} - \mathcal{E}_{5} - \mathcal{E}_{7} - \mathcal{E}_{9} = 
				\mathcal{H}_{g} - \mathcal{F}_{4} - \mathcal{F}_{6} - \mathcal{F}_{7}, \\
		\delta_{1} &= \mathcal{E}_{3} - \mathcal{E}_{4} = \mathcal{F}_{6} - \mathcal{F}_{7}, &\qquad
			\delta_{4} &= \mathcal{E}_{5} - \mathcal{E}_{6} - \mathcal{E}_{9} = \mathcal{F}_{1} - \mathcal{F}_{2} - \mathcal{F}_{3}, \\
		\delta_{2} &=  \mathcal{H}_{q} - \mathcal{E}_{1} - \mathcal{E}_{3} = \mathcal{F}_{7} - \mathcal{F}_{8}, &\qquad
			\delta_{5} &= \mathcal{E}_{7} - \mathcal{E}_{8} = \mathcal{F}_{4} - \mathcal{F}_{5};
	\end{aligned}
	\end{equation}
	and the symmetry roots,
	\begin{equation}\label{eq:sym-KNY-Ok-5}
	\begin{aligned}
			\alpha_{0} &= \mathcal{H}_{p} - \mathcal{E}_{1} - \mathcal{E}_{2}
		= 2 \mathcal{H}_{f} + \mathcal{H}_{g} - \mathcal{F}_{1} - \mathcal{F}_{2} - \mathcal{F}_{6} - \mathcal{F}_{7} - \mathcal{F}_{8} - \mathcal{F}_{9}, \\
			\alpha_{1} &= \mathcal{H}_{q} - \mathcal{E}_{5} - \mathcal{E}_{6} =  \mathcal{F}_{2} - \mathcal{F}_{3}, \\
			\alpha_{2} &= \mathcal{H}_{p} - \mathcal{E}_{3} - \mathcal{E}_{4} = 
			\mathcal{H}_{g} - \mathcal{F}_{1} - \mathcal{F}_{2},\\
			\alpha_{3} &= \mathcal{H}_{q} - \mathcal{E}_{7} - \mathcal{E}_{8} = 
			\mathcal{H}_{f} - \mathcal{F}_{4} - \mathcal{F}_{5}.
		\end{aligned}
	\end{equation}
	The symplectic form is the standard one, $\symp{Ok} = dg\wedge df$, and the root variables are 
	\begin{equation*}
		a_{0} = 1 + \frac{\theta-\kappa_{0} + \kappa_{\infty}}{2},\quad  a_{1} = - \kappa_{\infty},\quad
		a_{2} = \frac{-\theta - \kappa_{0} + \kappa_{\infty}}{2},\quad a_{3} = \kappa_{0},
	\end{equation*}
	where the standard normalization $a_{0}+a_{1}+a_{2}+a_{3}=1$ is equivalent to rescaling $\eta = -1$ in \eqref{eq:Ok-pars-PV}.
\end{lemma}

The birational change of variables corresponding to this change of bases of the Picard lattice  is given in the following Lemma.

\begin{lemma}\label{lem:coords-KNY-Ok-5} The change of coordinates and parameter matching between the Kajiwara-Noumi-Yamada 
	and Okamoto Hamiltonian systems is given by
    \begin{equation}\label{eq:KNYtoOk-5}
   	 \left\{\begin{aligned}
   	 	q(f,g,t)&=\frac{1}{1-f},\\
   		p(f,g,t)&= \rlap{$\displaystyle\frac{(f-1)(2(f-1)g - \theta- \kappa_{0}+\kappa_{\infty})}{2}$,}\\
		a_{0}&= 1 + \frac{\theta- \kappa_{0}+\kappa_{\infty}}{2},&\quad a_{1}&=-\kappa_{\infty},\\  
		a_{2}&= - \frac{\theta + \kappa_{0}-\kappa_{\infty}}{2},&\quad a_{3}&=\kappa_{0}.
   	 \end{aligned}\right.
    \qquad\text{and} \qquad 
    	\left\{\begin{aligned}
   	 	f(q,p,t)&=\rlap{$\displaystyle 1 - \frac{1}{q} \quad\text{\emph{(as expected)}}$},\\
   		g(q,p,t)&=q(pq + a_{2}),\\
		\kappa_{0} &= a_{3},&\quad \kappa_{\infty}&=-a_{1},\\  
		\theta &= -a_{1} - 2 a_{2} - a_{3},&\quad \eta &= -1.
    	\end{aligned}\right.
    \end{equation}
\end{lemma}

Since the change of variables is time-independent, the Hamiltonians match (up to $t$-dependent terms), 
\begin{equation}\label{eq:Ham-KNY-Ok-5}
	\Ham{KNY}{V}(q,p;t) = \Ham{Ok}{V}(f(q,p,t),g(q,p,t);t) + a_{2} + \frac{a_{2}(1 - a_{0})}{t}.
\end{equation}


\subsection{The Jimbo-Miwa Hamiltonian System} 
\label{sub:JM-5}
\begin{notation*}
For the Jimbo-Miwa system we use the following notation: coordinates $(y,z)$, parameters \eqref{eq:JM-pars-PV},
base points $x_{i}$, exceptional divisors $M_{i}$.	
\end{notation*}

The space of initial conditions for the system \eqref{eq:JM-Ham5-sys} is given on Figure~\ref{fig:soic-JM}, where the basepoints are
\begin{equation}\label{eq:JM-pts}
\begin{aligned}
	&x_{1}(0,-\Theta_{0}),\quad x_{2}\left(0,-\frac{\Theta_{0} + \Theta_{1} + \Theta_{\infty}}{2}\right),\quad x_{3}(\infty,0),\quad 
	x_{4}\left(\infty,-\frac{\Theta_{0} - \Theta_{1} + \Theta_{\infty}}{2}\right),\\
	&x_{5}(1,\infty)\leftarrow x_{6}(u_{5} = 0, v_{5} = 0)\leftarrow x_{7}(t,0)\leftarrow x_{8}\Big(t\left(t + \Theta_{0} + \Theta_{1} - 1\right),0\Big).
\end{aligned}
\end{equation}

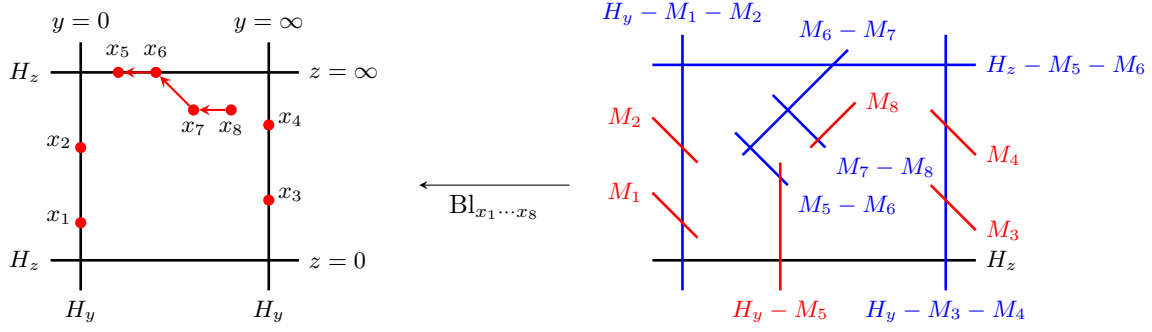
\begin{figure}[ht]
	\begin{tikzpicture}[>=stealth,basept/.style={circle, draw=red!100, fill=red!100, thick, inner sep=0pt,minimum size=1.2mm}]
	\begin{scope}[xshift=0cm,yshift=0cm]
	\draw [black, line width = 1pt] (-0.4,0) -- (2.9,0)	node [pos=0,left] {\small $H_{z}$} node [pos=1,right] {\small $z=0$};
	\draw [black, line width = 1pt] (-0.4,2.5) -- (2.9,2.5) node [pos=0,left] {\small $H_{z}$} node [pos=1,right] {\small $z=\infty$};
	\draw [black, line width = 1pt] (0,-0.4) -- (0,2.9) node [pos=0,below] {\small $H_{y}$} node [pos=1,above] {\small $y=0$};
	\draw [black, line width = 1pt] (2.5,-0.4) -- (2.5,2.9) node [pos=0,below] {\small $H_{y}$} node [pos=1,above] {\small $y=\infty$};
	\node (x1) at (0,0.5) [basept,label={[left] \small $x_{1}$}] {};
	\node (x2) at (0,1.5) [basept,label={[left] \small $x_{2}$}] {};
	\node (x3) at (2.5,0.8) [basept,label={[right] \small $x_{3}$}] {};
	\node (x4) at (2.5,1.8) [basept,label={[right] \small $x_{4}$}] {};
	\node (x5) at (0.5,2.5) [basept,label={[above] \small $x_{5}$}] {};
	\node (x6) at (1,2.5) [basept,label={[above] \small $x_{6}$}] {};
	\node (x7) at (1.5,2) [basept,label={[xshift = 0pt, yshift=-15pt] \small $x_{7}$}] {};
	\node (x8) at (2,2) [basept,label={[xshift = 0pt, yshift=-15pt] \small $x_{8}$}] {};
	\draw [red, line width = 0.8pt, ->] (x6) -- (x5);
	\draw [red, line width = 0.8pt, ->] (x7) -- (x6);
	\draw [red, line width = 0.8pt, ->] (x8) -- (x7);
	\end{scope}
	\draw [->] (6.5,1)--(4.5,1) node[pos=0.5, below] {$\operatorname{Bl}_{x_{1}\cdots x_{8}}$};
	\begin{scope}[xshift=8cm,yshift=0cm]
	\draw [black, line width = 1pt] (-0.4,0) -- (3.9,0)	node [pos=1, right] {\small $H_{z}$};
	\draw [blue, line width = 1pt] (-0.4,2.6) -- (3.9,2.6)	node [pos=1, right] {\small $H_{z}-M_{5}-M_{6}$};
	\draw [blue, line width = 1pt] (0,-0.4) -- (0,3)	node [pos=1, above] {\small $H_{y}-M_{1}-M_{2}$};
	\draw [blue, line width = 1pt] (3.5,-0.4) -- (3.5,3)	node [pos=0, below] {\small $H_{y}-M_{3}-M_{4}$};	
	\draw [red, line width = 1pt] (0.2,0.3) -- (-0.4,0.9)	node [pos=1, left] {\small $M_{1}$};	
	\draw [red, line width = 1pt] (0.2,1.3) -- (-0.4,1.9)	node [pos=1, left] {\small $M_{2}$};	
	\draw [red, line width = 1pt] (3.3,1) -- (3.9,0.4)	node [pos=1, right] {\small $M_{3}$};		
	\draw [red, line width = 1pt] (3.3,2) -- (3.9,1.4)	node [pos=1, right] {\small $M_{4}$};		
	\draw [blue, line width = 1pt] (2.2,2.8) -- (0.8,1.4)	node [pos=0, above] {\small $M_{6}-M_{7}$};	
	\draw [blue, line width = 1pt] (0.7,1.7) -- (1.4,1)	node [pos=1, below right] {\small $M_{5}-M_{6}$};	
	\draw [blue, line width = 1pt] (1.2,2.2) -- (1.9,1.5)	node [pos=1, below right] {\small $M_{7}-M_{8}$};	
	\draw [red, line width = 1pt] (1.3,-0.4) -- (1.3,1.3)	node [pos=0, below] {\small $H_{y} - M_{5}$};	
	\draw [red, line width = 1pt] (1.7,1.5) -- (2.3,2.1)	node [pos=1, right] {\small $M_{8}$};	
	\end{scope}
	\end{tikzpicture}
	\caption{The Space of Initial Conditions for the Jimbo-Miwa Hamiltonian System}
	\label{fig:soic-JM}
\end{figure}

From the geometric point of view, it is best to match it to the standard $D_{5}^{(1)}$ Sakai surface, which we do next.
\begin{lemma}\label{lem:JM-KNY-Ok} The change of bases for Picard lattices between the standard $D_{5}^{(1)}$ Kajiwara-Noumi-Yamada surface
	and the Jimbo-Miwa space of initial condition is given by 
	\begin{equation}\label{eq:basis-KNY-JM}
		\begin{aligned}
			\mathcal{H}_{q} & = \mathcal{H}_{y}, &\qquad 
				\mathcal{H}_{y} &= \mathcal{H}_{q},\\
			\mathcal{H}_{p} &= 2\mathcal{H}_{y} + \mathcal{H}_{z} - \mathcal{M}_{2} - \mathcal{M}_{4} - \mathcal{M}_{5} - \mathcal{M}_{6},  &\qquad 	
				\mathcal{H}_{z} & = 2 \mathcal{H}_{q} + \mathcal{H}_{p} - \mathcal{E}_{3} - \mathcal{E}_{4} - \mathcal{E}_{5} - \mathcal{E}_{7}, \\
			\mathcal{E}_{1} &= \mathcal{M}_{7}, &\qquad 
				\mathcal{M}_{1}	&= \mathcal{E}_{8},\\ 
			\mathcal{E}_{2} &= \mathcal{M}_{8}, &\qquad 
				\mathcal{M}_{2}	&= \mathcal{H}_{q} - \mathcal{E}_{7},\\ 
			\mathcal{E}_{3} &= \mathcal{H}_{y} - \mathcal{M}_{6}, &\qquad 
				\mathcal{M}_{3}	&= \mathcal{E}_{6},\\ 
			\mathcal{E}_{4} &= \mathcal{H}_{y} - \mathcal{M}_{5}, &\qquad 
				\mathcal{M}_{4}	&=\mathcal{H}_{q} - \mathcal{E}_{5},\\ 
			\mathcal{E}_{5} &= \mathcal{H}_{y} - \mathcal{M}_{4}, &\qquad 
				\mathcal{M}_{5}	&=\mathcal{H}_{q} -  \mathcal{E}_{4},\\ 
			\mathcal{E}_{6} &= \mathcal{M}_{3}, &\qquad 
				\mathcal{M}_{6}	&= \mathcal{H}_{q} - \mathcal{E}_{3},\\ 
			\mathcal{E}_{7} &=\mathcal{H}_{y} - \mathcal{M}_{2}, &\qquad 
				\mathcal{M}_{7}	&= \mathcal{E}_{1},\\ 
			\mathcal{E}_{8} &= \mathcal{M}_{1}, &\qquad 
				\mathcal{M}_{8}	&= \mathcal{E}_{2}.
		\end{aligned}
	\end{equation}
	This results in the following correspondences between the surface roots, 
	\begin{equation}\label{eq:geom-KNY-JM}
	\begin{aligned}
		\delta_{0} &= \mathcal{E}_{1} - \mathcal{E}_{2} = \mathcal{M}_{7} - \mathcal{M}_{8}, &\qquad
			\delta_{3} &=\mathcal{H}_{p} - \mathcal{E}_{5} - \mathcal{E}_{7}  = 
				\mathcal{H}_{z} - \mathcal{M}_{5} - \mathcal{M}_{6}, \\
		\delta_{1} &= \mathcal{E}_{3} - \mathcal{E}_{4} = \mathcal{M}_{5} - \mathcal{M}_{6}, &\qquad
			\delta_{4} &= \mathcal{E}_{5} - \mathcal{E}_{6}  = \mathcal{H}_{y}  - \mathcal{M}_{3} - \mathcal{M}_{4}, \\
		\delta_{2} &=  \mathcal{H}_{q} - \mathcal{E}_{1} - \mathcal{E}_{3} = \mathcal{M}_{6} - \mathcal{M}_{7}, &\qquad
			\delta_{5} &= \mathcal{E}_{7} - \mathcal{E}_{8} = \mathcal{H}_{y} - \mathcal{M}_{1} - \mathcal{M}_{2};
	\end{aligned}
	\end{equation}
	and the symmetry roots,
	\begin{equation}\label{eq:sym-KNY-JM}
	\begin{aligned}
			\alpha_{0} &= \mathcal{H}_{p} - \mathcal{E}_{1} - \mathcal{E}_{2}
		= 2 \mathcal{H}_{y} + \mathcal{H}_{z} - \mathcal{M}_{2} - \mathcal{M}_{4} - \mathcal{M}_{5} - \mathcal{M}_{6} - \mathcal{M}_{7} - \mathcal{M}_{8}, \\
			\alpha_{1} &= \mathcal{H}_{q} - \mathcal{E}_{5} - \mathcal{E}_{6} =  \mathcal{M}_{4} - \mathcal{M}_{3}, \\
			\alpha_{2} &= \mathcal{H}_{p} - \mathcal{E}_{3} - \mathcal{E}_{4} = 
			\mathcal{H}_{z} - \mathcal{M}_{2} - \mathcal{M}_{4},\\
			\alpha_{3} &= \mathcal{H}_{q} - \mathcal{E}_{7} - \mathcal{E}_{8} = 
			\mathcal{M}_{2} - \mathcal{M}_{1}.
		\end{aligned}
	\end{equation}
	The symplectic form is the logarithmic one, $\symp{JM} = (1/y)dy\wedge dg$, and the root variables are 
	\begin{equation*}
		a_{0} = 1 -\Theta_{0},\quad  a_{1} = \frac{\Theta_{0} - \Theta_{1} + \Theta_{\infty}}{2},\quad
		a_{2} = \Theta_{1},\quad a_{3} = \frac{\Theta_{0} - \Theta_{1} - \Theta_{\infty}}{2}.
	\end{equation*}
\end{lemma}

In the next Lemma we give the birational change of variables corresponding to this change of bases of the Picard lattice, as well as 
the resulting change of variables between the Jimbo-Miwa and the Okamoto systems.

\begin{lemma}\label{lem:coords-JM-KNY-Ok} The change of coordinates and parameter matching between the Kajiwara-Noumi-Yamada 
	and Jimbo-Miwa Hamiltonian systems is given by
    \begin{equation}\label{eq:KNYtoJM}
   	 \left\{\begin{aligned}
   	 	q(y,z,t)&=\frac{1}{1-y},\\
   		p(y,z,t)&= \rlap{$\displaystyle-\frac{(y-1)\left( (\Theta_{0} + \Theta_{\infty} + 2z)(y-1) -\Theta_{1}(y+1) \right)}{2y}$,}\\
		a_{0} &= 1 -\Theta_{0},&\quad  a_{1} &= \frac{\Theta_{0} - \Theta_{1} + \Theta_{\infty}}{2},\\  
		a_{2} &= \Theta_{1},& \quad a_{3} &= \frac{\Theta_{0} - \Theta_{1} - \Theta_{\infty}}{2}.
   	 \end{aligned}\right.
    \qquad\hskip0.7in\text{and} \quad 
    	\left\{\begin{aligned}
   	 	y(q,p,t)&=1 - \frac{1}{q}\\
   		z(q,p,t)&=-a_{1} -a_{2} q - (q-1)qp,\\
		\Theta_{0} &= a_{1} + a_{2} + a_{3},\\  
		\Theta_{1} &= a_{2},\quad \Theta_{\infty} = a_{1}-a_{3}.
    	\end{aligned}\right.
    \end{equation}
	
	Combining it with the change of variables \eqref{eq:KNYtoOk-5}, we get the change of variables between the Okamoto (with $\eta=-1$) 
	and the Jimbo-Miwa Hamiltonian systems,
    \begin{equation}\label{eq:JMtoOk}
   	 \left\{\begin{aligned}
   	 	y(f,g,t)&=f,\\
   		z(f,g,t)&= \frac{\theta + \kappa_{0}+\kappa_{\infty}}{2} - fg,\\
		\Theta_{0}&= \frac{\kappa_{0} - \kappa_{\infty} - \theta}{2},\\ 
		\Theta_{1}&= -\frac{\kappa_{0} - \kappa_{\infty} + \theta}{2},\\  
		\Theta_{\infty}&= - \kappa_{0} - \kappa_{\infty}.
   	 \end{aligned}\right.
    \qquad\text{and} \qquad 
    	\left\{\begin{aligned}
   	 	f(y,z,t)&=y,\\
   		g(y,z,t)&=  - \frac{2z + \Theta_{0} + \Theta_{1} + \Theta_{\infty}}{2y},\\
		\kappa_{0} &= \frac{\Theta_{0}-\Theta_{1} -\Theta_{\infty}}{2},\\
		 \kappa_{\infty}&=- \frac{\Theta_{0}-\Theta_{1} +\Theta_{\infty}}{2},\\  
		\theta &= -\Theta_{0}-\Theta_{1}.
    	\end{aligned}\right.
    \end{equation}
\end{lemma}

Since the change of variables is time-independent, the Hamiltonians match (up to $t$-dependent terms), 
\begin{equation}\label{eq:Ham-JM-Ok}
	\Ham{JM}{V}(y,z;t) = \Ham{Ok}{V}(f(y,z,t),g(y,z,t);t) + \frac{\Theta_{1}}{2} + \frac{(\Theta_{0} + \Theta_{1})^{2} - \Theta_{\infty}^{2}}{4t}.
\end{equation}

\subsection{The Filipuk-\.{Z}o\l\c{a}dek Hamiltonian System} 
\label{sub:FZ-5}
\begin{notation*}
For the Filipuk-\.{Z}o\l\c{a}dek system we use the following notation: coordinates $(x,y)$, parameters \eqref{eq:Ok-pars-PV},
base points $z_{i}$, exceptional divisors $K_{i}$.	
\end{notation*}

The space of initial conditions for the system \eqref{eq:FZ-Ham5-sys} is given on Figure~\ref{fig:soic-FZ}, where the basepoints are
\begin{equation}\label{eq:FZ-pts}
	\begin{tikzpicture}
	\node (z1) at (0,0) {$z_{1}(\infty,0)$}; 
	\node (z2) at (2,0.7) {$z_{2}\left(0,\frac{\kappa_{\infty}}{2}\right),$}; 
	\node (z3) at (2,-0.7) {$z_{3}\left(0,-\frac{\kappa_{\infty}}{2}\right),$}; 
	\draw[->] (z2)--(1,0)--(z1);	\draw[->] (z3)--(1,0)--(z1); 
	\node (z4) at (3.5,0) {$z_{4}(0,\infty)$}; 
	\node (z5) at (5.5,0.7) {$z_{5}\left(\frac{\kappa_{0}}{2},0\right),$}; 
	\node (z6) at (5.5,-0.7) {$z_{6}\left(-\frac{\kappa_{0}}{2},0\right),$}; 
	\draw[->] (z5)--(4.5,0)--(z4);	\draw[->] (z6)--(4.5,0)--(z4); 
	\node (z78) at (8.5,0) {$z_{7}(1,\infty)\leftarrow z_{8}(u_{7}=0,v_{7}=0)$}; 
	\node (z910) at (13,0.7) {$z_{9}\left(\frac{t\eta}{2}, 0\right)\leftarrow z_{10}(\left(\frac{\theta+2)t \eta}{4},0\right),$}; 
	\node (z1112) at (13,-0.7) {$z_{11}\left(2\frac{t\eta}{2}, 0\right)\leftarrow z_{12}\left( \frac{\theta t \eta}{4},0\right).$}; 
	\draw[->] (z910)--(11.5,0)--(z78);	\draw[->] (z1112)--(11.5,0)--(z78);
	\end{tikzpicture}
\end{equation}
Note that we have twelve basepoints and four $-3$ curves, so this surface is not minimal and we need to blow down some curves. 
From Figure~\ref{fig:soic-FZ} it is pretty clear that we should blow down the $-1$- curves $H_{x} - K_{1}$ and $H_{x}-K_{4}$, 
as well as the cascade $H_{x}-K_{7}$ and $K_{7}-K_{8}$. As usual, instead we blow up the standard surface at additional points 
$p_{9},\ldots p_{12}$.

\begin{figure}[ht]
	\begin{tikzpicture}[>=stealth,basept/.style={circle, draw=red!100, fill=red!100, thick, inner sep=0pt,minimum size=1.2mm}]
	\begin{scope}[xshift=0cm,yshift=0cm]
	\draw [black, line width = 1pt] (-0.4,0) -- (3.9,0)	node [pos=0,left] {\small $H_{g}$} node [pos=1,right] {\small $g=0$};
	\draw [black, line width = 1pt] (-0.4,2.5) -- (3.9,2.5) node [pos=0,left] {\small $H_{g}$} node [pos=1,right] {\small $g=\infty$};
	\draw [black, line width = 1pt] (0,-0.4) -- (0,2.9) node [pos=0,below] {\small $H_{f}$} node [pos=1,above] {\small $f=0$};
	\draw [black, line width = 1pt] (3.5,-0.4) -- (3.5,2.9) node [pos=0,below] {\small $H_{f}$} node [pos=1,above] {\small $f=\infty$};
	\node (z1) at (3.5,0) [basept,label={[xshift = -8pt, yshift=-15pt] \small $z_{1}$}] {};
	\node (z2) at (4,0.5) [basept,label={[yshift=0pt] \small $z_{2}$}] {};
	\node (z3) at (4,-0.5) [basept,label={[yshift=-15pt] \small $z_{3}$}] {};
	\node (z4) at (0,2.5) [basept,label={[xshift = -6pt, above] \small $z_{4}$}] {};
	\node (z5) at (-.5,2) [basept,label={[yshift=-15pt] \small $z_{5}$}] {};
	\node (z6) at (.5,2) [basept,label={[yshift=-15pt] \small $z_{6}$}] {};
	\node (z7) at (0.6,2.5) [basept,label={[above] \small $z_{7}$}] {};
	\node (z8) at (1.2,2.5) [basept,label={[above] \small $z_{8}$}] {};
	\node (z9) at (2,3) [basept,label={[above] \small $z_{9}$}] {};
	\node (z10) at (2.6,3) [basept,label={[above] \small $z_{10}$}] {};
	\node (z11) at (2,2) [basept,label={[yshift=-15pt] \small $z_{11}$}] {};
	\node (z12) at (2.6,2) [basept,label={[yshift=-15pt] \small $z_{12}$}] {};
	\draw [red, line width = 0.8pt, ->] (z2) -- (z1);
	\draw [red, line width = 0.8pt, ->] (z3) -- (z1);
	\draw [red, line width = 0.8pt, ->] (z5) -- (z4);	
	\draw [red, line width = 0.8pt, ->] (z6) -- (z4);	
	\draw [red, line width = 0.8pt, ->] (z8) -- (z7);	
	\draw [red, line width = 0.8pt, ->] (z10) -- (z9);		
	\draw [red, line width = 0.8pt, ->] (z12) -- (z11);		
	\draw [red, line width = 0.8pt, ->] (z9) -- (1.6,2.5) -- (z8);
	\draw [red, line width = 0.8pt, ->] (z11) -- (1.6,2.5) -- (z8);
	\end{scope}
	\draw [->] (6.5,1)--(4.5,1) node[pos=0.5, below] {$\operatorname{Bl}_{z_{1}\cdots z_{12}}$};
	\begin{scope}[xshift=7.8cm,yshift=0cm]
	\draw [red, line width = 1pt] (-0.4,0) -- (3.7,0)	node [pos=0, left] {\small $H_{y}-K_{1}$};
	\draw [red, line width = 1pt] (0,-0.4) -- (0,2.1) node [pos=0, below, xshift=-10pt] {\small $H_{x}-K_{4}$};
	\draw [teal, line width = 1pt] (-0.2,1.4) -- (1.3,2.9) node [pos=0, left] {\small $K_{4}-K_{5}-K_{6}$};
	\draw [red, line width = 1pt] (0.7,1.6) -- (-0.1,2.4) node [pos=1, left] {\small $K_{5}$};
	\draw [red, line width = 1pt] (1.1,2) -- (0.3,2.8) node [pos=1, above] {\small $K_{6}$};
	\draw [teal, line width = 1pt] (0.6,2.7) -- (4.5,2.7) node [pos=1,right] {\small $H_{y} - K_{4} - K_{7} - K_{8}$};
	\draw [teal, line width = 1pt] (3,-0.2) -- (4.5,1.3) node [pos=1,right] {\small $K_{1} - K_{2} - K_{3}$};
	\draw [red, line width = 1pt] (4.3,0.6) -- (4.3,3) node [pos=1, above, xshift=10pt] {\small $H_{x}-K_{1}$};
	\draw [red, line width = 1pt] (3.2,0.7) -- (4,-0.1) node [pos=1, below] {\small $K_{2}$};
	\draw [red, line width = 1pt] (3.6,1.1) -- (4.4,0.3) node [pos=1, right] {\small $K_{3}$};
	\draw [red, line width = 1pt] (1.5,-0.4) -- (1.5,1) node [pos=0, below] {\small $H_{x}-K_{7}$};
	\draw [teal, line width = 1pt] (0.8,1) -- (2.8,3) node [pos=1, xshift = 0pt, above] {\small $K_{8}-K_{9}-K_{11}$};	
	\draw [blue, line width = 1pt] (0.9,1.4) -- (1.7,0.6) node [pos=1, xshift = 15pt, yshift=-5pt ] {\small $K_{7}-K_{8}$};	
	\draw [blue, line width = 1pt] (1.5,2) -- (2.3,1.2) node [pos=1, xshift = 15pt, yshift=-5pt ] {\small $K_{9}-K_{10}$};
	\draw [blue, line width = 1pt] (2.1,2.6) -- (2.9,1.8) node [pos=1, xshift = 15pt, yshift=-5pt ] {\small $K_{11}-K_{12}$};	
	\draw [red, line width = 1pt] (1.8,1.2) -- (2.3,1.7) node [pos=1, xshift=-3pt,yshift=5pt] {\small $K_{10}$};
	\draw [red, line width = 1pt] (2.6,1.8) -- (3.1,2.3) node [pos=1, right] {\small $K_{12}$};
	\end{scope}
	\end{tikzpicture}
	\caption{The Space of Initial Conditions for the Filipuk-\.{Z}o\l\c{a}dek Hamiltonian System}
	\label{fig:soic-FZ}
\end{figure}
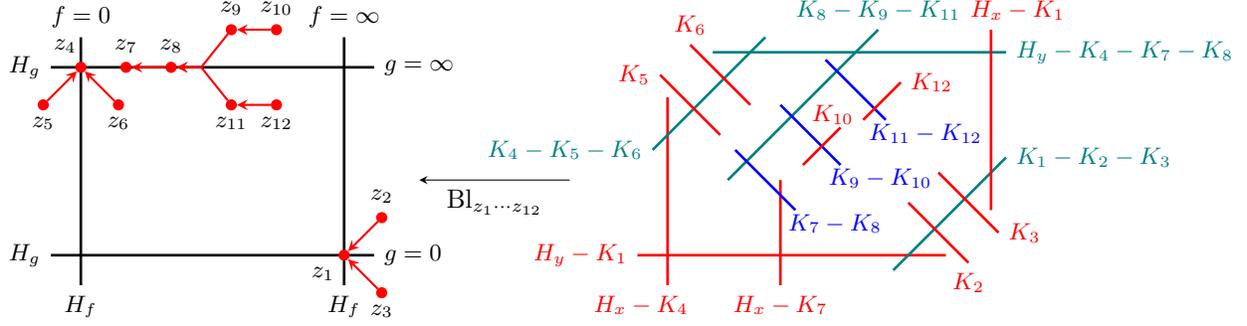

\begin{lemma}\label{lem:FZ-KNY-Ok} The change of bases for Picard lattices between the standard Kajiwara-Noumi-Yamada (with four additional blowup points) 
	and the Filipuk-\.{Z}o\l\c{a}dek surfaces is given by 
	\begin{equation}\label{eq:basis-KNY-FZ}
		\begin{aligned}
			\mathcal{H}_{q} & = \mathcal{H}_{x}, &
			\mathcal{H}_{x} &= \mathcal{H}_{q},\\
			\mathcal{H}_{p} &= 3\mathcal{H}_{x} + \mathcal{H}_{y} - \mathcal{K}_{1} - \mathcal{K}_{3} - \mathcal{K}_{4} 
				- \mathcal{K}_{6}- \mathcal{K}_{7} - \mathcal{K}_{8},  & 	
			\mathcal{H}_{y} & = 3 \mathcal{H}_{q} + \mathcal{H}_{p} - \mathcal{E}_{5} - \mathcal{E}_{7} - \mathcal{E}_{9} - \mathcal{E}_{10}
			 - \mathcal{E}_{11}  - \mathcal{E}_{12}, \\
			\mathcal{E}_{1} &= \mathcal{K}_{9}, &
			\mathcal{K}_{1}	&= \mathcal{H}_{q} - \mathcal{E}_{12},\\ 
			\mathcal{E}_{2} &= \mathcal{K}_{10}, &
			\mathcal{K}_{2}	&= \mathcal{E}_{6},\\ 
			\mathcal{E}_{3} &=  \mathcal{K}_{11}, &
			\mathcal{K}_{3}	&= \mathcal{H}_{q} - \mathcal{E}_{5},\\ 
			\mathcal{E}_{4} &= \mathcal{K}_{12}, &
			\mathcal{K}_{4}	&= \mathcal{H}_{q} - \mathcal{E}_{11},\\ 
			\mathcal{E}_{5} &= \mathcal{H}_{x} - \mathcal{K}_{3}, &
			\mathcal{K}_{5}	&= \mathcal{E}_{8},\\ 
			\mathcal{E}_{6} &= \mathcal{K}_{2}, &
			\mathcal{K}_{6}	&= \mathcal{H}_{q} - \mathcal{E}_{7},\\ 
			\mathcal{E}_{7} &= \mathcal{H}_{x} - \mathcal{K}_{6}, &
			\mathcal{K}_{7}	&= \mathcal{H}_{q} - \mathcal{E}_{10},\\ 
			\mathcal{E}_{8} &= \mathcal{K}_{5}, &
			\mathcal{K}_{8}	&= \mathcal{H}_{q} - \mathcal{E}_{9},\\ 
			\mathcal{E}_{9} &= \mathcal{H}_{x} - \mathcal{K}_{8}, &
			\mathcal{K}_{9}	&= \mathcal{E}_{1},\\
			\mathcal{E}_{10} &= \mathcal{H}_{x} - \mathcal{K}_{7}, &
			\mathcal{K}_{10} &= \mathcal{E}_{2},\\
			\mathcal{E}_{11} &= \mathcal{H}_{x} - \mathcal{K}_{4}, &
			\mathcal{K}_{11} &= \mathcal{E}_{3},\\
			\mathcal{E}_{12} &= \mathcal{H}_{x} - \mathcal{K}_{1}, &
			\mathcal{K}_{12} &= \mathcal{E}_{4}.
		\end{aligned}
	\end{equation}
	This results in the following correspondences between the surface roots with additional blowups 
	(i.e., classes of curves with indices $-2$ and $-3$), 
	\begin{equation}\label{eq:geom-KNY-FZ}
	\begin{aligned}
		\delta_{0} &= \mathcal{E}_{1} - \mathcal{E}_{2} = \mathcal{K}_{9} - \mathcal{K}_{10}, &\quad 
			\delta_{3} &=\mathcal{H}_{p} - \mathcal{E}_{5} - \mathcal{E}_{7} - \mathcal{E}_{12} = 
				\mathcal{H}_{y} - \mathcal{K}_{4} - \mathcal{K}_{7} - \mathcal{K}_{8}, \\
		\delta_{1} &= \mathcal{E}_{3} - \mathcal{E}_{4} = \mathcal{K}_{11} - \mathcal{K}_{12}, &\quad
		\delta_{4} &= \mathcal{E}_{5} - \mathcal{E}_{6} - \mathcal{E}_{12} = \mathcal{K}_{1} - \mathcal{K}_{2} - \mathcal{K}_{3}, \\
		\delta_{2} &=  \mathcal{H}_{q} - \mathcal{E}_{1} - \mathcal{E}_{3} - \mathcal{E}_{9} = 
		\mathcal{K}_{8} - \mathcal{K}_{9} - \mathcal{K}_{11}, &\quad 
		\delta_{5} &= \mathcal{E}_{7} - \mathcal{E}_{8} -\mathcal{E}_{11} = \mathcal{K}_{4} - \mathcal{K}_{5} - \mathcal{K}_{6},
	\end{aligned}
	\end{equation}
	as well as $\mathcal{E}_{9}-\mathcal{E}_{10}=\mathcal{K}_{7}-\mathcal{K}_{8}$; and the symmetry roots,
	\begin{equation}\label{eq:sym-KNY-FZ}
	\begin{aligned}
			\alpha_{0} &= \mathcal{H}_{p} - \mathcal{E}_{1} - \mathcal{E}_{2}
			= 3 \mathcal{H}_{x} + \mathcal{H}_{y} - \mathcal{K}_{1} - \mathcal{K}_{3} - \mathcal{K}_{4} 
				- \mathcal{K}_{6} - \mathcal{K}_{7} - \mathcal{K}_{8} - \mathcal{K}_{9} - \mathcal{K}_{10}, \\
			\alpha_{1} &= \mathcal{H}_{q} - \mathcal{E}_{5} - \mathcal{E}_{6} =  \mathcal{K}_{3} - \mathcal{K}_{2}, \\
			\alpha_{2} &= \mathcal{H}_{p} - \mathcal{E}_{3} - \mathcal{E}_{4}
			= 3 \mathcal{H}_{x} + \mathcal{H}_{y} - \mathcal{K}_{1} - \mathcal{K}_{3} - \mathcal{K}_{4} 
				- \mathcal{K}_{6} - \mathcal{K}_{7} - \mathcal{K}_{8} - \mathcal{K}_{11} - \mathcal{K}_{12}, \\
			\alpha_{3} &= \mathcal{H}_{q} - \mathcal{E}_{7} - \mathcal{E}_{8} = 
			\mathcal{K}_{6} - \mathcal{K}_{5}.
		\end{aligned}
	\end{equation}
	The symplectic form is $\symp{FZ} = dy\wedge dx$, and the root variables are the same as in 
	the Okamoto case \eqref{eq:KNYtoOk-5}. 
\end{lemma}

\begin{lemma}\label{lem:coords-FZ-KNY-Ok} The change of coordinates and parameter matching between the Kajiwara-Noumi-Yamada 
	and Filipuk-\.{Z}o\l\c{a}dek (with $\eta=-1$) Hamiltonian systems is given by
    \begin{equation}\label{eq:KNY-FZ}
   	 \left\{\begin{aligned}
   	 	q(x,y,t)&=\frac{1}{1-x},\\
   		p(x,y,t)&= (x-1)^{2} y+ \frac{\eta t x + (x-1)(\kappa_{\infty} x - \kappa_{0})}{2x},\\
   	 \end{aligned}\right.
    \qquad\text{and} \quad 
    	\left\{\begin{aligned}
   	 	x(q,p,t)&=1 - \frac{1}{q},\\
   		y(q,p,t)&= q^{2} \left(p - \frac{\eta t}{2}\right) - \frac{a_{3} q^{2}}{2(q-1)} - \frac{a_{1}q}{2}.
    	\end{aligned}\right.
    \end{equation}
	
	Combining it with the change of variables \eqref{eq:KNYtoOk-5}, we get the change of variables between the Okamoto
	and the Filipuk-\.{Z}o\l\c{a}dek Hamiltonian systems,
    \begin{equation}\label{eq:FZtoOk}
   	 \left\{\begin{aligned}
   	 	x(f,g,t)&=f,\\
   		y(f,g,t)&= g - \frac{1}{2}\left(\frac{\theta}{f-1} + \frac{\eta t}{(f-1)^{2}} + \frac{\kappa_{0}}{f}\right),\\
   	 \end{aligned}\right.
    \qquad\text{and} \qquad 
    	\left\{\begin{aligned}
   	 	f(x,y,t)&=x,\\
   		g(x,y,t)&=  y +  \frac{1}{2}\left(\frac{\theta}{x-1} + \frac{\eta t}{(x-1)^{2}} + \frac{\kappa_{0}}{x}\right).
    	\end{aligned}\right.
    \end{equation}
\end{lemma}

Note that the change of variables is now time-dependent, and so there will be a correction term that, from
\begin{align*}
dg \wedge df - d\Ham{Ok}{V}\wedge dt &= dy \wedge dx - d\Ham{FZ}{V}\wedge dt,\\
dg\wedge df &=  dy\wedge dx + \frac{\eta}{2(x-1)^{2}}dt\wedge dx = 	
dy\wedge dx + d\left(\frac{\eta}{2(x-1)}\right)\wedge dt,
\end{align*}
should be equal to $-\frac{\eta}{2(x-1)}$. Indeed,
\begin{equation*}
	\Ham{FZ}{V}(x,y;t) = \Ham{Ok}{V}(f(x,y,t),g(x,y,t),t) - \frac{\eta}{2(x-1)} + \frac{\eta(\theta + \kappa_{0})}{2} + 
		\frac{\theta^{2} - \kappa_{0}^{2} - \kappa_{\infty}^{2}}{4t}.
\end{equation*}



\section{The Sixth Painlev\'e Equation $\Pain{VI}$} 
\label{sec:P6}

For the standard form of the sixth Painlev\'e equations we again take the one in \cite{Oka:1987:SPEISPEP},
\begin{equation}\label{eq:P6-std}
	\begin{aligned}
	\frac{d^{2} w}{dt^2} &= \frac{1}{2}\left(\frac{1}{w} + \frac{1}{w-1} + \frac{1}{w-t}\right)\left(\frac{dw}{dt}\right)^{2}
	-\left(\frac{1}{t} + \frac{1}{t-1} + \frac{1}{w-t}\right)\frac{dw}{dt} + \\
	&\qquad \frac{w(w-1)(w-t)}{t^{2}(t-1)^{2}}\left(\alpha + \beta \frac{t}{w^{2}} + \gamma \frac{t-1}{(w-1)^{2}} + 
	\delta\frac{t(t-1)}{(w-t)^{2}}\right), 		
	\end{aligned}
\end{equation}
where $t$ is an independent (complex) variable, $w(t)$ is the dependent variable, and $\alpha,\beta,\gamma,\delta\in \mathbb{C}$ are some 
parameters. Following Okamoto \cite{Oka:1987:SPEIFPEP}, it is convenient to introduce the \emph{Okamoto parameters}
$\kappa_{0}$, $\kappa_{1}$, $\kappa_{\infty}$, and $\theta$ given by
\begin{equation}\label{eq:Ok-pars-PVI}
	\alpha = \frac{\kappa_{\infty}^{2}}{2},\quad \beta = -\frac{\kappa_{0}^{2}}{2},\quad \gamma = \frac{\kappa_{1}^{2}}{2}, 
		\quad \delta = \frac{1-\theta^{2}}{2},
	\quad\text{as well as } \kappa = \frac{(\kappa_{0}+\kappa_{1}+\theta-1)^{2} - \kappa_{\infty}^{2}}{4},
\end{equation}
where the parameter notation is coming from a certain isomonodromy problem. 
Our reference Hamiltonian for $\Pain{VI}$ is the one in \cite{Oka:1987:SPEISPEP},
\begin{equation}\label{eq:Ok-Ham6}
	\Ham{Ok}{VI}(f,g;t) = \frac{1}{t(t-1)}\Big(f(f-1)(f-t)g^{2} - \big(
	\kappa_{0}(f-1)(f-t) + \kappa_{1}f(f-t)+ (\theta-1)f(f-1)
	\big)g + \kappa (f-t)\Big),
\end{equation}
which gives, w.r.t.~the standard symplectic form $\symp{Ok} = dg\wedge df$, the Hamiltonian system
  \begin{equation}\label{eq:Ok-Ham6-sys}
  	\left\{
 	\begin{aligned}
 		\frac{df}{dt} &= \frac{\partial \Ham{Ok}{VI}}{\partial g} 
			=\frac{1}{t(t-1)}\Big( (f-1)(f-t)(2 f g - \kappa_{0}) -\kappa_{1} f(f-t) - (\theta-1) f (f-1)\Big),\\
		\frac{dg}{dt} &= -\frac{\partial \Ham{Ok}{VI}}{\partial f} 
			=-\frac{1}{t(t-1)}\Big( \big( (3f^2 -2(t+1)f +t)g - (\kappa_{0} + \kappa_{1}) (2f-t) - (\theta-1)(2f-1) + \kappa_{0}\big)g + \kappa\Big).
 	\end{aligned}
	\right.
  \end{equation}
Eliminating the function $g=g(t)$ from these equations we get the $\Pain{VI}$ equation \eqref{eq:P6-std} for the function $f=f(t)$ with parameters given by 
\eqref{eq:Ok-pars-PVI}. The same Hamiltonian in \cite{KajNouYam:2017:GAPE} is written in terms of the root variables 
\begin{equation}\label{eq:pars-root2Ok-P6}
	a_{0} = \theta,\quad a_{1}=-\kappa_{\infty},\quad a_{2} = \frac{1-\theta-\kappa_{0}-\kappa_{1}+\kappa_{\infty}}{2},\quad 
	a_{3} = \kappa_{1},\quad a_{4}=\kappa_{0}
\end{equation}
with the usual normalization condition $a_{0}+a_{1}+2a_{2}+a_{3}+a_{4}=1$. From the geometric point of view it is better to rewrite 
everything in the coordinates $(q,p)=(f,fg)$\footnote{For consistency of notation throughout the paper our coordinate labels here are reverse to  those in 
\cite{KajNouYam:2017:GAPE}} since the corresponding space of initial conditions is \emph{minimal} and corresponds to the standard realization of the 
$D_{4}^{(1)}$ algebraic surface. Thus, we let the Kajiwara-Noumi-Yamada Hamiltonian be
\begin{equation}\label{eq:KNY-Ham6}
	\Ham{KNY}{VI}(f,g;t) = \frac{(q-1)(q-t)p}{t(t-1)}
	\left(\frac{p}{q}-\left(\frac{a_{0}-1}{q-t} + \frac{a_{3}}{q-1} + \frac{a_{4}}{q}\right)\right) + \frac{a_{2}(a_{1}+a_{2})(q-t)}{t(t-1)}.
\end{equation}
In these coordinates the symplectic form becomes logarithmic, $\symp{KNY}=(1/q)dp\wedge dq$ and the resulting Hamiltonian system is 
\begin{equation}\label{eq:KNY-Ham6-sys}
\left\{
\begin{aligned}
	\frac{dq}{dt} &= q \frac{\partial \Ham{KNY}{VI}}{\partial p} 
		=\frac{1}{t(t-1)}\Big( (q-1)(q-t)(2p-a_{4}) - (a_{0}-1)q(q-1) -a_{3} q(q-t) \Big),\\
	\frac{dp}{dt} &= -q\frac{\partial \Ham{KNY}{VI}}{\partial q} 
		=-\frac{1}{t(t-1)}\Big( q(p+a_{2})(p+a_{1}+a_{2}) - \frac{t p}{q}(p-a_{4})\Big),
\end{aligned}
\right.
\end{equation}
which reduces to the standard Painlev\'e $\Pain{VI}$ equation \eqref{eq:P6-std} with parameters
\begin{equation}\label{eq:pars-root-P6}
	\alpha = \frac{a_{1}^{2}}{2},\quad \beta=-\frac{a_{4}^{2}}{2},\quad \gamma=\frac{a_{3}^{2}}{2},\quad \delta=\frac{1-a_{0}^{2}}{2}.
\end{equation}

The Jimbo-Miwa Hamiltonian \cite{JimMiw:1981:MPDLODEWRC-II}
\begin{equation}\label{eq:JM-Ham6}
	\begin{aligned}
	\Ham{JM}{VI}(y,z;t) &= 
	\frac{1}{t(t-1)}\Big(
	y(y-1)(y-t)z^{2} -\big(\Theta_{0}(y-1)(y-t)+ \Theta_{1}y(y-t) + \Theta_{t}y(y-1)\big)z +\\ 
	&\qquad\qquad K_{1} K_{2} (y-t) + \Theta_{0}\Theta_{t}(t-1)	+ \Theta_{1} \Theta_{t} t
	\Big)
	\end{aligned} 
\end{equation}
is, up to purely $t$-dependent terms, the same as the Okamoto Hamiltonian \eqref{eq:Ok-Ham6}
\begin{equation*}
	\Ham{JM}{VI}(f,g;t) = \Ham{Ok}{VI}(f,g;t) - \Theta_{t}\left(\frac{\Theta_{0}}{t} + \frac{\Theta_{1}}{t-1}\right)
\end{equation*} 
under the parameter identification
\begin{equation*}
	\Theta_{0} = \kappa_{0},\quad \Theta_{1} = \kappa_{1},\quad \Theta_{t}=\theta-1,\quad \Theta_{\infty}=1 - \kappa_{\infty},\quad \text{and }
	4K_{1}K_{2} = (\Theta_{0}+\Theta_{1}+\Theta_{t})^{2} + (\Theta_{\infty}-1)^{2},
\end{equation*}
or
\begin{equation*}
	\alpha=\frac{(\Theta_{\infty}-1)^{2}}{2},\quad \beta = -\frac{\Theta_{0}^{2}}{2},\quad \gamma = \frac{\Theta_{1}^{2}}{2},\quad 
	\delta = \frac{1-(1+\Theta_{t})^{2}}{2},
\end{equation*}
and so we do not consider it further.

The Hamiltonian given in \cite{ItsPro:2018:SHPIF} 
\begin{equation}\label{eq:IP-Ham6}
	\begin{aligned}
	\Ham{IP}{VI}(x,y;t) &= y^{2} \frac{x(x-1)(x-t)}{t(t-1)} + y \frac{x(x-1)}{t(t-1)} + \frac{\Theta_{\infty}(1 - \Theta_{\infty})(x-t)}{t(t-1)}+ \\
	&\qquad  \frac{\Theta_{0}^{2}(x-t)}{x t (t-1)} - \frac{\Theta_{1}^{2}(x-t)}{(x-1)t(t-1)} + \frac{\Theta_{t}^{2}(t^{2}-x(2t-1))}{(x-t)t(t-1)},
	\end{aligned}
\end{equation}
where $\Theta_{i}$ are some isomonodromy parameters (different from the Jimbo-Miwa parameters above),
gives, w.r.t.~the standard symplectic form $\symp{IP}=dy\wedge dx$, the Hamiltonian system
\begin{equation}\label{eq:IP-Ham6-sys}
\left\{
\begin{aligned}
	\frac{dx}{dt} &= \frac{\partial \Ham{IP}{VI}}{\partial y} 
		=\frac{x(x-1)}{t(t-1)} + 2y \frac{x(x-1)(x-t)}{t(t-1)},\\
	\frac{dy}{dt} &= -\frac{\partial \Ham{IP}{VI}}{\partial x} 
		=\frac{y(t y - 1)(2x-1) + y^{2}x(2-3x)}{t(t-1)} - \frac{\Theta_{0}^{2}}{x^{2}(t-1)} - \frac{\Theta_{t}^{2}}{(t-x)^{2}} 
		+ \frac{\Theta_{1}^{2}}{t(x-1)^{2}} - \frac{\Theta_{\infty}(\Theta_{\infty}-1)}{t(t-1)},
\end{aligned}
\right.
\end{equation}
which reduces to $\Pain{VI}$ for $x(t)$ for parameters
\begin{equation}\label{eq:pars-P2IP-P6}
	\alpha = \frac{(2 \Theta_{\infty}-1)^{2}}{2},\quad \beta = - 2\Theta_{0}^{2},\quad \gamma = 2 \Theta_{1}^{2},\quad 
	\delta = \frac{1 - 4 \Theta_{t}^{2}}{2}.
\end{equation}
The matching between the Its-Prokhorov and Okamoto parameters is 
\begin{equation}\label{eq:pars-IP2Ok-P6}
	\Theta_{0} = \frac{\kappa_{0}}{2}, \quad \Theta_{1} = \frac{\kappa_{1}}{2},\quad \Theta_{t} = \frac{\theta}{2},\quad 
	\Theta_{\infty} = \frac{1 - \kappa_{\infty}}{2}.
\end{equation}

Closely related to the Its-Prokhorov Hamiltonian is the Filipuk-\.{Z}o\l\c{a}dek Hamiltonian \cite{ZolFil:2015:PEEIEF}
\begin{equation}\label{eq:FZ-Ham6-pre}
	\widetilde{\Ham{FZ}{VI}}(x,\tilde{y};t) = 
	\frac{1}{t(t-1)}\left(\frac{x(x-1)(x-t)\tilde{y}^{2}}{2} - \alpha x + \beta \frac{t}{x} + \gamma\frac{t-1}{x-1} 
	+ \delta \frac{t(t-1)}{x-t}\right),
\end{equation}
which we again rescale slightly and consider instead
\begin{equation}\label{eq:FZ-Ham6}
	\Ham{FZ}{VI}(x,y;t) =  \frac{1}{2}\widetilde{\Ham{FZ}{VI}}(x,2y;t)=
	\frac{1}{t(t-1)}\left( x(x-1)(x-t)y^{2} - \frac{\alpha}{2} x + \frac{\beta}{2} \frac{t}{x} + \frac{\gamma}{2}\frac{t-1}{x-1} 
	+ \frac{\delta}{2} \frac{t(t-1)}{x-t}\right).
\end{equation}
It is also convenient to rewrite this Hamiltonian in Okamoto parameters \eqref{eq:Ok-pars-PVI}, which gives us the following Hamiltonian system,
\begin{equation}\label{eq:FZ-Ham6-sys}
\left\{
\begin{aligned}
	\frac{dx}{dt} &= \frac{\partial \Ham{FZ}{VI}}{\partial y} 
		=\frac{2xy(x-1)(x-t)}{t(t-1)},\\
	\frac{dy}{dt} &= -\frac{\partial \Ham{FZ}{VI}}{\partial x} 
		=-\frac{y^{2}(3x^{2}-2(t+1)x+t)}{t(t-1)} + \frac{1}{4}\left(
		\frac{1-\theta^{2}}{(x-t)^{2}} - \frac{\kappa_{0}^{2}}{(t-1) x^{2}} + \frac{\kappa_{1}^{2}}{t(x-1)^{2}} + \frac{\kappa_{\infty}^{2}}{t(t-1)}
		\right).
\end{aligned}
\right.
\end{equation}
Next we construct the Okamoto spaces of initial conditions for each of these systems and use them to find the coordinate identification between them.

%

\subsection{The Kajiwara-Noumi-Yamada Hamiltonian system} 
\label{sub:KNY-6}
\begin{notation*}
For the Kajiwara-Noumi-Yamada system \eqref{eq:KNY-Ham6-sys} we use the following notation: 
coordinates $(q,p)$, parameters $a_{0},\ldots, a_{4}$ (\emph{root variables}),
base points $p_{i}$, exceptional divisors $E_{i}$.	
\end{notation*}

\begin{figure}[ht]
	\begin{tikzpicture}[>=stealth,basept/.style={circle, draw=red!100, fill=red!100, thick, inner sep=0pt,minimum size=1.2mm}]
	\begin{scope}[xshift=0cm,yshift=0cm]
	\draw [black, line width = 1pt] (-0.4,0) -- (2.9,0)	node [pos=0,left] {\small $H_{p}$} node [pos=1,right] {\small $p=0$};
	\draw [black, line width = 1pt] (-0.4,2.5) -- (2.9,2.5) node [pos=0,left] {\small $H_{p}$} node [pos=1,right] {\small $p=\infty$};
	\draw [black, line width = 1pt] (0,-0.4) -- (0,2.9) node [pos=0,below] {\small $H_{q}$} node [pos=1,above] {\small $q=0$};
	\draw [black, line width = 1pt] (2.5,-0.4) -- (2.5,2.9) node [pos=0,below] {\small $H_{q}$} node [pos=1,above] {\small $q=\infty$};
	\node (p1) at (2.5,0.9) [basept,label={[right] \small $p_{1}$}] {};
	\node (p2) at (2.5,1.5) [basept,label={[right] \small $p_{2}$}] {};
	\node (p3) at (1.5,2.5) [basept,label={[above] \small $p_{3}$}] {};
	\node (p4) at (2,2) [basept,label={[xshift = 0pt, yshift=-15pt] \small $p_{4}$}] {};
	\node (p5) at (0,0) [basept,label={[xshift = -8pt, yshift=-15pt] \small $p_{5}$}] {};
	\node (p6) at (0,0.6) [basept,label={[left] \small $p_{6}$}] {};
	\node (p7) at (0.5,2.5) [basept,label={[above] \small $p_{7}$}] {};
	\node (p8) at (1,2) [basept,label={[xshift = 0pt, yshift=-15pt] \small $p_{8}$}] {};
	\draw [red, line width = 0.8pt, ->] (p4) -- (p3);
	\draw [red, line width = 0.8pt, ->] (p8) -- (p7);
	\end{scope}
	\draw [->] (5.5,1)--(3.5,1) node[pos=0.5, below] {$\operatorname{Bl}_{p_{1}\cdots p_{8}}$};
	\begin{scope}[xshift=7.5cm,yshift=0cm]
	\draw [blue, line width = 1pt] (0,-0.3) -- (0,3) node [pos=1, above, xshift=-10pt] {\small $H_{q}-E_{5}-E_{6}$};
	\draw [blue, line width = 1pt] (-0.4,2.6) -- (5.2,2.6) node [pos=1,right] {\small $H_{p} - E_{3} - E_{7}$};
	\draw [blue, line width = 1pt] (4.8,-0.9) -- (4.8,3) node [pos=1, above, xshift=10pt] {\small $H_{q}-E_{1} - E_{3}$};	
	\draw [red, line width = 1pt] (0.2,-0.5) -- (5.2,-0.5)	node [pos=1, right] {\small $H_{p}-E_{5}$};
	\draw [red, line width = 1pt] (0.7,-0.8) -- (-0.3,0.2) node [pos=1,left] {\small $E_{5}$};
	\draw [red, line width = 1pt] (0.2,0) -- (5.2,0)	node [pos=1, right] {\small $H_{p}-E_{6}$};
	\draw [red, line width = 1pt] (0.7,-0.3) -- (-0.3,0.7) node [pos=1,left] {\small $E_{6}$};
	\draw [red, line width = 1pt] (-0.4,1) -- (4.6,1)	node [pos=0, left] {\small $H_{p}-E_{1}$};
	\draw [red, line width = 1pt] (4.1,0.7) -- (5.1,1.7) node [pos=1,right] {\small $E_{1}$};
	\draw [red, line width = 1pt] (-0.4,1.5) -- (4.6,1.5)	node [pos=0, left] {\small $H_{p}-E_{2}$};
	\draw [red, line width = 1pt] (4.1,1.2) -- (5.1,2.2) node [pos=1,right] {\small $E_{2}$};
	\draw [red, line width = 1pt] (0.9,-0.9) -- (0.9,2.3) node [pos=0,below] {\small $H_{q} - E_{7}$};
	\draw [red, line width = 1pt] (1,2.5) -- (1.5,2) node [pos=1,xshift=5pt,yshift=-2pt] {\small $E_{8}$};
	\draw [blue, line width = 1pt] (0.7,1.9) -- (1.7,2.9) node [pos=1,above] {\small $E_{7}-E_{8}$};
	\draw [red, line width = 1pt] (2.4,-0.9) -- (2.4,2.3) node [pos=0,below] {\small $H_{q} - E_{3}$};
	\draw [blue, line width = 1pt] (2.2,1.9) -- (3.2,2.9) node [pos=1,above] {\small $E_{3}-E_{4}$};
	\draw [red, line width = 1pt] (2.5,2.5) -- (3,2) node [pos=1,xshift=5pt,yshift=-2pt] {\small $E_{4}$};
	\end{scope}
	\end{tikzpicture}
	\caption{The standard realization of the $D_{4}^{(1)}$ Sakai surface}
	\label{fig:soic-KNY-P6}
\end{figure}
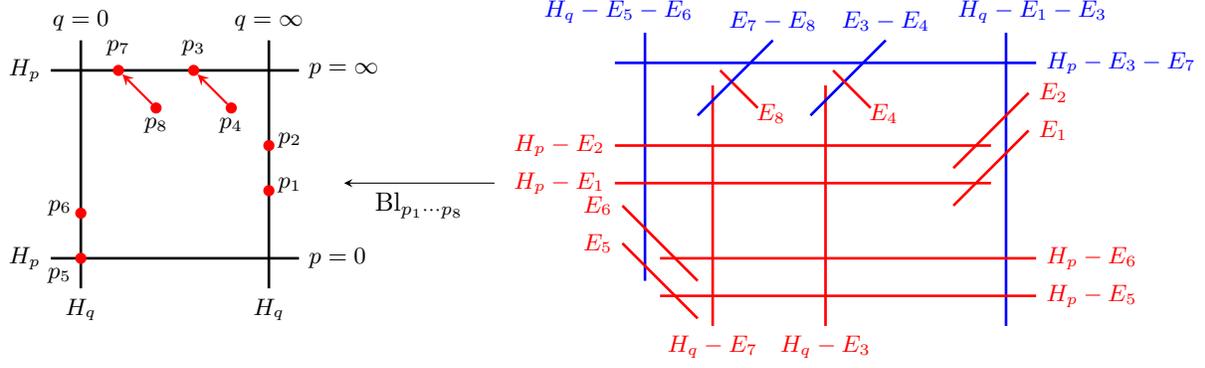

Its space of initial conditions is the standard geometric realization 
of the $D_{4}^{(1)}$ surface shown on Figure~\ref{fig:soic-KNY-P6}. 
The configuration of the base points is given
in terms of root variables $a_{i}$,
\begin{equation*}
	p_{1}(\infty,-a_{2}),\quad p_{2}(\infty,-a_{1}-a_{2}),\quad p_{3}(t,\infty)\leftarrow p_{4}(ta_{0},0),\quad 
	p_{5}(0,0),\quad p_{6}(0,a_{4}),\quad p_{7}(1,\infty)\leftarrow p_{8}(a_{3},0).
\end{equation*}

The surface and symmetry root bases for this standard realization of the $D_{4}^{(1)}$ surface are given on Fig.~\ref{fig:d-roots-d41-KNY}	
and Fig.~\ref{fig:a-roots-d41-KNY} respectively. The birational representation of the extended affine Weyl symmetry group 
$\widetilde{W}(D_{4}^{(1)})$ is given in \cite{KajNouYam:2017:GAPE} and we do not reproduce it here. 

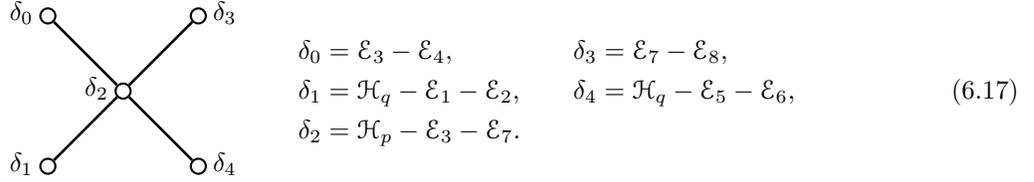
\begin{figure}[ht]
\begin{equation}\label{eq:d-roots-d41}			
	\raisebox{-32.1pt}{\begin{tikzpicture}[
			elt/.style={circle,draw=black!100,thick, inner sep=0pt,minimum size=2mm}]
		\path 	(-1,1) 	node 	(d0) [elt, label={[xshift=-10pt, yshift = -10 pt] $\delta_{0}$} ] {}
		        (-1,-1) node 	(d1) [elt, label={[xshift=-10pt, yshift = -10 pt] $\delta_{1}$} ] {}
		        ( 0,0) 	node  	(d2) [elt, label={[xshift=-10pt, yshift = -10 pt] $\delta_{2}$} ] {}
		        ( 1,1) 	node  	(d3) [elt, label={[xshift=10pt, yshift = -10 pt] $\delta_{3}$} ] {}
		        ( 1,-1) node 	(d4) [elt, label={[xshift=10pt, yshift = -10 pt] $\delta_{4}$} ] {};
		\draw [black,line width=1pt ] (d0) -- (d2) -- (d1)  (d3) -- (d2) -- (d4);
	\end{tikzpicture}} \qquad
			\begin{alignedat}{2}
			\delta_{0} &= \mathcal{E}_{3} - \mathcal{E}_{4}, &\qquad  \delta_{3} &= \mathcal{E}_{7} - \mathcal{E}_{8},\\
			\delta_{1} &= \mathcal{H}_{q} - \mathcal{E}_{1} - \mathcal{E}_{2}, &\qquad  \delta_{4} &= \mathcal{H}_{q} - \mathcal{E}_{5} - \mathcal{E}_{6},\\
			\delta_{2} &= \mathcal{H}_{p} - \mathcal{E}_{3} - \mathcal{E}_{7}.
			\end{alignedat}
\end{equation}
	\caption{The Surface Root Basis for the standard $D_{4}^{(1)}$ Sakai surface point configuration}
	\label{fig:d-roots-d41-KNY}	
\end{figure}
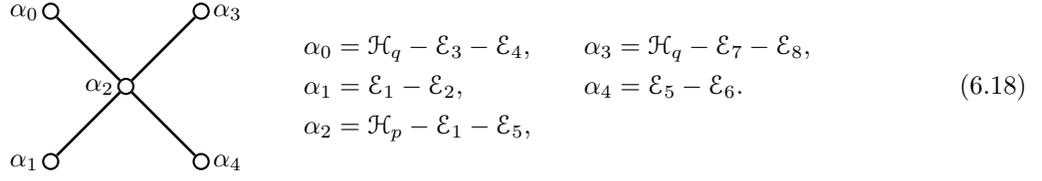
\begin{figure}[H]
\begin{equation}\label{eq:a-roots-d41}			
	\raisebox{-32.1pt}{\begin{tikzpicture}[
			elt/.style={circle,draw=black!100,thick, inner sep=0pt,minimum size=2mm}]
		\path 	(-1,1) 	node 	(d0) [elt, label={[xshift=-10pt, yshift = -10 pt] $\alpha_{0}$} ] {}
		        (-1,-1) node 	(d1) [elt, label={[xshift=-10pt, yshift = -10 pt] $\alpha_{1}$} ] {}
		        ( 0,0) 	node  	(d2) [elt, label={[xshift=-10pt, yshift = -10 pt] $\alpha_{2}$} ] {}
		        ( 1,1) 	node  	(d3) [elt, label={[xshift=10pt, yshift = -10 pt] $\alpha_{3}$} ] {}
		        ( 1,-1) node 	(d4) [elt, label={[xshift=10pt, yshift = -10 pt] $\alpha_{4}$} ] {};
		\draw [black,line width=1pt ] (d0) -- (d2) -- (d1)  (d3) -- (d2) -- (d4);
	\end{tikzpicture}} \qquad
			\begin{alignedat}{2}
			\alpha_{0} &=  \mathcal{H}_{q} -\mathcal{E}_{3} - \mathcal{E}_{4}, &\qquad  \alpha_{3} &= \mathcal{H}_{q} - \mathcal{E}_{7} - \mathcal{E}_{8},\\
			\alpha_{1} &= \mathcal{E}_{1} - \mathcal{E}_{2}, &\qquad  \alpha_{4} &= \mathcal{E}_{5} - \mathcal{E}_{6}.\\
			\alpha_{2} &= \mathcal{H}_{p} - \mathcal{E}_{1} - \mathcal{E}_{5},
			\end{alignedat}
\end{equation}
	\caption{The Symmetry Root Basis for the standard $D_{4}^{(1)}$ symmetry sub-lattice}
	\label{fig:a-roots-d41-KNY}	
\end{figure}


\subsection{The Okamoto Hamiltonian system} 
\label{sub:Ok-6}
\begin{notation*}
For the Okamoto system \eqref{eq:Ok-Ham6-sys} we use the following notation: coordinates $(f,g)$, parameters are the Okamoto parameters \eqref{eq:Ok-pars-PVI};
base points $q_{i}$, exceptional divisors $F_{i}$.	
\end{notation*}

The space of initial conditions for the system \eqref{eq:Ok-Ham6-sys} is given on Figure~\ref{fig:soic-Ok-P6}, where the basepoints are
\begin{equation}\label{eq:Ok-pts-P6}
	\begin{tikzpicture}
	\node (q9) at (-1.5,0) {$q_{9}(\infty,0)$}; 
	\node (q1) at (3,0.7) {$q_{1}\left(0,\frac{\theta + \kappa_{0} + \kappa_{1} - \kappa_{\infty} -1}{2}\right),$}; 
	\node (q2) at (3,-0.7) {$q_{2}\left(0,\frac{\theta + \kappa_{0} + \kappa_{1} + \kappa_{\infty}-1}{2}\right),$}; 
	\draw[->] (q1)--(0.5,0.7) --(0,0) --(q9);	\draw[->] (q2)--(0.5,-0.7) --(0,0) --(q9);	
	\node (q3) at (6.5,0.7) {$q_{3}(t,\infty)$};
	\node (q4) at (8.5,0.7) {$q_{4}(\theta,0),$};
	\draw[->] (q4)--(q3);
	\node (q5) at (6.5,0) {$q_{5}(0,\infty)$};
	\node (q6) at (8.5,0) {$q_{6}(\kappa_{0},0),$};
	\draw[->] (q6)--(q5);
	\node (q7) at (6.5,-0.7) {$q_{7}(1,\infty)$};
	\node (q8) at (8.5,-0.7) {$q_{8}(\kappa_{1},0).$};
	\draw[->] (q8)--(q7);
	\end{tikzpicture}
\end{equation}
Note that we have nine basepoints and two $-3$ curves, so this surface is not minimal and we need to blow down the $-1$-curve $H_{f} - F_{9}$. As mentioned before, the change of coordinates is $(q,p)=(f,fg)$, and the corresponding change of bases is given in the following Lemma.

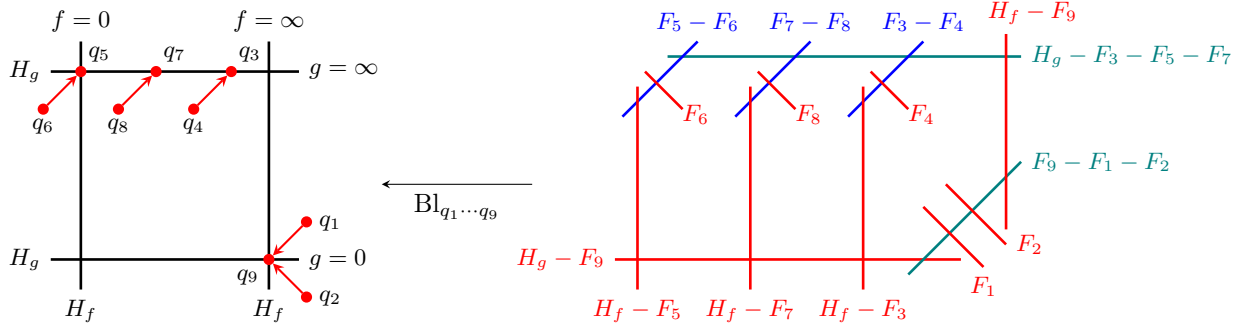
\begin{figure}[ht]
	\begin{tikzpicture}[>=stealth,basept/.style={circle, draw=red!100, fill=red!100, thick, inner sep=0pt,minimum size=1.2mm}]
	\begin{scope}[xshift=0cm,yshift=0cm]
	\draw [black, line width = 1pt] (-0.4,0) -- (2.9,0)	node [pos=0,left] {\small $H_{g}$} node [pos=1,right] {\small $g=0$};
	\draw [black, line width = 1pt] (-0.4,2.5) -- (2.9,2.5) node [pos=0,left] {\small $H_{g}$} node [pos=1,right] {\small $g=\infty$};
	\draw [black, line width = 1pt] (0,-0.4) -- (0,2.9) node [pos=0,below] {\small $H_{f}$} node [pos=1,above] {\small $f=0$};
	\draw [black, line width = 1pt] (2.5,-0.4) -- (2.5,2.9) node [pos=0,below] {\small $H_{f}$} node [pos=1,above] {\small $f=\infty$};
	\node (q9) at (2.5,0) [basept,label={[xshift = -8pt, yshift=-15pt] \small $q_{9}$}] {};
	\node (q1) at (3,0.5) [basept,label={[xshift = 9pt, yshift=-9pt] \small $q_{1}$}] {};
	\node (q2) at (3,-0.5) [basept,label={[xshift = 9pt, yshift=-9pt] \small $q_{2}$}] {};
	\node (q3) at (2,2.5) [basept,label={[xshift = 7pt, yshift=-2pt] \small $q_{3}$}] {};
	\node (q4) at (1.5,2) [basept,label={[xshift = 0pt, yshift=-15pt] \small $q_{4}$}] {};
	\node (q5) at (0,2.5) [basept,label={[xshift = 7pt, yshift=-2pt] \small $q_{5}$}] {};
	\node (q6) at (-0.5,2) [basept,label={[xshift = 0pt, yshift=-15pt] \small $q_{6}$}] {};
	\node (q7) at (1,2.5) [basept,label={[xshift = 7pt, yshift=-2pt] \small $q_{7}$}] {};
	\node (q8) at (0.5,2) [basept,label={[xshift = 0pt, yshift=-15pt] \small $q_{8}$}] {};
	\draw [red, line width = 0.8pt, ->] (q1) -- (q9);
	\draw [red, line width = 0.8pt, ->] (q2) -- (q9);
	\draw [red, line width = 0.8pt, ->] (q4) -- (q3);
	\draw [red, line width = 0.8pt, ->] (q6) -- (q5);
	\draw [red, line width = 0.8pt, ->] (q8) -- (q7);
	\end{scope}
	\draw [->] (6,1)--(4,1) node[pos=0.5, below] {$\operatorname{Bl}_{q_{1}\cdots q_{9}}$};
	\begin{scope}[xshift=7.5cm,yshift=0cm]
	\draw [red, line width = 1pt] (-0.4,0) -- (4.2,0)	node [pos=0, left] {\small $H_{g}-F_{9}$};
	\draw [teal, line width = 1pt] (0.3,2.7) -- (5,2.7) node [pos=1,right] {\small $H_{g} - F_{3} - F_{5} - F_{7}$};
	\draw [teal, line width = 1pt] (3.5,-0.2) -- (5,1.3) node [pos=1,right] {\small $F_{9} - F_{1} - F_{2}$};
	\draw [red, line width = 1pt] (4.8,0.4) -- (4.8,3) node [pos=1, above, xshift=10pt] {\small $H_{f}-F_{9}$};
	\draw [red, line width = 1pt] (3.7,0.7) -- (4.5,-0.1) node [pos=1, below] {\small $F_{1}$};
	\draw [red, line width = 1pt] (4,1) -- (4.8,0.2) node [pos=1, right] {\small $F_{2}$};
	\draw [blue, line width = 1pt] (-0.3,1.9) -- (0.7,2.9) node [pos=1,above] {\small $F_{5}-F_{6}$};
	\draw [red, line width = 1pt] (0,2.5) -- (0.5,2) node [pos=1,xshift=5pt,yshift=-2pt] {\small $F_{6}$};
	\draw [red, line width = 1pt] (-0.1,-0.4) -- (-0.1,2.3) node [pos=0, below] {\small $H_{f}-F_{5}$};
	\draw [blue, line width = 1pt] (1.2,1.9) -- (2.2,2.9) node [pos=1,above] {\small $F_{7}-F_{8}$};
	\draw [red, line width = 1pt] (1.5,2.5) -- (2,2) node [pos=1,xshift=5pt,yshift=-2pt] {\small $F_{8}$};
	\draw [red, line width = 1pt] (1.4,-0.4) -- (1.4,2.3) node [pos=0,below] {\small $H_{f} - F_{7}$};
	\draw [blue, line width = 1pt] (2.7,1.9) -- (3.7,2.9) node [pos=1,above] {\small $F_{3}-F_{4}$};
	\draw [red, line width = 1pt] (3,2.5) -- (3.5,2) node [pos=1,xshift=5pt,yshift=-2pt] {\small $F_{4}$};
	\draw [red, line width = 1pt] (2.9,-0.4) -- (2.9,2.3) node [pos=0,below] {\small $H_{f} - F_{3}$};
	\end{scope}
	\end{tikzpicture}
	\caption{The Space of Initial Conditions for the Okamoto Hamiltonian System}
	\label{fig:soic-Ok-P6}
\end{figure}

\begin{lemma}\label{lem:KNY-to-Ok-6} The change of bases for Picard lattices between the standard Kajiwara-Noumi-Yamada (with an additional blowup point) 
	and the Okamoto surfaces is given by 
	\begin{equation}\label{eq:basis-KNY-Ok-P6}
		\begin{aligned}
			\mathcal{H}_{q} & = \mathcal{H}_{f}, &\qquad 
				\mathcal{H}_{f} &= \mathcal{H}_{q},\\
			\mathcal{H}_{p} &= \mathcal{H}_{f} + \mathcal{H}_{g} - \mathcal{F}_{5} - \mathcal{F}_{9},  &\qquad 	
				\mathcal{H}_{g} & = \mathcal{H}_{q} + \mathcal{H}_{p} - \mathcal{E}_{5} - \mathcal{E}_{9}, \\
			\mathcal{E}_{i} &= \mathcal{H}_{f}-\mathcal{F}_{i}, &\qquad 
				\mathcal{F}_{i}	&= \mathcal{H}_{q} - \mathcal{E}_{i},\qquad\text{for } i=5,9,\\ 
			\mathcal{E}_{i} &= \mathcal{F}_{i}, &\qquad 
				\mathcal{F}_{i}	&= \ \mathcal{E}_{i},\qquad\text{otherwise}. 
		\end{aligned}
	\end{equation}
	This results in the following correspondences between the surface roots (with an additional blowup) is 
	\begin{equation}\label{eq:geom-KNY-Ok-P6}
	\begin{aligned}
		\delta_{0} &= \mathcal{E}_{3} - \mathcal{E}_{4} = \mathcal{F}_{3} - \mathcal{F}_{4}, &\qquad
			\delta_{3} &= \mathcal{E}_{7} - \mathcal{E}_{8} = 
				\mathcal{F}_{7} - \mathcal{F}_{8}, \\
		\delta_{1} &= \mathcal{H}_{q} - \mathcal{E}_{1} - \mathcal{E}_{2} - \mathcal{E}_{9} 
		= \mathcal{F}_{9} - \mathcal{F}_{1} - \mathcal{F}_{2}, &\qquad
			\delta_{4} &= \mathcal{H}_{q} - \mathcal{E}_{5} - \mathcal{E}_{6} = \mathcal{F}_{5} - \mathcal{F}_{6}; \\
		\delta_{2} &=  \mathcal{H}_{p} - \mathcal{E}_{3} - \mathcal{E}_{7} - \mathcal{E}_{9} = \mathcal{H}_{g} - 
		\mathcal{F}_{3} - \mathcal{F}_{5} - \mathcal{F}_{7}, &\qquad
	\end{aligned}
	\end{equation}
	and the symmetry roots,
	\begin{equation}\label{eq:sym-KNY-Ok-P6}
	\begin{aligned}
		\alpha_{0} &= \mathcal{H}_{q} - \mathcal{E}_{3} - \mathcal{E}_{4} = \mathcal{H}_{f} - \mathcal{F}_{3} - \mathcal{F}_{4}, &\qquad
		\alpha_{3} &= \mathcal{H}_{q} - \mathcal{E}_{7} - \mathcal{E}_{8} = 
				\mathcal{H}_{f} - \mathcal{F}_{7} - \mathcal{F}_{8}, \\
		\alpha_{1} &= \mathcal{E}_{1} - \mathcal{E}_{2}
		= \mathcal{F}_{1} - \mathcal{F}_{2}, &\qquad
		\alpha_{4} &= \mathcal{E}_{5} - \mathcal{E}_{6} = \mathcal{H}_{f} - \mathcal{F}_{5} - \mathcal{F}_{6}. \\
		\alpha_{2} &=  \mathcal{H}_{p} - \mathcal{E}_{1} - \mathcal{E}_{5} = \mathcal{H}_{g} - \mathcal{F}_{1} - \mathcal{F}_{9}, 
	\end{aligned}
	\end{equation}
\end{lemma}


\subsection{The Its-Prokhorov Hamiltonian system} 
\label{sub:IP-6}
\begin{notation*}
For the Its-Prokhorov system we use the following notation: coordinates $(x,y)$, parameters \eqref{eq:pars-IP2Ok-P6},
base points $w_{i}$, exceptional divisors $K_{i}$.	
\end{notation*}

The space of initial conditions for the system \eqref{eq:IP-Ham6-sys} is given on Figure~\ref{fig:soic-IP-6}, where the basepoints are
\begin{equation}\label{eq:IP-pts-6}
	\begin{tikzpicture}[baseline=0cm]
	\node (w9) at (0,0) {$w_{9}(\infty,0)$}; 
	\node (w1) at (2,0.7) {$w_{1}(0,\Theta_{\infty}-1),$}; 
	\node (w2) at (2,-0.7) {$w_{2}(0,-\Theta_{\infty}),$}; 
	\draw[->] (w1)--(1.3,0)--(w9);	\draw[->] (w2)--(1.3,0)--(w9); 
	\node (w3) at (3.5,0) {$w_{3}(t,\infty)$}; 
	\node (w4) at (5.5,0.7) {$w_{4}(\Theta_{t},0),$}; 
	\node (w10) at (5.5,-0.7) {$w_{10}(-\Theta_{t},0),$}; 
	\draw[->] (w4)--(4.8,0)--(w3);	\draw[->] (w10)--(4.8,0)--(w3); 
	\node (w5) at (7,0) {$w_{5}(0,\infty)$}; 
	\node (w6) at (9,0.7) {$w_{6}(\Theta_{0},0),$}; 
	\node (w11) at (9,-0.7) {$w_{11}(-\Theta_{0},0),$}; 
	\draw[->] (w6)--(8.3,0)--(w5);	\draw[->] (w11)--(8.3,0)--(w5); 
	\node (w7) at (10.5,0) {$w_{7}(1,\infty)$}; 
	\node (w8) at (12.5,0.7) {$w_{8}(\Theta_{1},0),$}; 
	\node (w12) at (12.5,-0.7) {$w_{12}(-\Theta_{1},0).$}; 
	\draw[->] (w8)--(11.8,0)--(w7);	\draw[->] (w12)--(11.8,0)--(w7); 
	\end{tikzpicture}
\end{equation}
Note that we have twelve basepoints and four $-3$ curves, so this surface is not minimal and we need to blow down some curves. 
From Figure~\ref{fig:soic-IP-6} it is pretty clear that we should blow down the $-1$- curves $H_{x} - K_{3}$, $H_{x}-K_{5}$, 
$H_{x}-K_{7}$ and $H_{x}-K_{9}$. As usual, instead we blow up the standard surface at additional points 
$p_{9},\ldots p_{12}$.

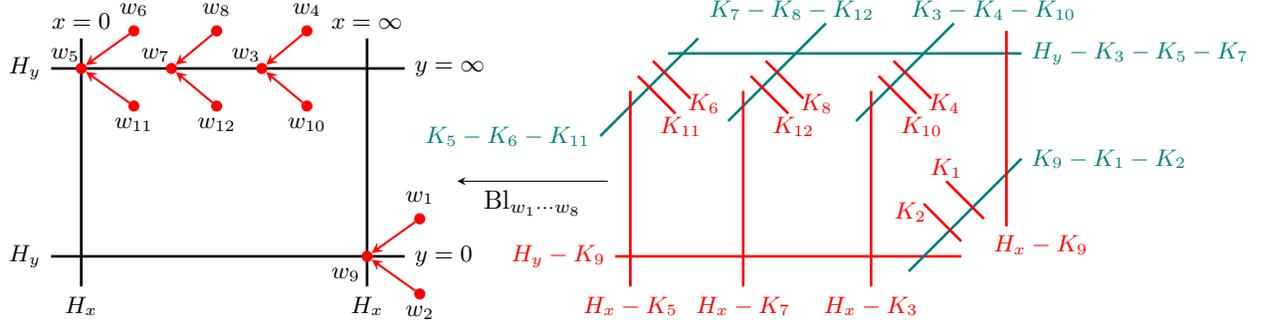
\begin{figure}[ht]
	\begin{tikzpicture}[>=stealth,basept/.style={circle, draw=red!100, fill=red!100, thick, inner sep=0pt,minimum size=1.2mm}]
	\begin{scope}[xshift=0cm,yshift=0cm]
	\draw [black, line width = 1pt] (-0.4,0) -- (4.3,0)	node [pos=0,left] {\small $H_{y}$} node [pos=1,right] {\small $y=0$};
	\draw [black, line width = 1pt] (-0.4,2.5) -- (4.3,2.5) node [pos=0,left] {\small $H_{y}$} node [pos=1,right] {\small $y=\infty$};
	\draw [black, line width = 1pt] (0,-0.4) -- (0,2.9) node [pos=0,below] {\small $H_{x}$} node [pos=1,above] {\small $x=0$};
	\draw [black, line width = 1pt] (3.8,-0.4) -- (3.8,2.9) node [pos=0,below] {\small $H_{x}$} node [pos=1,above] {\small $x=\infty$};
	\node (w9) at (3.8,0) [basept,label={[xshift = -8pt, yshift=-15pt] \small $w_{9}$}] {};
	\node (w1) at (4.5,0.5) [basept,label={[yshift=0pt] \small $w_{1}$}] {};
	\node (w2) at (4.5,-0.5) [basept,label={[yshift=-15pt] \small $w_{2}$}] {};
	\node (w5) at (0,2.5) [basept,label={[xshift = -6pt, yshift=-3pt] \small $w_{5}$}] {};
	\node (w6) at (.7,3) [basept,label={[above] \small $w_{6}$}] {};
	\node (w11) at (.7,2) [basept,label={[yshift=-15pt] \small $w_{11}$}] {};
	\node (w7) at (1.2,2.5) [basept,label={[xshift = -6pt, yshift=-3pt] \small $w_{7}$}] {};
	\node (w8) at (1.8,3) [basept,label={[above] \small $w_{8}$}] {};
	\node (w12) at (1.8,2) [basept,label={[yshift=-15pt] \small $w_{12}$}] {};
	\node (w3) at (2.4,2.5) [basept,label={[xshift = -6pt, yshift=-3pt] \small $w_{3}$}] {};
	\node (w4) at (3,3) [basept,label={[above] \small $w_{4}$}] {};
	\node (w10) at (3,2) [basept,label={[yshift=-15pt] \small $w_{10}$}] {};
	\draw [red, line width = 0.8pt, ->] (w1) -- (w9);
	\draw [red, line width = 0.8pt, ->] (w2) -- (w9);
	\draw [red, line width = 0.8pt, ->] (w6) -- (w5);	
	\draw [red, line width = 0.8pt, ->] (w11) -- (w5);	
	\draw [red, line width = 0.8pt, ->] (w8) -- (w7);	
	\draw [red, line width = 0.8pt, ->] (w12) -- (w7);		
	\draw [red, line width = 0.8pt, ->] (w4) -- (w3);	
	\draw [red, line width = 0.8pt, ->] (w10) -- (w3);		
	\end{scope}
	\draw [->] (7,1)--(5,1) node[pos=0.5, below] {$\operatorname{Bl}_{w_{1}\cdots w_{8}}$};
	\begin{scope}[xshift=7.5cm,yshift=0cm]
	\draw [red, line width = 1pt] (-0.4,0) -- (4.2,0)	node [pos=0, left] {\small $H_{y}-K_{9}$};
	\draw [teal, line width = 1pt] (0.3,2.7) -- (5,2.7) node [pos=1,right] {\small $H_{y} - K_{3} - K_{5} - K_{7}$};
	\draw [teal, line width = 1pt] (3.5,-0.2) -- (5,1.3) node [pos=1,right] {\small $K_{9} - K_{1} - K_{2}$};
	\draw [red, line width = 1pt] (4.8,0.4) -- (4.8,3) node [pos=0, xshift=13pt,yshift=-7pt] {\small $H_{x}-K_{9}$};
	\draw [red, line width = 1pt] (4,1) -- (4.5,0.5) node [pos=0, xshift=0pt,yshift=5pt] {\small $K_{1}$};
	\draw [red, line width = 1pt] (3.7,0.7) -- (4.2,0.2) node [pos=0, xshift=-5pt,yshift=-3pt] {\small $K_{2}$};

	\draw [teal, line width = 1pt] (-0.6,1.6) -- (0.7,2.9) node [pos=0,left] {\small $K_{5}-K_{6} - K_{11}$};
	\draw [red, line width = 1pt] (0.1,2.6) -- (0.6,2.1) node [pos=1,xshift=5pt,yshift=-2pt] {\small $K_{6}$};
	\draw [red, line width = 1pt] (-0.1,2.4) -- (0.4,1.9) node [pos=1,xshift=2pt,yshift=-5pt] {\small $K_{11}$};
	\draw [red, line width = 1pt] (-0.2,-0.4) -- (-0.2,2.2) node [pos=0, below] {\small $H_{x}-K_{5}$};

	\draw [teal, line width = 1pt] (1.1,1.8) -- (2.4,3.1) node [pos=1,xshift=-13pt,yshift=5pt] {\small $K_{7}-K_{8} - K_{12}$};
	\draw [red, line width = 1pt] (1.6,2.6) -- (2.1,2.1) node [pos=1,xshift=5pt,yshift=-2pt] {\small $K_{8}$};
	\draw [red, line width = 1pt] (1.4,2.4) -- (1.9,1.9) node [pos=1,xshift=2pt,yshift=-5pt] {\small $K_{12}$};
	\draw [red, line width = 1pt] (1.3,-0.4) -- (1.3,2.2) node [pos=0, below] {\small $H_{x}-K_{7}$};

	\draw [teal, line width = 1pt] (2.8,1.8) -- (4.1,3.1) node [pos=1,xshift=15pt,yshift=5pt] {\small $K_{3}-K_{4} - K_{10}$};
	\draw [red, line width = 1pt] (3.3,2.6) -- (3.8,2.1) node [pos=1,xshift=5pt,yshift=-2pt] {\small $K_{4}$};
	\draw [red, line width = 1pt] (3.1,2.4) -- (3.6,1.9) node [pos=1,xshift=2pt,yshift=-5pt] {\small $K_{10}$};
	\draw [red, line width = 1pt] (3,-0.4) -- (3,2.2) node [pos=0, below] {\small $H_{x}-K_{3}$};

	\end{scope}
	\end{tikzpicture}
	\caption{The Space of Initial Conditions for the Its-Prokhorov Hamiltonian System}
	\label{fig:soic-IP-6}
\end{figure}

\begin{lemma}\label{lem:FZ-KNY-Ok-6} The change of bases for Picard lattices between the standard Kajiwara-Noumi-Yamada (with four additional blowup points) 
	and the Its-Prokhorov surfaces is given by 
	\begin{equation}\label{eq:basis-KNY-IP-6}
		\begin{aligned}
			\mathcal{H}_{q} & = \mathcal{H}_{x}, &
			\mathcal{H}_{x} &= \mathcal{H}_{q},\\
			\mathcal{H}_{p} &= 3\mathcal{H}_{x} + \mathcal{H}_{y} - \mathcal{K}_{3} - \mathcal{K}_{5} - \mathcal{K}_{7} 
				- \mathcal{K}_{9}- \mathcal{K}_{10} - \mathcal{K}_{12},  & 	
			\mathcal{H}_{y} & = 3 \mathcal{H}_{q} + \mathcal{H}_{p} - \mathcal{E}_{3} - \mathcal{E}_{7} - \mathcal{E}_{9} - \mathcal{E}_{10}
			 - \mathcal{E}_{11}  - \mathcal{E}_{12}, \\
			\mathcal{E}_{1} &= \mathcal{K}_{1}, &
			\mathcal{K}_{1}	&= \mathcal{E}_{1},\\ 
			\mathcal{E}_{2} &= \mathcal{K}_{2}, &
			\mathcal{K}_{2}	&= \mathcal{E}_{2},\\ 
			\mathcal{E}_{3} &= \mathcal{H}_{x} -  \mathcal{K}_{10}, &
			\mathcal{K}_{3}	&= \mathcal{H}_{q} - \mathcal{E}_{10},\\ 
			\mathcal{E}_{4} &= \mathcal{K}_{4}, &
			\mathcal{K}_{4}	&= \mathcal{E}_{4},\\ 
			\mathcal{E}_{5} &= \mathcal{K}_{11}, &
			\mathcal{K}_{5}	&= \mathcal{H}_{q} - \mathcal{E}_{11},\\ 
			\mathcal{E}_{6} &= \mathcal{K}_{6}, &
			\mathcal{K}_{6}	&= \mathcal{E}_{6},\\ 
			\mathcal{E}_{7} &= \mathcal{H}_{x} - \mathcal{K}_{12}, &
			\mathcal{K}_{7}	&= \mathcal{H}_{q} - \mathcal{E}_{12},\\ 
			\mathcal{E}_{8} &= \mathcal{K}_{8}, &
			\mathcal{K}_{8}	&= \mathcal{E}_{8},\\ 
			\mathcal{E}_{9} &= \mathcal{H}_{x} - \mathcal{K}_{9}, &
			\mathcal{K}_{9}	&= \mathcal{H}_{q} - \mathcal{E}_{9},\\
			\mathcal{E}_{10} &= \mathcal{H}_{x} - \mathcal{K}_{3}, &
			\mathcal{K}_{10} &= \mathcal{H}_{q} - \mathcal{E}_{3},\\
			\mathcal{E}_{11} &= \mathcal{H}_{x} - \mathcal{K}_{5}, &
			\mathcal{K}_{11} &= \mathcal{E}_{5},\\
			\mathcal{E}_{12} &= \mathcal{H}_{x} - \mathcal{K}_{7}, &
			\mathcal{K}_{12} &= \mathcal{H}_{q} - \mathcal{E}_{7}.
		\end{aligned}
	\end{equation}
	This results in the following correspondences between the surface roots with additional blowups 
	(i.e., classes of curves with indices $-2$ and $-3$), 
	\begin{equation}\label{eq:geom-KNY-IP-6}
	\begin{aligned}
		\delta_{0} &= \mathcal{E}_{3} - \mathcal{E}_{4} - \mathcal{E}_{10} = \mathcal{K}_{3} - \mathcal{K}_{4} - \mathcal{K}_{10}, &\quad 
		\delta_{3} &= \mathcal{E}_{7} - \mathcal{E}_{8} - \mathcal{E}_{12} = 
				 \mathcal{K}_{7} - \mathcal{K}_{8} - \mathcal{K}_{12}, \\
		\delta_{1} &= \mathcal{H}_{q} - \mathcal{E}_{1} - \mathcal{E}_{2} - \mathcal{E}_{9} = \mathcal{K}_{9} - \mathcal{K}_{1} - \mathcal{K}_{2} &\quad
		\delta_{4} &= \mathcal{H}_{q} - \mathcal{E}_{5} - \mathcal{E}_{6} - \mathcal{E}_{11} = 
			\mathcal{K}_{5} - \mathcal{K}_{6} - \mathcal{K}_{11}, \\
		\delta_{2} &=  \mathcal{H}_{p} - \mathcal{E}_{3} - \mathcal{E}_{7} - \mathcal{E}_{9} = 
		\mathcal{H}_{y} - \mathcal{K}_{3} - \mathcal{K}_{5} - \mathcal{K}_{7}, &\quad 
	\end{aligned}
	\end{equation}
	and the symmetry roots,
	\begin{equation}\label{eq:sym-KNY-IP-6}
	\begin{aligned}
			\alpha_{0} &= \mathcal{H}_{q} - \mathcal{E}_{3} - \mathcal{E}_{4}
			= \mathcal{K}_{10} - \mathcal{K}_{4}, &\quad 
			\alpha_{3} &= \mathcal{H}_{q} - \mathcal{E}_{7} - \mathcal{E}_{8} = 
			\mathcal{K}_{12} - \mathcal{K}_{8},\\
			\alpha_{1} &= \mathcal{E}_{1} - \mathcal{E}_{2} =  \mathcal{K}_{1} - \mathcal{K}_{2}, &\quad 
			\alpha_{4} &= \mathcal{E}_{5} - \mathcal{E}_{6} = 
			\mathcal{K}_{11} - \mathcal{K}_{6},\\
			\alpha_{2} &= \rlap{$\displaystyle{\mathcal{H}_{p} - \mathcal{E}_{1} - \mathcal{E}_{5}
			= 3 \mathcal{H}_{x} + \mathcal{H}_{y} - \mathcal{K}_{1} - \mathcal{K}_{3} - \mathcal{K}_{5} 
				- \mathcal{K}_{7} - \mathcal{K}_{9} - \mathcal{K}_{10} - \mathcal{K}_{11} - \mathcal{K}_{12}}, $}\\
		\end{aligned}
	\end{equation}
	The symplectic form is the standard one, $\symp{IP} = dy\wedge dx$, and the root variables are 
\begin{equation}\label{eq:IP-root-P6}
	a_{0} = 2\Theta_{t},\quad a_{1} = 2\Theta_{\infty}-1,\quad a_{2} = 1 - \Theta_{0}-\Theta_{1}-\Theta_{t}-\Theta_{\infty},\quad 
	a_{3} = 2\Theta_{1},\quad a_{4} = 2\Theta_{0}.
\end{equation}
\end{lemma}

From this we immediately get the change of coordinates to the standard surface, and to the Okamoto case.

\begin{lemma}\label{lem:coords-IP-KNY-Ok-6} The change of coordinates and parameter matching between the Kajiwara-Noumi-Yamada 
	and the Its-Prokhorov Hamiltonian systems is given by
    \begin{equation}\label{eq:KNY-IP-6}
   	 \left\{\begin{aligned}
   	 	q(x,y,t)&=x,\\
   		p(x,y,t)&= \Theta_{0} + x\left(y + \frac{\Theta_{1}}{x-1} + \frac{\Theta_{t}}{x-t}\right),\\
   	 \end{aligned}\right.
    \qquad\text{and} \quad 
    	\left\{\begin{aligned}
   	 	x(q,p,t)&=q,\\
   		y(q,p,t)&= \frac{1}{2}\left(\frac{2p-a_{4}}{q} - \frac{a_{0}}{q-t} - \frac{a_{3}}{q-1}\right),
    	\end{aligned}\right.
    \end{equation}
	with the parameter correspondence given by \eqref{eq:IP-root-P6}.
	
	Combining it with the change of variables $q=f$, $p=fg$, we get the change of variables between the Okamoto
	and the Its-Prokhorov Hamiltonian systems,
    \begin{equation}\label{eq:IP-Ok-6}
   	 \left\{\begin{aligned}
   	 	x(f,g,t)&=f,\\
   		y(f,g,t)&= g - \frac{\theta}{2(f-t)}- \frac{\kappa_{0}}{2f} - \frac{\kappa_{1}}{2(f-1)},\\
		\Theta_{0}&=\frac{\kappa_{0}}{2},\quad \Theta_{1}= \frac{\kappa_{1}}{2},\\
		\Theta_{t}&=\frac{\theta}{2},\quad \Theta_{\infty} = \frac{1-\kappa_{\infty}}{2},
   	 \end{aligned}\right.
    \qquad\text{and} \qquad 
    	\left\{\begin{aligned}
   	 	f(x,y,t)&=x,\\
   		g(x,y,t)&= y + \frac{\Theta_{0}}{x} + \frac{\Theta_{1}}{x-1} + \frac{\Theta_{t}}{x-t},\\
		\kappa_{0}&=2\Theta_{0},\quad \kappa_{1} = 2 \Theta_{1},\\ 
		\theta&= 2\Theta_{t},\quad \kappa_{\infty} = 1 - 2\Theta_{\infty}.
    	\end{aligned}\right.
    \end{equation}
\end{lemma}

Since the change of variables is time-dependent, there will be a correction term in the Hamiltoninan (as well as some purely $t$-dependent terms),
\begin{equation*}
	\Ham{IP}{VI}(x,y;t) = 
	\Ham{Ok}{VI}\big(f(x,y,t),g(x,y,t);t\big)-\frac{\theta}{2(x-t)} + 
	\frac{\theta-1}{2}\left(\frac{\kappa_{1}}{t-1} + \frac{\kappa_{0}}{t}\right) -\frac{\theta}{2}\cdot\frac{2t-1}{t(t-1)}.
\end{equation*}


\subsection{The Filipuk-\.{Z}o\l\c{a}dek Hamiltonian system} 
\label{sub:FZ-6}

\begin{notation*}
For the Filipuk-\.{Z}o\l\c{a}dek system we use the following notation: coordinates $(x,y)$, parameters \eqref{eq:Ok-pars-PVI},
base points $z_{i}$, exceptional divisors $K_{i}$.	
\end{notation*}

The space of initial conditions for the system \eqref{eq:FZ-Ham6-sys} is the same as given on Figure~\ref{fig:soic-IP-6}
for the Its-Prokhorov case, but the coordinates of the basepoints are slightly different:
\begin{equation}\label{eq:FZ-pts-6}
	\begin{tikzpicture}[baseline=0cm]
	\node (z9) at (0,0) {$z_{9}(\infty,0)$}; 
	\node (z1) at (2,0.7) {$z_{1}\left(0,-\frac{\kappa_{\infty}}{2}\right),$}; 
	\node (z2) at (2,-0.7) {$z_{2}\left(0,\frac{\kappa_{\infty}}{2}\right),$}; 
	\draw[->] (z1)--(1.3,0)--(z9);	\draw[->] (z2)--(1.3,0)--(z9); 
	\node (z3) at (3.5,0) {$z_{3}(t,\infty)$}; 
	\node (z4) at (5.5,0.7) {$z_{4}\left(\frac{1+\theta}{2},0\right),$}; 
	\node (z10) at (5.5,-0.7) {$z_{10}\left(\frac{1-\theta}{2},0\right),$}; 
	\draw[->] (z4)--(4.8,0)--(z3);	\draw[->] (z10)--(4.8,0)--(z3); 
	\node (z5) at (7,0) {$z_{5}(0,\infty)$}; 
	\node (z6) at (9,0.7) {$z_{6}\left(\frac{\kappa_{0}}{2},0\right),$}; 
	\node (z11) at (9,-0.7) {$z_{11}\left(\frac{-\kappa_{0}}{2},0\right),$}; 
	\draw[->] (z6)--(8.3,0)--(z5);	\draw[->] (z11)--(8.3,0)--(z5); 
	\node (z7) at (10.5,0) {$z_{7}(1,\infty)$}; 
	\node (z8) at (12.5,0.7) {$z_{8}\left(\frac{\kappa_{1}}{2},0\right),$}; 
	\node (z12) at (12.5,-0.7) {$z_{12}\left(-\frac{\kappa_{0}}{2},0\right).$}; 
	\draw[->] (z8)--(11.8,0)--(z7);	\draw[->] (z12)--(11.8,0)--(z7); 
	\end{tikzpicture}
\end{equation}
The change of basis on the level of Picard lattices is the same as in Lemma~\ref{lem:FZ-KNY-Ok-6}, but the change of coordinates is slightly 
different and is given in the following Lemma.

\begin{lemma}\label{lem:coords-FZ-KNY-Ok-6} The change of coordinates and parameter matching between the Kajiwara-Noumi-Yamada 
	and the Filipuk-\.{Z}o\l\c{a}dek Hamiltonian systems is given by
    \begin{equation}\label{eq:KNY-FZ-6}
   	 \left\{\begin{aligned}
   	 	q(x,y,t)&=x,\\
   		p(x,y,t)&= x y + \frac{1}{2}\left(\kappa_{0} + \kappa_{1}\frac{x}{x-1} + (\theta-1)\frac{x}{x-t}\right),\\
   	 \end{aligned}\right.
\ \text{and}\ 
    	\left\{\begin{aligned}
   	 	x(q,p,t)&=q,\\
   		y(q,p,t)&= \frac{p}{q} - \frac{1}{2}\left(\frac{a_{0}-1}{q-t} + \frac{a_{3}}{q-1} + \frac{a_{4}}{q}\right),
    	\end{aligned}\right.
    \end{equation}
	with the parameter correspondence given by \eqref{eq:pars-root2Ok-P6}.
	
	Combining it with the change of variables $q=f$, $p=fg$, we get the change of variables between the Okamoto
	and the Filipuk-\.{Z}o\l\c{a}dek Hamiltonian systems,
    \begin{equation}\label{eq:FZ-Ok-6}
   	 \left\{\begin{aligned}
   	 	x(f,g,t)&=f,\\
   		y(f,g,t)&= g - \frac{1}{2}\left(\frac{\kappa_{0}}{f} + \frac{\kappa_{1}}{f-1} + \frac{\theta-1}{f-t}\right),\\
   	 \end{aligned}\right.
    \ \text{and} \  
    	\left\{\begin{aligned}
   	 	f(x,y,t)&=x,\\
   		g(x,y,t)&= y +\frac{1}{2}\left(\frac{\kappa_{0}}{x} + \frac{\kappa_{1}}{x-1} + \frac{\theta-1}{x-t}\right).\\
    	\end{aligned}\right.
    \end{equation}
\end{lemma}

The Hamiltonians are related by
\begin{align*}
	\Ham{FZ}{VI}(x,y;t) &= 
	\Ham{Ok}{VI}\big(f(x,y,t),g(x,y,t);t\big)-\frac{\theta-1}{2(x-t)} \\
	&\qquad  - \frac{\kappa_{\infty}^{2}}{4(t-1)} + 
	\frac{((2t-1)(\theta-1) - \kappa_{0} + \kappa_{1})(\theta-1 + \kappa_{0} + \kappa_{1})}{4t(t-1)}.
\end{align*}

\begin{remark}
	Note that the coordinates for the Its-Prokhorov and Filipuk-\.{Z}o\l\c{a}dek systems are related by
	\begin{equation*}
		x^{\mathrm{FZ}} = x^{\mathrm{IP}} = x,\qquad y^{\mathrm{FZ}} = y^{\mathrm{IP}} + \frac{1}{2(x-t)},
	\end{equation*}
	with the parameter correspondence the same as in \eqref{eq:IP-Ok-6}.
\end{remark}


\section{Conclusion} 
\label{sec:conclusion}
In this paper we showed that, similar to the discrete case considered earlier in \cite{DzhFilSto:2020:RCDOPWHWDPE}, 
the geometric approach to the theory of Painlev\'e equations is an effective tool in studying the 
\emph{Painlev\'e equivalence problem} and reducing Painlev\'e equations, when written in a Hamiltonian form, to 
some chosen reference Hamiltonian via a birational change of variables. We introduced an essentially 
algorithmic procedure on how to obtain this change of variables explicitly via the geometry identification 
between the corresponding spaces of initial conditions. The procedure is illustrated in detail for the systems 
related to $\Pain{IV}$, but is very general and can be used for other systems and scalar forms of Painlev\'e equations
as well. One interesting follow up question to consider is to see what additional insights, if any, can the 
geometric approach provide to the notion of the tau-functions corresponding to these Hamiltonian systems. 


\section*{Acknowledgements} 
\label{sec:acknowledgements}
AD acknowledges the support of the MIMUW grant and the  OPUS 2017/25/B/ST1/00931 grant to visit Warsaw in January 2020 where some of this work 
was done. AS was supported by a London Mathematical Society Early Career Fellowship and gratefully acknowledges the support of 
the London Mathematical Society. The authors thank Philip Boalch %
    for some helpful remarks and suggestions.


\bibliographystyle{amsalpha}

\providecommand{\bysame}{\leavevmode\hbox to3em{\hrulefill}\thinspace}
\providecommand{\MR}{\relax\ifhmode\unskip\space\fi MR }
\providecommand{\MRhref}[2]{%
  \href{http://www.ams.org/mathscinet-getitem?mr=#1}{#2}
}
\providecommand{\href}[2]{#2}

\end{document}